\documentclass[aps,notitlepage,twocolumn,nofootinbib,pra,10pt]{revtex4-1}
\usepackage{amsmath,amsfonts,amssymb,amsthm,bbm,graphicx,enumerate,times}
\usepackage{mathtools}
\usepackage[usenames,dvipsnames]{color}

\usepackage{hyperref}
\usepackage{tikz}
\usepackage{graphicx}
\usepackage{float}
\usetikzlibrary{shapes,positioning}
\usepackage{etoolbox}

\definecolor{jens}{rgb}{0,.8,.5}

\definecolor{mgcolor}{rgb}{8,.4,.4}

\definecolor{ajcolor}{rgb}{.3,.2,.9}

\usepackage{dsfont}
\newcommand{\id}{\mathds{1}}

 \renewcommand{\i}{\,\ensuremath\mathrm{i}}

\newtheorem{theorem}{Theorem}
\newtheorem{lemma}[theorem]{Lemma}

\newtheorem{definition}[theorem]{Definition}

\newcommand{\spl}{\,{+}\,}

\newcommand{\seq}{\,{=}\,}
\newcommand{\sneq}{\,{\neq}\,}
\newcommand{\sle}{\,{<}\,}
\newcommand{\sleq}{\,{\leq}\,}
\newcommand{\sgr}{\,{>}\,}

\newcommand{\stimes}{\,{\times}\,}
\newcommand{\sto}{\,{\to}\,}
\newcommand{\sapprox}{\,{\approx}\,}
\newcommand{\selem}{\,{\in}\,}
\newcommand{\sdef}{\,{\coloneqq}\,}

\newcommand{\ket}[1]{\left.\left|{#1}\right.\right\rangle}

\newcommand{\bra}[1]{\left.\left\langle{#1}\right.\right|}

\newcommand{\braket}[2]{\left\langle #1 \middle| #2 \right\rangle}

\newcommand\vacket{{\ket{\emptyset}}}

\DeclareMathAlphabet{\mathpzcc}{OT1}{pzc}{m}{it}
\DeclareMathAlphabet{\mathpzc}{T1}{pzc}{m}{it}{\huge}
\newcommand{\msp}{\phantom{-}}
\newcommand{\m}{\operatorname{\gamma}}
\newcommand{\fe}{\operatorname{f}^{\phantom{\dagger}}}
\newcommand{\ft}{\tilde{\operatorname{f}}^{\phantom{\dagger}}}
\newcommand{\fd}{\operatorname{f}^\dagger}
\newcommand{\ftd}{\tilde{\operatorname{f}}^\dagger}

\newcommand\Par[1][]{%
  \ifstrempty{#1}{%
    \mathcal P_\text{tot}
  }{
    \mathcal P_{#1}
  }
}

\newcommand\Pf[1]{\mathrm{Pf}\left(#1\right)}

\def\f{{\operatorname{f}^{(full)}}}
\def\nUG{{ :U_G:}}

\def\and{\quad\text{and}\quad}

\def\xor{\ \texttt{XOR}\ }

\newcommand\BC[1]{ {\{0,1\}^{\times #1}}}
\def\BCL{ \{0,1\}^{\times L}}
\def\BCR{ \{0,1\}^{\times r}}
\def\d{\text{d}}

\begin{document}
\title{Holography and criticality in matchgate tensor networks}
\author{A.\ Jahn, M.\ Gluza, F.\ Pastawski, J.\ Eisert}
\affiliation{Dahlem Center for Complex Quantum Systems, Freie Universit{\"a}t Berlin, 14195 Berlin, Germany}

\date{\today}

\begin{abstract}
The AdS/CFT correspondence 
conjectures a holographic duality between gravity in a bulk space and a critical quantum field theory on its boundary.
Tensor networks have come to provide toy models to understand such bulk-boundary correspondences, shedding light on connections between geometry and entanglement. 
We introduce a versatile and efficient framework for studying tensor networks, extending previous tools for Gaussian matchgate tensors in $1{+}1$ dimensions. Using regular bulk tilings, we show that the critical Ising theory can be realized on the boundary of both flat and hyperbolic bulk lattices, obtaining highly accurate critical data. Within our framework, we also produce translation-invariant critical states by an efficiently contractible tensor network with the geometry of the multi-scale entanglement renormalization ansatz. Furthermore, we establish a link between holographic quantum error correcting codes and tensor networks. This work is expected to stimulate a more comprehensive study of tensor-network models capturing bulk-boundary correspondences.
\end{abstract}
\maketitle

The notion of holography in the context of bulk-boundary dualities, most famously expressed through the \emph{AdS/CFT correspondence} \cite{Maldacena98}, has had an enormously stimulating effect on recent developments in theoretical physics. 
A key feature of these dualities is the relationship between bulk geometry and boundary entanglement entropies \cite{VanRaamsdonk2010, Pastawski2016a, AreaReview}, prominently elucidated by the Ryu-Takayanagi formula \cite{PhysRevLett.96.181602}. 
Due to the importance of entanglement in the context of AdS/CFT \cite{PhysRevD.86.065007}, it was quickly realized that tensor networks are ideally suited for constructing holographic toy models, prominently the \emph{multiscale entanglement renormalization ansatz} (MERA) \cite{PhysRevLett.101.110501,MERAAlgorithms,PhysRevLett.100.130501}. 
The realization that \emph{quantum error correction} could be realized by a holographic duality \cite{Almheiri15} further connected to ideas from quantum information theory.
Despite the successful construction of several tensor network models that reproduce various aspects of AdS/CFT (see e.g.\ Refs.~\cite{Qi2013,Pastawski2015, Hayden2016}), a general understanding of the features and limits of tensor network holography is still lacking. 
Particular obstacles are the potentially large parameter spaces of tensor networks as well as the considerable computational cost of contraction.

In this work, we overcome some of these challenges by applying highly efficient contraction techniques developed for \emph{matchgate} tensors \cite{Valiant2002, Bravyi2008}, 
which replace tensor contraction by a Grassmann-variate integration scheme. 
These techniques allow us to comprehensively study the interplay of geometry and correlations in Gaussian fermionic tensor networks
in a versatile fashion, incorporating toy models for quantum error correction and tensor network approaches for conformal field theory (CFT), such as the MERA, into a single framework, highlighting the connections between them.
Furthermore, this framework includes highly symmetrical tensor networks based on \textsl{regular tilings} (see Fig.\ \ref{FIG_TRI_CURVATURE}). We are thus in a position to efficiently probe the full space of Gaussian bulk-boundary correspondences from a small set of parameters, including the bulk curvature. 
We show that matchgate tensor networks with a variety of bulk geometries contain the Ising CFT in their parameter space
to remarkably good approximation as a special case, with 
properties similar to the wavelet MERA model \cite{PhysRevLett.116.140403,RigorousWavelet}. While regular hyperbolic tilings have recently been considered as a MERA alternative \cite{HyperInvariantTensors}, 
we show that flat tilings can lead to very similar boundary states.
In our studies, we restrict ourselves to tensor networks that are non-unitary and real, resembling a \textsl{Euclidean} evolution from bulk to boundary. In particular, we do not require the causal constraints of the MERA for efficient contraction, thus providing new approaches in the context of tensor network renormalization \cite{Evenbly2015TNR1,Evenbly2015TNR2}. 
While we provide significant evidence that tensor networks are capable of describing bulk/boundary correspondences beyond known models and introduce a framework for their study, our work is by no means exhaustive. We do hope to provide a starting point for more systematic studies of holography in tensor networks.
\begin{figure}
\centering
\includegraphics[width=0.22\textwidth]{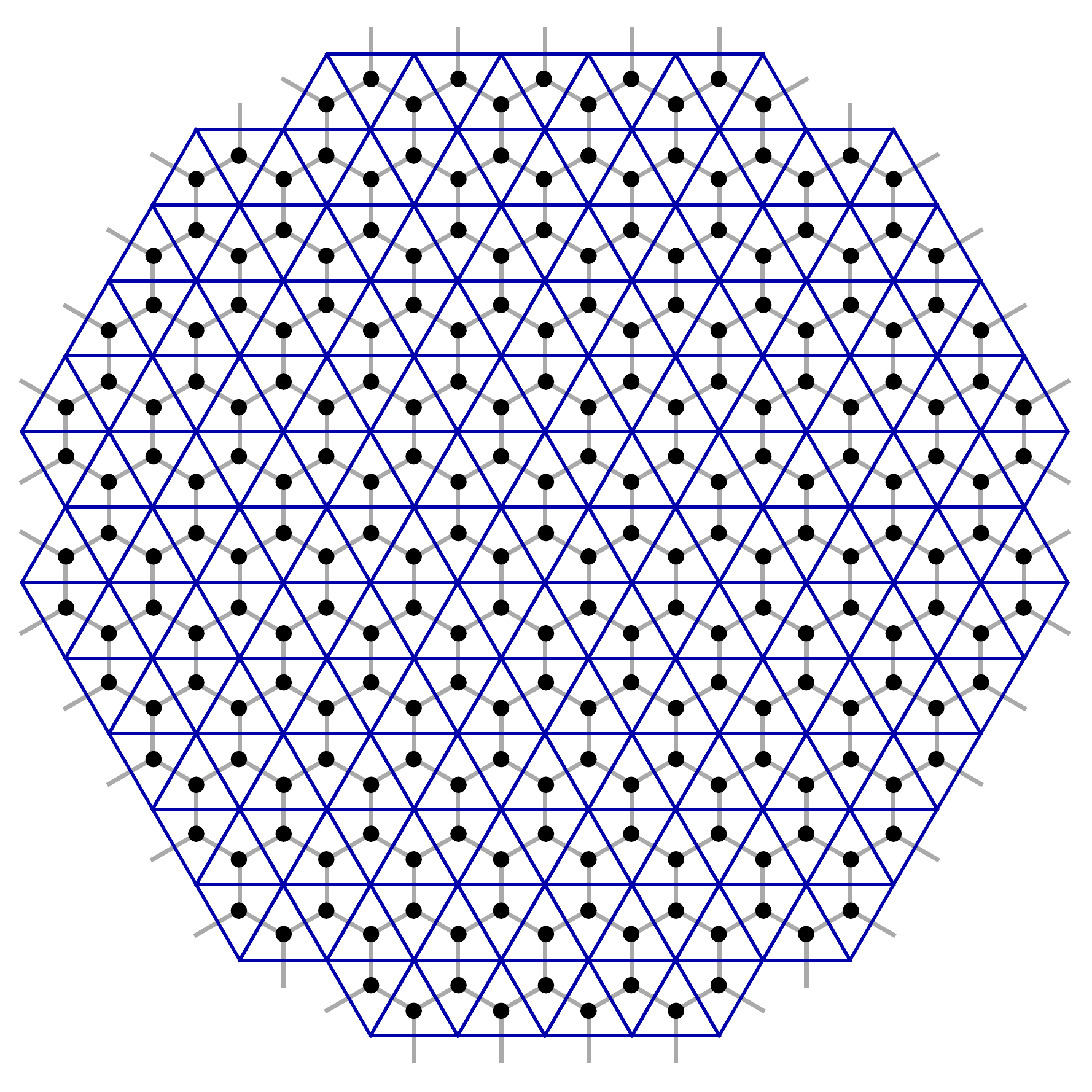}
\includegraphics[width=0.22\textwidth]{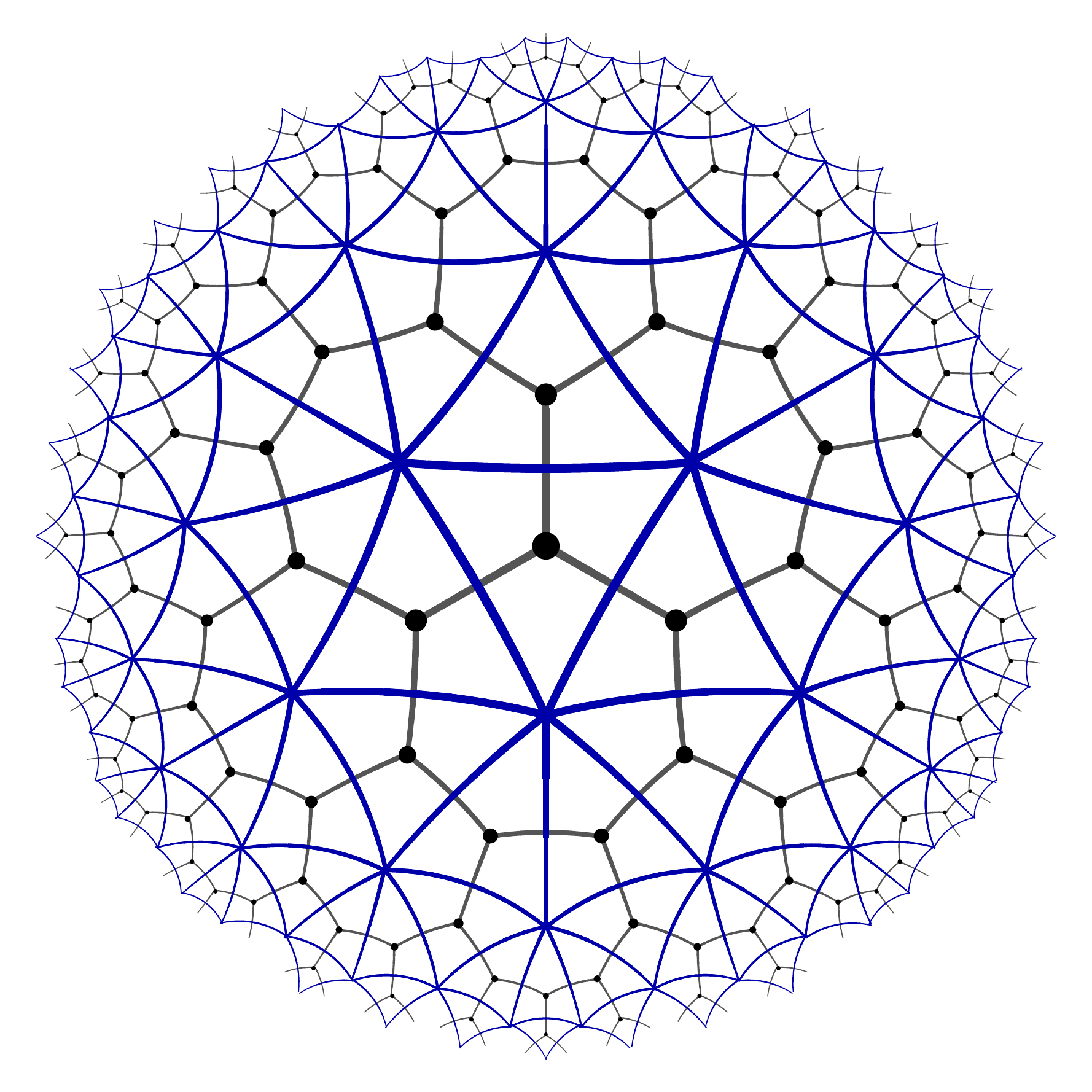}

\caption{Triangular tilings of flat and hyperbolic space (blue edges) and the corresponding tensor network (black lattice). In the matchgate formalism, joint edges between triangles correspond to an integration over a pair of Grassmann numbers, analogous to tensor network contraction over indices.
}
\label{FIG_TRI_CURVATURE}
\end{figure}
\paragraph*{Setting.}
We construct two-dimensional planar tensor networks with fermionic \emph{bulk} and \emph{boundary} degrees of freedom. The bulk degrees of freedom are associated with a set $V$ of vertices of a tensor network.  At each vertex $v \selem V$, a local tensor $T_v$ with $k_v$ indices is placed, which can be interpreted as a local fermionic state on $k_v$ sites. After contraction over all connected bulk indices, the $L$ remaining open indices are interpreted as boundary sites, with the boundary state specified by the full contracted tensor. Due to the planarity of the network, the boundary sites form a loop. The bulk geometry can be flat or negatively curved (a positively curved network closes in on itself after finite distance). We visualize our tensor networks by representing each tensor $T_v$ as a $k_v$-gon whose edges correspond to indices.
Thus, the tensor network is represented by a polygon tiling which determines the bulk geometry.
Adjacent edges between two polygons correspond to contracted indices and boundary edges to open ones. 
See Fig.\ \ref{FIG_TRI_CURVATURE} for examples.

Concretely, each bulk degree of freedom $v \selem V$ is associated with
 a local tensor $T_v: \{0,1\}^{\times r} \sto \mathbb C$ of tensor rank $r$ 
 (equal to the number of edges of the corresponding tile), which are contracted to form tensors of higher rank.
We denote the tensor component at indices $j \selem \BCR$ as $T_v(j)$ and the standard computational basis for $r$ boundary spins as $\ket {j} \sdef \otimes_{k=1}^r \ket{j_k}$. Each tensor is then equivalent to a state
\begin{align}
  \label{eq:psiT}
  \ket {\psi_v} =\sum_{j\in \BCR} T_v(j) \ket  {j} \text{ .}
\end{align}
For a broader introduction to tensor networks and their contractions, 
see Refs.\ \cite{Orus-AnnPhys-2014,SchuchReview,VerstraeteBig,EisertTensors}.

\paragraph*{Matchgate tensors.}
Instead of explicit tensor contraction along pairs of indices, we use the formalism from Ref.\ \cite{Bravyi2008} employing Grassmann integration. Any tensor $T$ can be represented by a Grassmann-variate \emph{characteristic function} 
\begin{align}
\label{eq:char}
\Phi_T(\theta) = \sum_{j\in \BC r} T(j) \theta_1^{j_1}\theta_2^{j_2}\ldots\theta_r^{j_r} \text{ ,}
\end{align}
where the $\theta_k$ are Grassmann numbers defined by the anti-commutation relation $\theta_k \theta_{k'}+\theta_{k'}\theta_k \seq 0$.
The contraction $T_{1\star 2}$ of two tensors $T_1$ and $T_2$ (of rank $r_1$ and $r_2$, respectively) over the last index of $T_1$ and the first index of $T_2$ is given by
\begin{equation}
T_{1\star 2}(x,y)=\sum_{z \in \lbrace 0,1 \rbrace} T_1(x,z)T_2(z,y) \text{ ,}
\end{equation}
where $x\selem\BC {(r_1-1)},y\selem \BC {(r_2-1)}$. Clearly, $T_{1\star 2}$ has rank $r_1+r_2-2$. The characteristic function of the contraction is obtained as
  \begin{align}
  \label{EQ_CONTR_INTEGR}
    \Phi_{T_{1\star 2}}(\xi) =\int {\rm d}\eta_1 \int {\rm d}\theta_{r_1} 
    \; \Phi_{T_1}(\theta) \Phi_{T_2}(\eta)\exp( \theta_{r_1}\eta_1) \text{ ,}
  \end{align}
 where we have used $\xi \seq ( \theta_1,\ldots,\theta_{r_1-1},\eta_2,\ldots,\eta_{r_2})$, and $\int {\rm d}\eta_1 \int {\rm d}\theta_{r_1}$ 
 denotes Grassmann integrals, anti-commuting multilinear functionals obeying $\int {\rm d}\xi_j \,\xi_j^{z_j} \seq  \delta_{z_j,1} $ (see Refs.\ \cite{Berezin, cahill1999density,Bravyi2008, bravyi2004lagrangian} for more details). 
A self-contained derivation of the equivalence of \eqref{EQ_CONTR_INTEGR} with tensor contraction, as well as a note on iterated integrals, is given in Appendix \ref{sec:app_int}. Anti-commutativity requires an appropriate labeling of all Grassmann variables, but such a labeling can always be found for contractions of planar networks \cite{Bravyi2008}.
Such Grassmann integrations are particularly efficient to compute for the case of \emph{matchgate} tensors, where their computation scales polynomially in the number of tensor indices.

Consider a rank-$r$ tensor $T(x)$ with inputs $x\selem \BC r$. We call $T(x)$ a  \emph{matchgate} if there exist an antisymmetric matrix $A \selem \mathbb C^{r\times r}$ and a $z\selem \BC r $ so that we can write
\begin{equation}
T(x) =\Pf{A_{|x \xor z}}T(z) \text{ ,}
\end{equation}
where $\Pf A$ is the Pfaffian of $A$, and $A_{|x}$ is the principal submatrix of $A$ acting on the subspace supported by $x$. Furthermore, we call $T(x)$ an even tensor if $T(x) \seq 0$ for any $x$ with odd $\sum_j x_j$.

A generic even matchgate has a simple Gaussian characteristic function of the form 
\begin{equation}
\Phi_T(\theta)=T(0)\exp\biggl(\frac12\sum_{j,k=1}^{r} A_{j,k} \theta_j \theta_k\biggr)\ \text{ ,}
\label{eq:genericChar}
\end{equation}
where $T(0)$, the tensor component for all-zero input, acts as a normalization factor. Apart from normalization, the full tensor is completely determined by $A$, which we therefore call the \emph{generating matrix}. 
Thus, the rules for contracting matchgate tensors can be written as rules for combining generating matrices. 
Full derivations of these, including the calculation of physical covariance matrices from the generating matrices, are provided in Appendices \ref{sec:app_conversion} and \ref{sec:app_contr}.
With our contraction rules, the computational cost of contracting two tensors is quadratic in the number of indices of the final tensor. Thus, we can bound the total computational cost for contracting an entire network of the type considered here by 
$O (L^2 N)$, where $L$ is the number of boundary sites and $N$ is the number of contracted tensors (for similar bounds on matchgate contraction, see Ref.\ \cite{Bravyi2008}).

Using Pauli matrices $\sigma^\alpha$ with $\alpha\selem\left\lbrace x,y,z \right\rbrace$, we can define Majorana operators $\m_i$ via the Jordan-Wigner transformation
\begin{align}
\label{EQ_JW_TRANS}
{\m}_{2k-1} &=
   ({\sigma^z})^{{\otimes}(k-1)}{\otimes}\,\sigma^x\,{\otimes}\left( \id_2 \right)^{\otimes(r-k)} \text{ ,}\\
{\m}_{2k}  &=
   ({\sigma^z})^{{\otimes}(k-1)}{\otimes}\,\sigma^y\,{\otimes}\left( \id_2 \right)^{\otimes(r-k)} \text{ .}
\end{align} 
The computational basis is then equivalent to an occupational basis.
In this context, we prove in Appendix \ref{sec:app_defs} that any fermionic Gaussian state vector in the form \eqref{eq:psiT} has coefficients $T(j)$ constituting a matchgate tensor. The converse is also true, providing a further perspective on the connection to free fermions \cite{terhal2002classical}.

\paragraph*{The holographic pentagon code.}

We will now apply our framework to the highly symmetric class of regular bulk tilings, first implementing the \emph{holographic error correcting code} (HaPPY code) proposed in Ref.\ \cite{Pastawski2015} and then exploiting the versatility of our framework to extend it towards more physical setups.
The HaPPY code furnishes a mapping between additional (uncontracted) bulk degrees of freedom on each tensor and the boundary state, realized by a bulk tiling of pentagons. Each pentagon tile encodes one fault-tolerant logical qubit via the encoding isometry of the \emph{five qubit code}.
This $[[ 5, 1, 3]]$ quantum error-correcting code \cite{Laflamme1996} saturates both the \emph{quantum Hamming bound} \cite{PhysRevA.54.1862,RevModPhys.87.307} as well as the \emph{singleton bound} \cite{RevModPhys.87.307} and can be expressed as a \emph{stabilizer code} \cite{Gottesman1997}. 

We observe that fixing the bulk degrees of freedom to computational basis states gives rise to a matchgate tensor network, as the logical computational basis states of the holographic pentagon code can be viewed as ground states of a quadratic fermionic Hamiltonian. This can be seen directly by applying \eqref{EQ_JW_TRANS} onto the \emph{stabilizers} $S_k$ of the underlying $[[5,1,3]]$ code, thus expressing it in terms of Majorana operators $\m_i$ and a total parity operator $\Par = (\sigma^z)^{\otimes 5}$:
\begin{equation}
\begin{aligned}
  S_1& = \sigma^x\otimes \sigma^z\otimes \sigma^z\otimes \sigma^x\otimes \id_2 = \i \m_7 \m_2 \text{ ,}\\
  S_2&=\id_2\otimes \sigma^x\otimes \sigma^z\otimes \sigma^z\otimes \sigma^x = \i \m_9 \m_4 \text{ ,}\\
  S_3&=\sigma^x\otimes \id_2\otimes \sigma^x\otimes  \sigma^z\otimes \sigma^z=\i\,\Par \m_6 \m_1 \text{ ,}\\
  S_4&=\sigma^z\otimes \sigma^x\otimes \id_2\otimes \sigma^x\otimes \sigma^z =\i\,\Par \m_8 \m_3 \text{ ,}\\
  S_5&=\sigma^z\otimes \sigma^z\otimes \sigma^x\otimes \id_2\otimes \sigma^x =\i\,\Par \m_{10} \m_{5} \text{ .}
\end{aligned}
\end{equation}
As the corresponding stabilizer Hamiltonian is given by $H=-\sum_{k=1}^5 S_k$, we find a doubly degenerate ground state whose degeneracy is lifted by the parity operator $\Par$. The resulting two states with parity eigenvalues $\pm 1$ correspond to the logical eigenstates $\bar{0}$ and $\bar{1}$, which are themselves ground states of purely quadratic Hamiltonians with different parity factors. 
Thus both computational basis states are pure Gaussian, leading us to the conclusion that for fixed computational input in the bulk, the holographic pentagon code yields a matchgate tensor on the boundary (see Fig.\ \ref{FIG_PENTAGON_CODE}). The explicit construction is given in Appendix \ref{ssec:happy}.
Using the Schl\"afli symbol $\lbrace p, q \rbrace$, where $p$ is the number of edges per polygon and $q$ the number of polygons around each corner, we can specify the hyperbolic geometry of the HaPPY model as a regular $\{5,4\}$ tiling.

\begin{figure}
\centering
\begin{equation*}
\begin{gathered}
\includegraphics[width=0.2\textwidth]{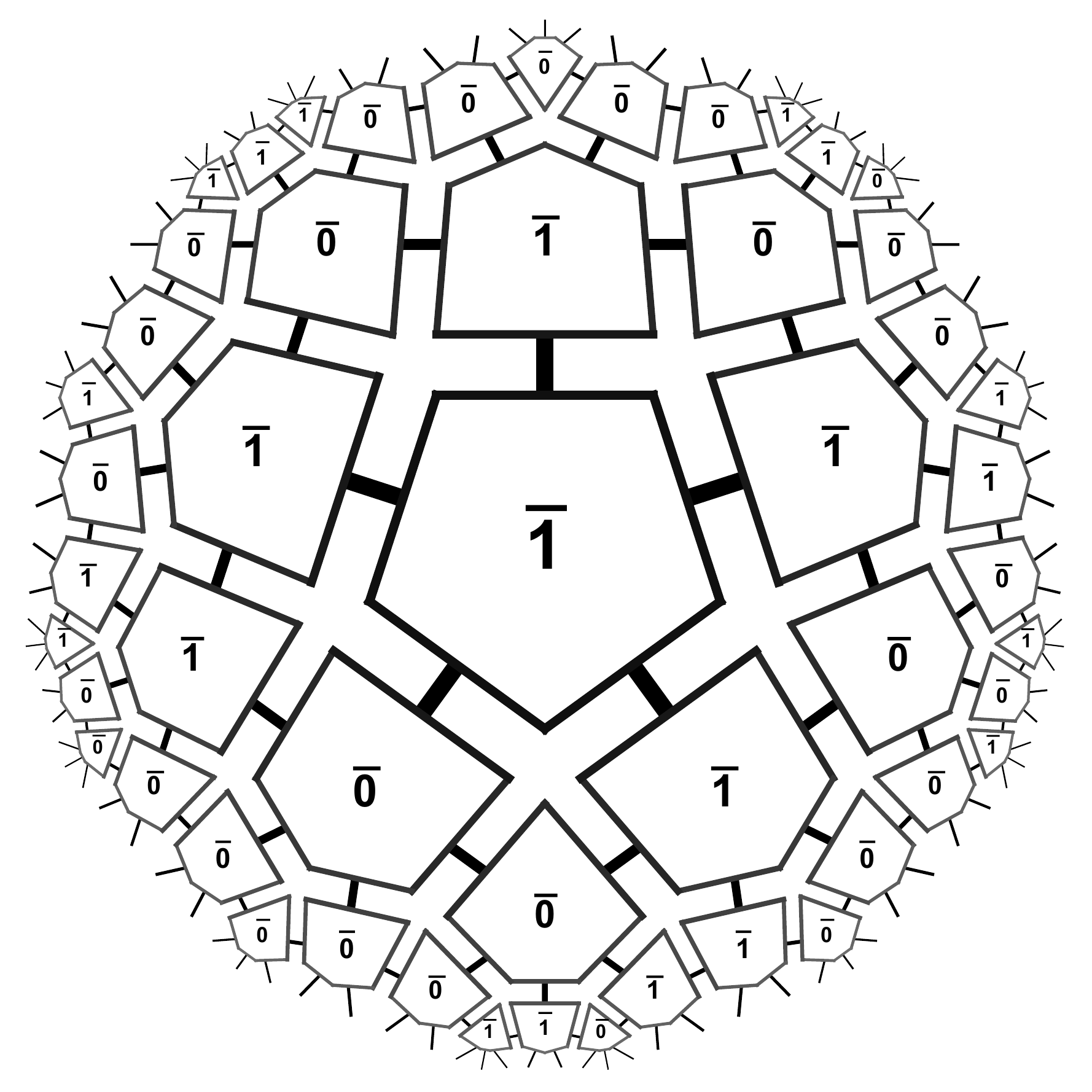}
\end{gathered}
\scalebox{1.5}{$\;=\;$}
\begin{gathered}
\includegraphics[width=0.2\textwidth]{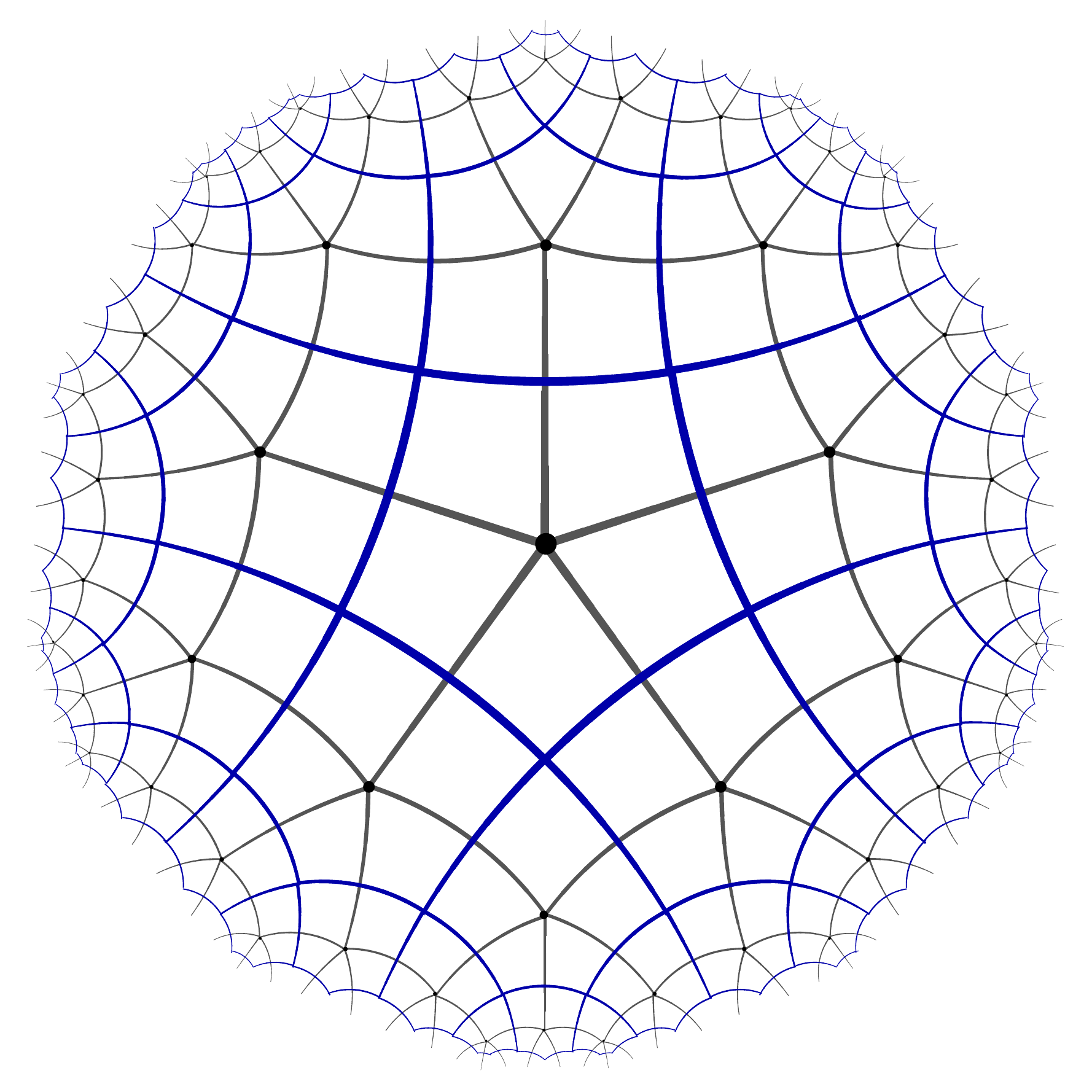}
\end{gathered}
\end{equation*}

\caption{The HaPPY code for fixed bulk inputs (left) is equal to a matchgate tensor network on a hyperbolic pentagon tiling (right).
}
\label{FIG_PENTAGON_CODE}
\end{figure}

We find that the correlation structure of this model is best captured in the Majorana picture.
Explicitly, consider the pentagon tiling of Ref.\ \cite{Pastawski2015} with all bulk inputs set to the positive-parity eigenvector $| \bar{0} \rangle$.
The entries of the Majorana covariance matrix $\Gamma_{i ,j} = \frac{\i\,}{2}\langle \psi | \left[ \m_i, \m_j \right] | \psi \rangle$ resulting from successive contraction steps are shown in Fig.\ \ref{FIG_COVMM_PENTA}. 
As we can see, both the individual pentagon state and the larger contracted states are characterized by a non-local pairing of Majorana fermions. 
Intriguingly, the contractions effectively connect Majorana pairs from each pentagon to a larger chain, so the pairs on the boundary of the contracted network can be seen as endpoints of a discretized ``geodesic'' spanning the bulk. 
While this discontinuous correlation pattern of $\Gamma_{i,j}$ makes the computation of CFT observables difficult, we can estimate the average correlation falloff by counting the relative frequency $n(d)$ of Majorana pairs at distance $d=|i-j|$ over which they connect points on the boundary. According to the results shown in Fig.\ \ref{FIG_FALLOFF_PENTA} (left), correlation falloff follows a power law $n(d) \propto d^{-1}$, as expected of a CFT. Furthermore, we compute the entanglement entropy $S_A$ of a subsystem $A$ of size $\ell$ averaged over all boundary positions, defined as
\begin{equation}
\mathbb{E}_\ell(S) = \sum_{k=1}^L S_{[k,k+\ell]} \text{ .}
\end{equation}
The result, shown in Fig.\ \ref{FIG_FALLOFF_PENTA} (right), closely follows the Calabrese-Cardy formula for periodic $1{+}1$-dimensional CFTs, given by \cite{Holzhey,CalabreseReview}
\begin{equation}
\label{EQ_CALABRESE_CARDY}
S_A = \frac{c}{3} \log \left( \frac{L}{\pi \epsilon} \sin\frac{\pi \ell}{L} \right) \simeq \frac{c}{3} \log \frac{\ell}{\epsilon} + O\left( (\ell/L)^2 \right) \text{ ,}
\end{equation}
with a numerical fit yielding $c\approx 4.2$ and $\epsilon\approx 1.1$ for a cutoff at $L=2605$ boundary sites.

The peculiar pair-wise correlation of boundary Majorana modes, suggesting a connection to \emph{Majorana dimer models} \cite{PhysRevB.94.115127}, will be more rigorously explored in the future. However, as the correlation structure clearly breaks the translation and scale invariance expected of CFT ground states, we now consider regular tilings with generic matchgate input.

\begin{figure}[htb]
\includegraphics[height=0.118\textheight]{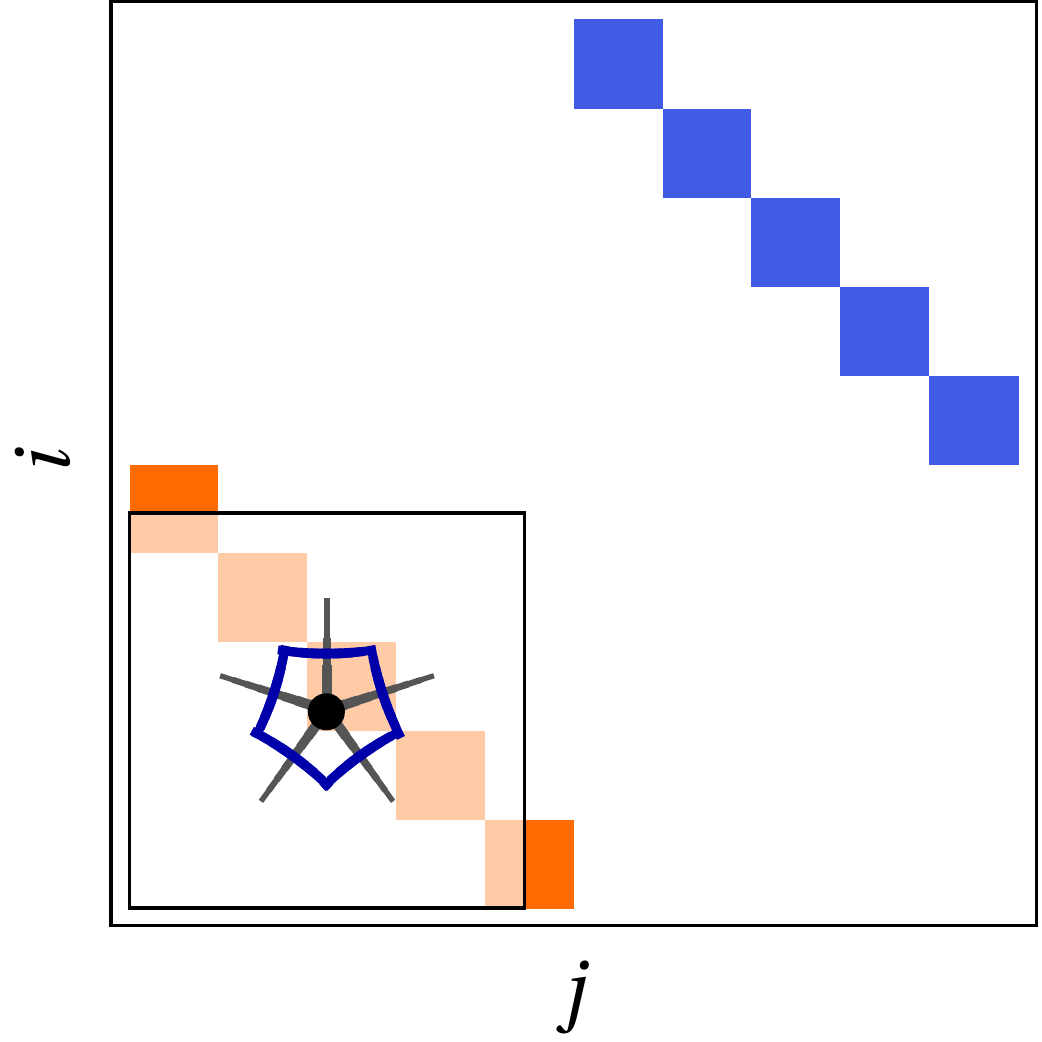}
\includegraphics[height=0.118\textheight]{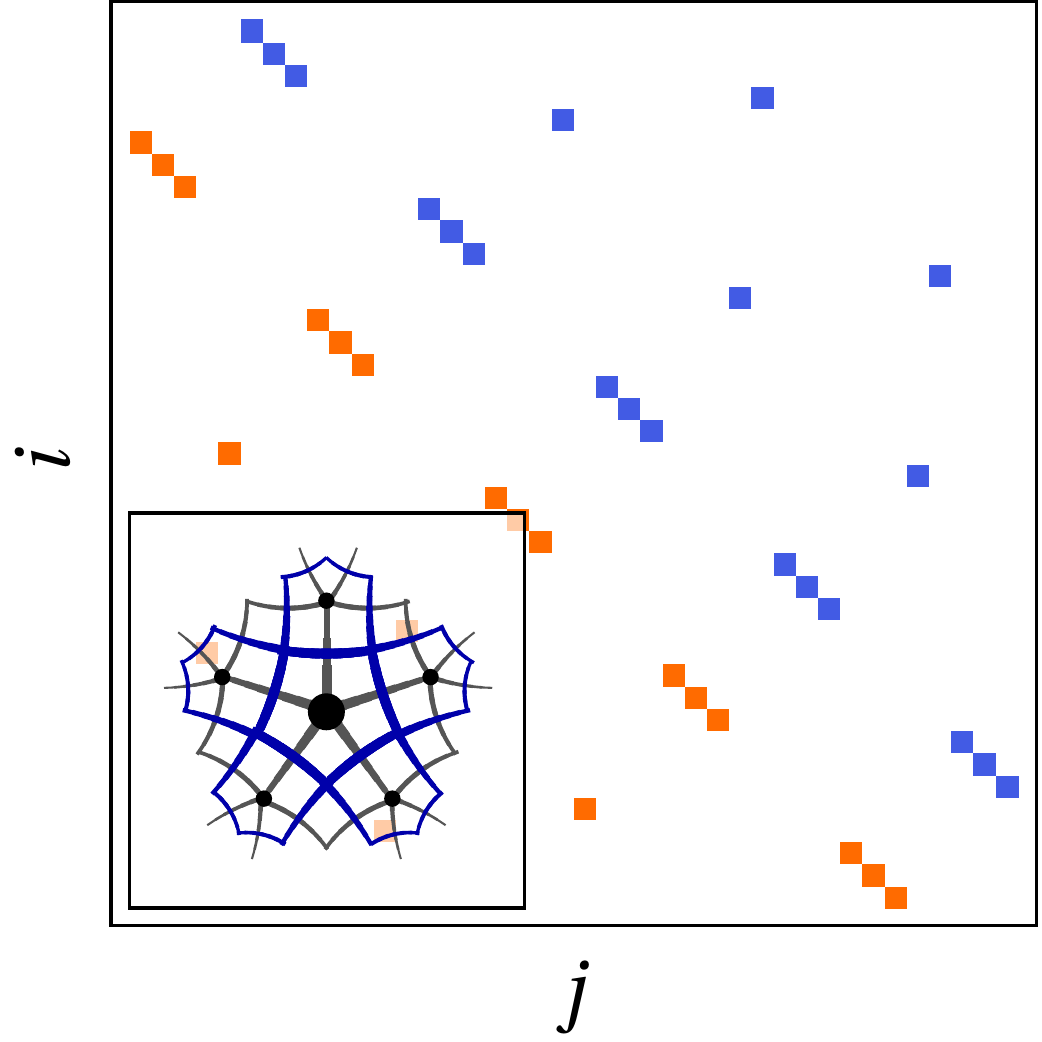}
\includegraphics[height=0.118\textheight]{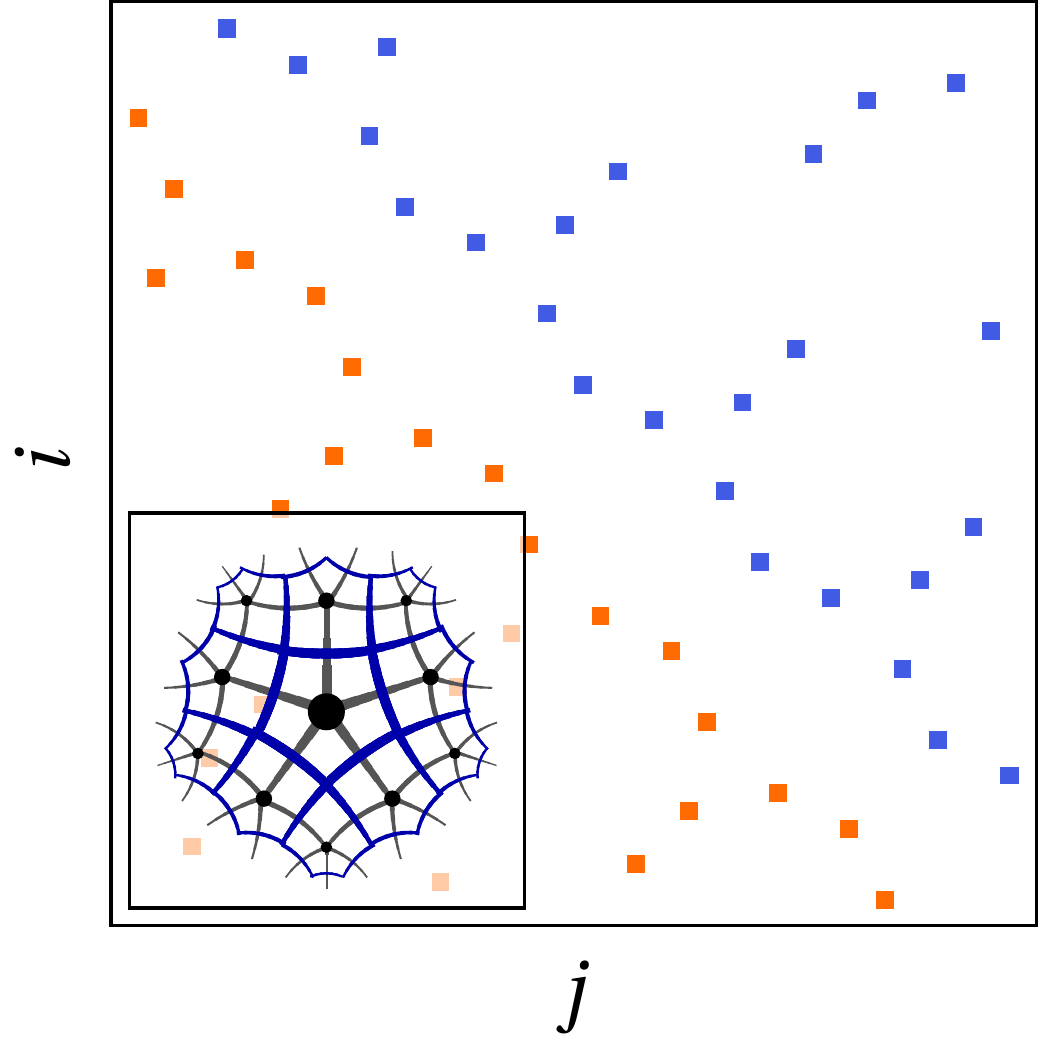}
\caption{Majorana covariance matrix $\Gamma$ for a boundary state of a pentagon tiling with 10, 40, and 50 boundary Majorana fermions (left to right). Entries $\Gamma_{i , j}$ are color-coded (orange = $+1$, blue = $-1$). The corresponding tiling/tensor network is shown in the lower left corner.}
\label{FIG_COVMM_PENTA}
\end{figure}

\begin{figure}[htb]
\includegraphics[width=0.23\textwidth]{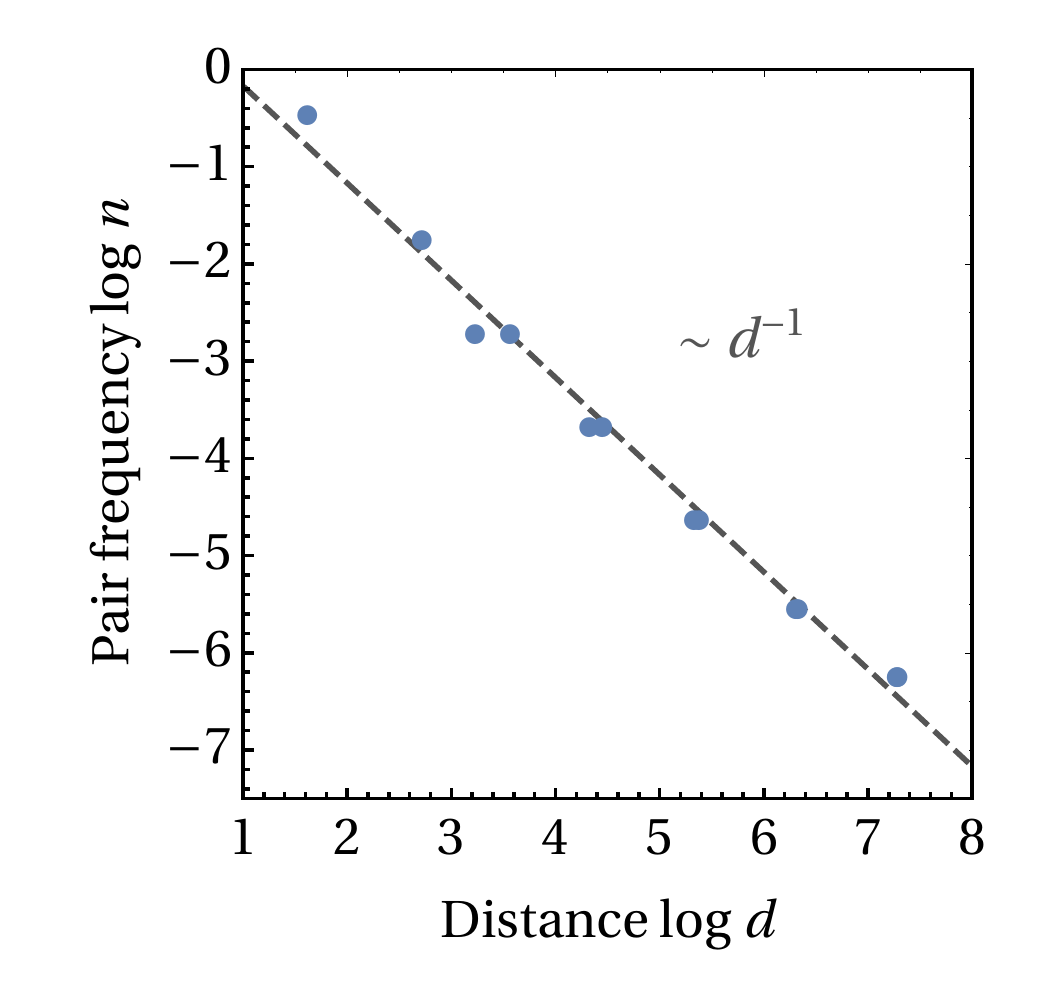}
\includegraphics[width=0.23\textwidth]{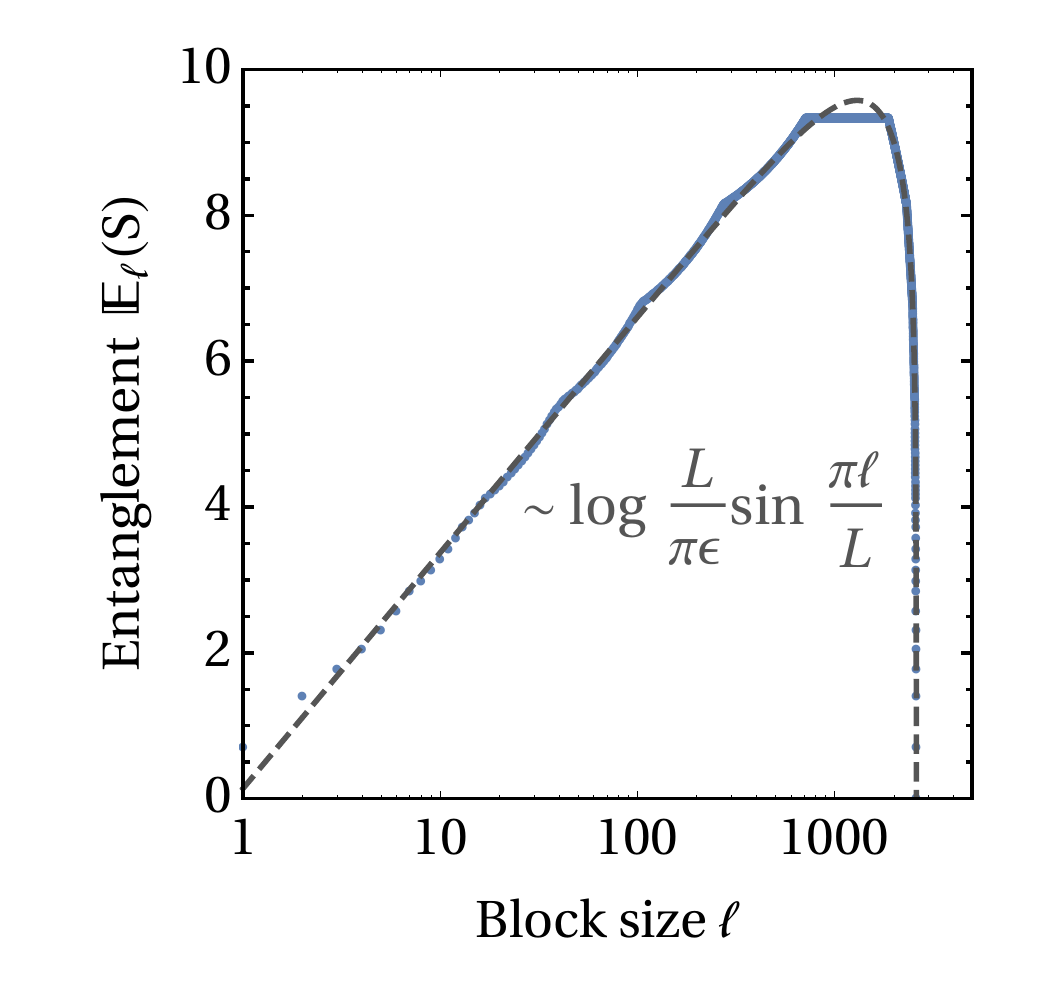}
\caption{Boundary state properties of the HaPPY code at 2605 boundary sites. \textsc{Left:} Average correlations at boundary distance $d$, computed as the relative frequency $n$ of Majorana pairs. Dashed gray line shows an $n(d) \sim 1/d$ numerical fit. \textsc{Right:} Scaling of average entanglement entropy $\mathbb{E}_\ell(S)$ with subsystem size $\ell$. Dashed gray line shows numerical fit using \eqref{EQ_CALABRESE_CARDY}.
}
\label{FIG_FALLOFF_PENTA}
\end{figure}

\paragraph*{Regular triangulations.}

As the boundary states of triangular tilings are necessarily Gaussian \cite{Bravyi2008}, 
we can study their properties comprehensively using matchgate tensors. The simplest such tilings are regular and isotropic, i.e.\ with each local tensor specified by the same antisymmetric $3 \stimes 3$ generating matrix $A$. Isotropy constrains its components to one parameter $a\seq A_{1,2} \seq A_{1,3} \seq A_{2,3}$.
The bulk topology follows from our choice of tiling. For triangular tilings ($p \seq 3$), setting $q \seq 6$ produces a flat tiling, whereas $q \sgr 6$ leads to a hyperbolic one (see Fig.\ \ref{FIG_TRI_CURVATURE}). 
Triangular tilings with $q \sle 6$ produce closed polyhedra
 that are positively curved and lack the notion of an asymptotic boundary.
As a convention, we choose the local orientation of the triangles so that the generating matrix for the contracted boundary state satisfies $A^\prime_{i,j} \sgr 0$ for $i \sgr j$, corresponding to anti-periodic boundary conditions: Covariance matrix entries $\Gamma_{i,j}$ acquire a sign flip when cyclic permutions push either index $i$ or $j$ over the boundary, as relative ordering is reversed.

We now consider the boundary states of $\lbrace 3, k \rbrace$ bulk tilings.
The falloff of correlations along the boundary generally depends on $k$, i.e.\ the bulk curvature, as shown in Fig.\ \ref{FIG_TRI_CFALLOFF} (top row) for the $a \seq 0.25$ case. 
While correlations between the boundary Majorana fermions of a flat bulk fall off exponentially, a hyperbolic bulk produces a polynomial decay (up to finite-size effects at large distances and rounding errors at very small correlations). In the hyperbolic case, geodesics between boundary points scale logarithmically in boundary distance, so the falloff is still exponential in bulk distance, as we would expect in AdS/CFT \cite{Balasubramanian:1999zv}.

\begin{figure}[htb]
\hspace{-1cm}

\includegraphics[width=0.23\textwidth]{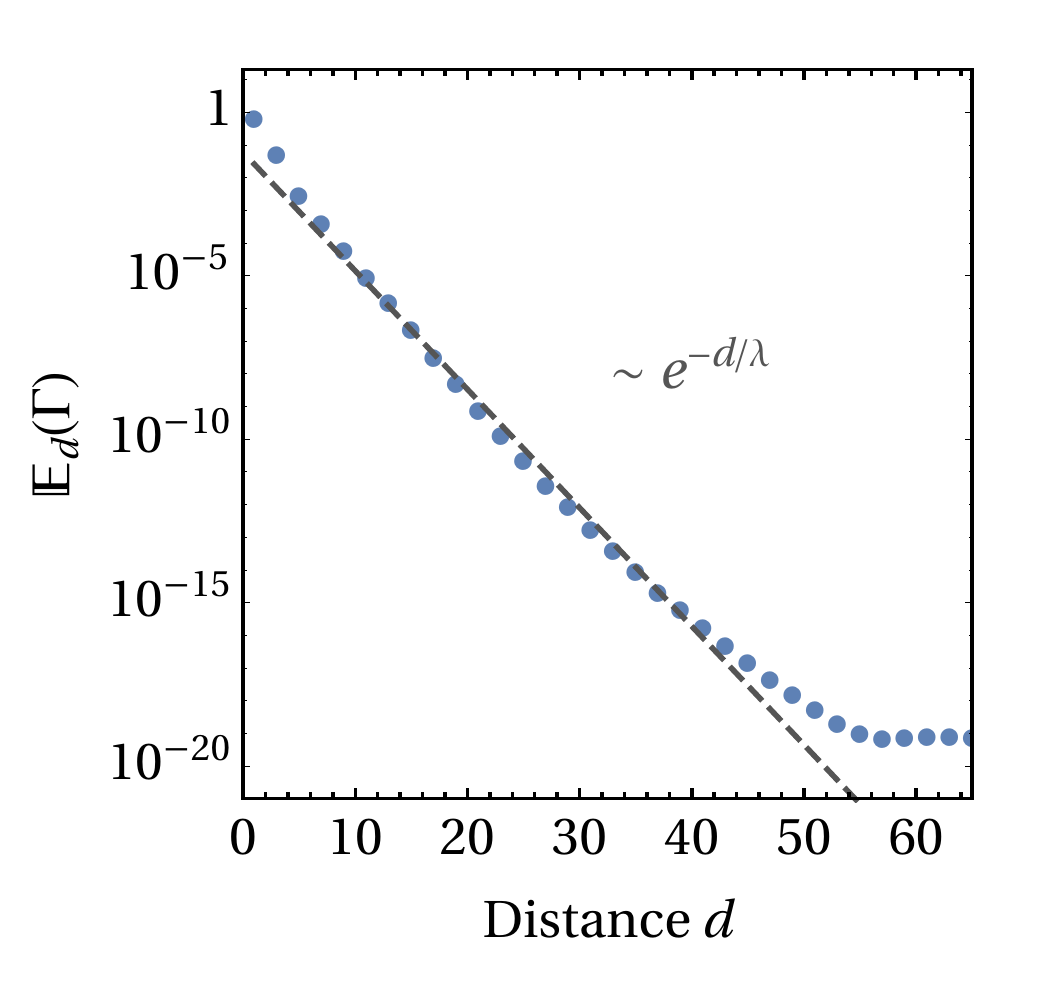}
\includegraphics[width=0.23\textwidth]{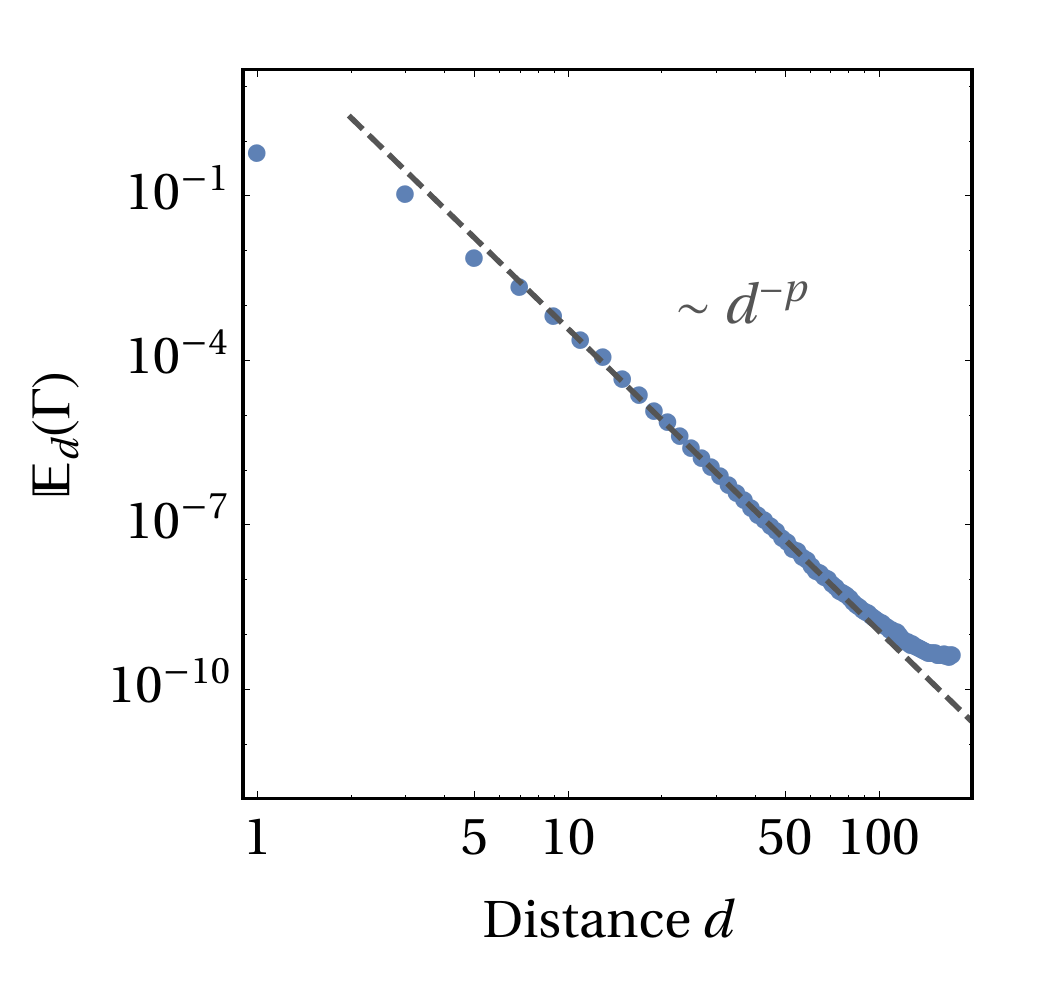}

\includegraphics[width=0.23\textwidth]{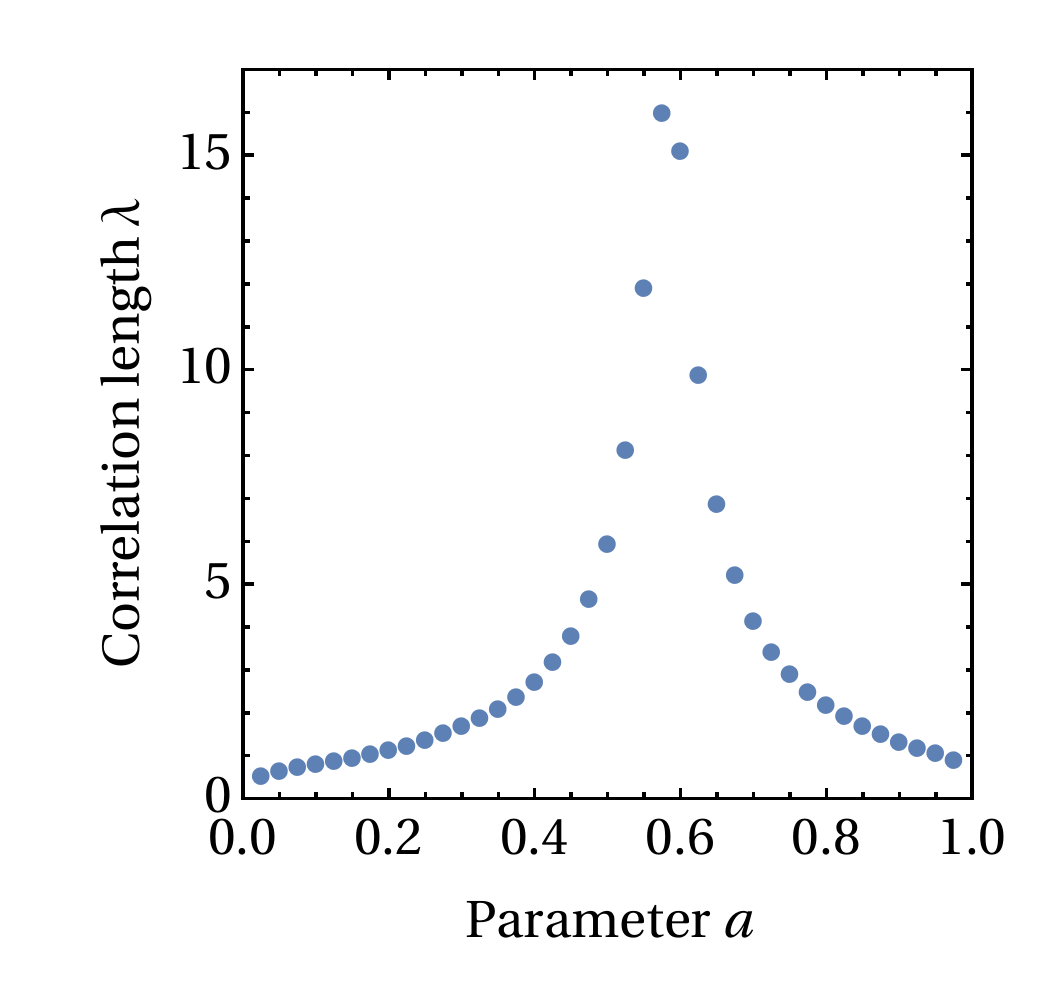}
\includegraphics[width=0.23\textwidth]{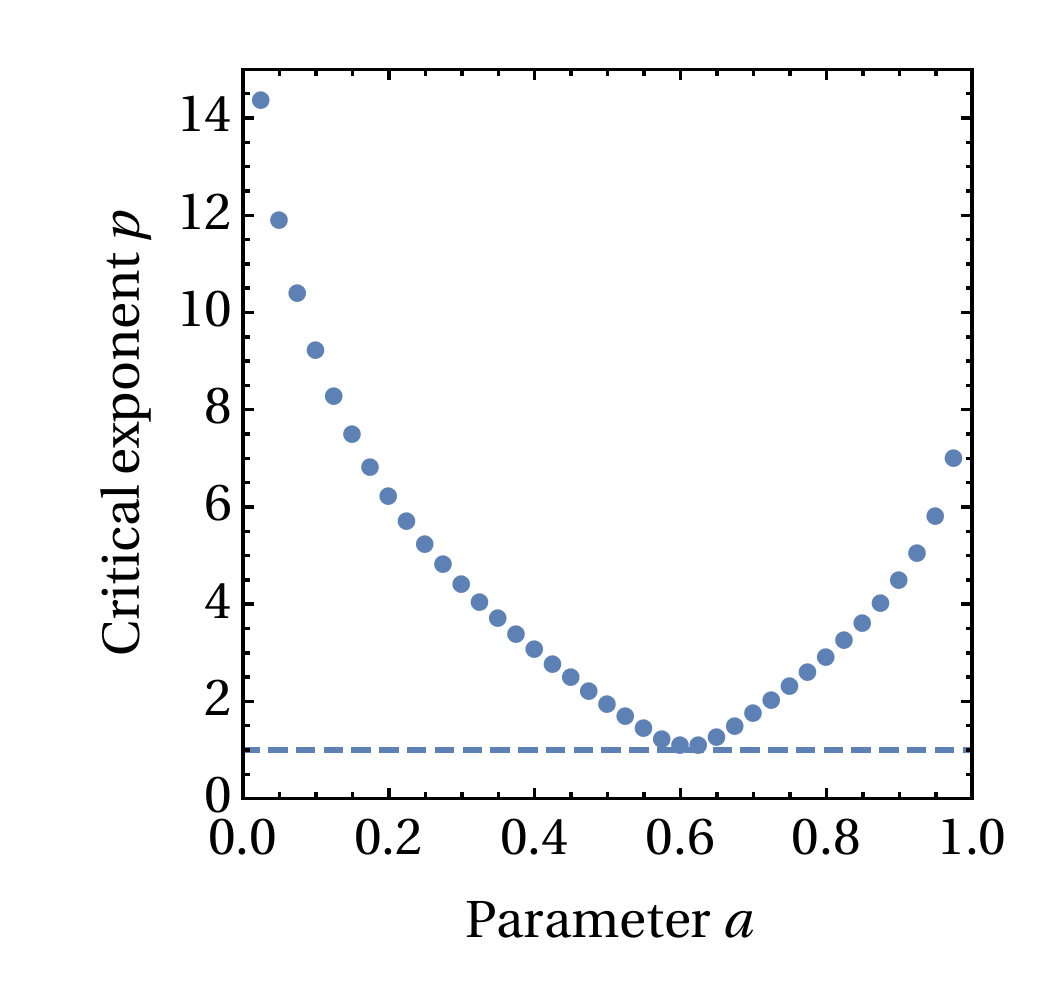}

\caption{\textsc{Top}: Mean value of Majorana covariance $\mathbb{E}_d(\Gamma) \seq  $ $
\sum_{k=1}^L |\Gamma_{k,k+d}|/L$ (with $\Gamma_{i,L+j} \seq  \Gamma_{i,j}$) at boundary distance $d$. For $\lbrace 3,6 \rbrace$ tiling with 150 boundary Majorana fermions (left) and $\lbrace 3,7 \rbrace$ tiling with 348 (right). $a \seq 0.25$ in both cases.
\textsc{Bottom}: Dependence of correlation falloff on $a$, for $\lbrace 3,6 \rbrace$ tiling with falloff $\sim e^{-d/\lambda}$ (left) and $\lbrace 3,7 \rbrace$ tiling with $\propto d^{-p}$ (right). 
}
\label{FIG_TRI_CFALLOFF}
\end{figure}

\begin{figure}[htb]
\hspace{-1cm}

\includegraphics[width=0.23\textwidth]{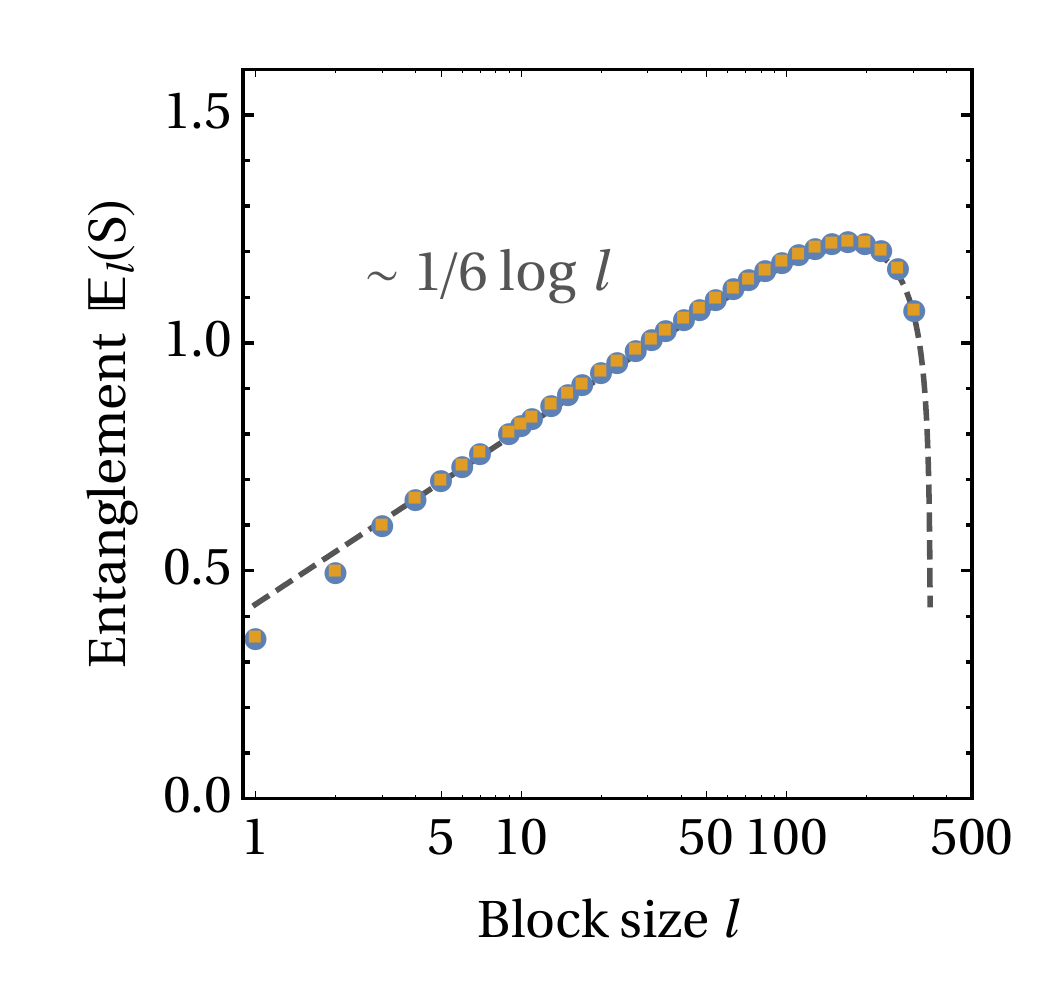}
\includegraphics[width=0.23\textwidth]{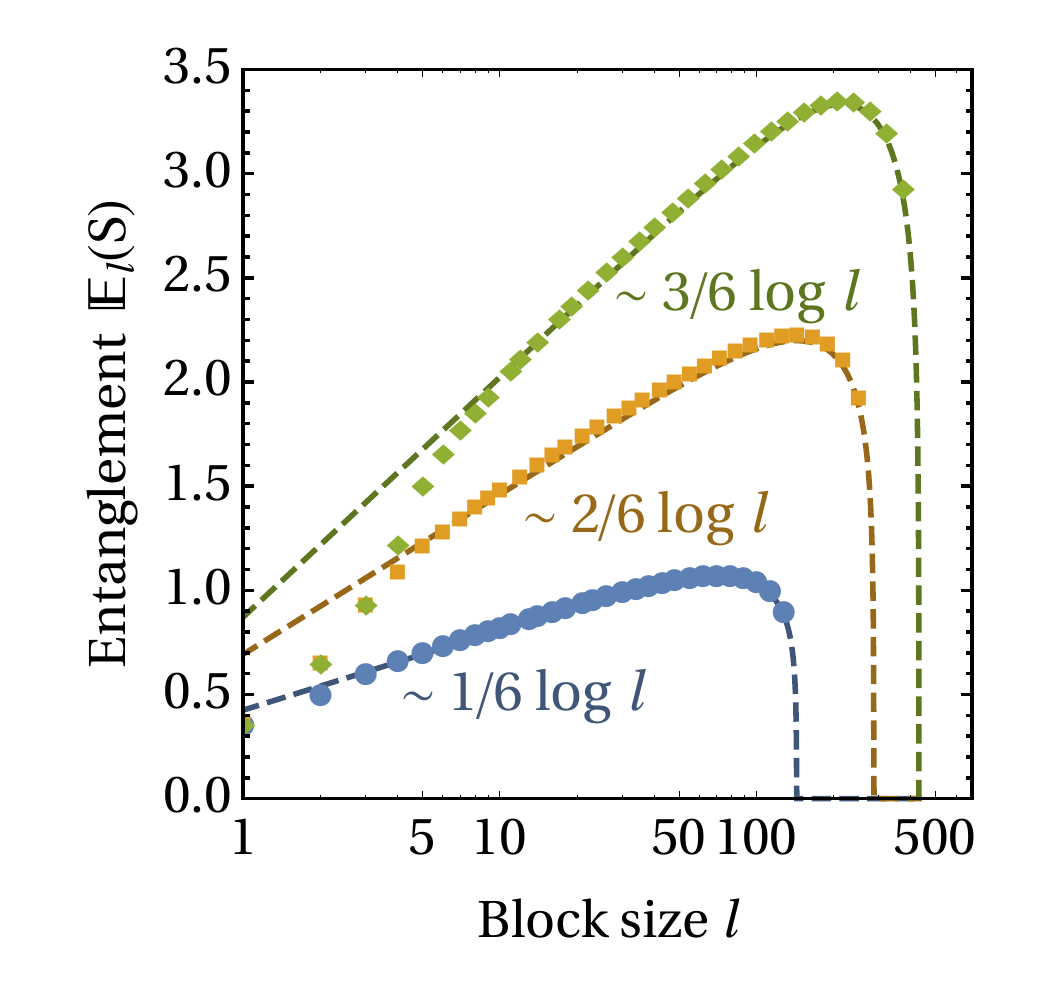}
\caption{Average entanglement entropy $\mathbb{E}_\ell(S)$ of a boundary subsystem with size $\ell$, numerical data and exact CFT solutions \eqref{EQ_CALABRESE_CARDY} (dashed lines). \textsc{Left:} Results for $\{3,6\}$ tiling (blue) at $a=0.580$ and $\{3,7\}$ tiling (yellow) at $a=0.609$ (for 348 Majorana sites each). \textsc{Right:} $\{3,8\}$ tiling with bond dimension $\chi=2,4,8$ on 144, 288, and 432 sites, respectively.
}
\label{FIG_REG_EE}
\end{figure}

Restricting ourselves to the $0 \sleq a \sleq 1$ region, we explore how quickly correlations decay in both settings. At $a \seq 0$ and $a \seq 1$, the boundary Majorana fermions only have neighboring pair correlations, either pairing within each edge ($a \seq 0$) or across the corners ($a \seq 1$). Thus correlation decay becomes infinite in the limits $a \sto 0 $ and $a \sto 1$, independent of bulk geometry.
We use numerical fits to study the remaining region $0 \sle a \sle 1$ (see Fig.\ \ref{FIG_TRI_CFALLOFF}, bottom row). For a hyperbolic bulk geometry the power law is generic with the slowest decay at $a \sapprox 0.61$, where we see a $\propto d^{-1}$ falloff over distance $d$. The exponential decay $\propto e^{-d/\lambda}$ generally produced by a flat bulk geometry, however, slows down to a power law (with correlation length $\lambda$ diverging) around $a \sapprox 0.58$, where correlations again decay as $\propto d^{-1}$. At their critical values, the boundary states of both bulk geometries have the same average properties.

\begin{table}[htb]
\begin{tabular}{c | c | c | c | c | c}
Parameter &  Exact & $\lbrace 3,6 \rbrace$ bulk & $\lbrace 3,7 \rbrace$ bulk & mMERA & Wavelets \\
\hline
$\epsilon_0$               & $-0.6366$      & $-0.6139$      & $-0.5617$      & $-0.6365$      & $-0.6211$\\
$c$	                       & $\msp 0.5000$  & $\msp 0.5006$  & $\msp 0.5018$  & $\msp 0.4958$  & $\msp 0.4957$ \\
$\Delta_\psi,\Delta_{\bar{\psi}}$    & $\msp 0.5000 $ & $\msp 0.4948$  & $\msp 0.4951$  & $\msp 0.5023$  & $\msp 0.5000$ \\
$\Delta_\epsilon$               & $\msp 1.0000 $ & $\msp 0.9856$  & $\msp 1.0121$  & $\msp 1.0027$  & $\msp 1.0000$ \\
$\Delta_\sigma$                 & $\msp 0.1250 $ & $\msp 0.1403$  & $\msp 0.1368$  & $\msp 0.1417$  & $\msp 0.1402$ \\
$C_{\sigma,\sigma,\epsilon}$ & $\msp 0.5000 $ & $\msp 0.5470$  & $\msp 0.5336$  & $\msp 0.5156$  & $\msp 0.4584$
\end{tabular}
\caption{Table of \emph{conformal data} for the regular $\lbrace 3,6 \rbrace$ and $\lbrace 3,7 \rbrace$ bulk tilings as well as the mMERA, compared to the exact results and the wavelet MERA \cite{PhysRevLett.116.140403}. Listed are the ground-state energy density $\epsilon_0$, central charge $c$, scaling dimensions $\Delta_\phi$ of the fields $\phi=\psi,\bar{\psi},\epsilon,\sigma$, and the structure constant $C_{\sigma,\sigma,\epsilon}$. The non-scaling of the identity $\id$ is discussed in Appendix \ref{ssec:app_conf_data}.}
\label{TBL_CFT_DATA}
\end{table}

Up to finite-size effects, this critical boundary theory turns out to be the \emph{Ising CFT}, as we confirm by computing a range of critical properties from the covariance matrix, shown in Table~\ref{TBL_CFT_DATA}. The entanglement entropy scaling, shown in Fig.\ \ref{FIG_REG_EE} (left), again matches the expected form \eqref{EQ_CALABRESE_CARDY} irrespective of the choice of tiling.
The Ising CFT state that we observe at the critical value of $a$ is the ground state of the Hamiltonian
\begin{equation}
\label{EQ_ISING_H}
H = \i \left( \sum_{i=1}^{N-1} \m_i \m_{i+1} + \m_1 \m_N \right) \text{ ,}
\end{equation}
where the sign of the boundary term $\m_1 \m_N$ signifies anti-periodic boundary conditions. Triangular tilings also incorporate more generic models: By associating each edge with a bond dimension $\chi \sgr 2$, it is possible to produce boundary theories with central charges in multiples of $1/2$, as laid out in Appendix \ref{ssec:app_higher_c} and shown in Fig.\ \ref{FIG_REG_EE} (right). Furthermore, by changing the tensor content in a central region of the network, a mass gap can be introduced, highlighting how radii in a hyperbolic bulk correspond to a renormalization scale on the boundary.
Details are provided in Appendix \ref{ssec:app_ir_cut}.

\paragraph*{Translation invariance and MERA.}

The regular bulk tilings considered so far possess a set of discrete symmetries. When choosing identical tensors on each polygon, the boundary states neccessarily inherit these symmetries, breaking translation invariance. To recover it, we consider a tiling with the same geometry as the MERA network. As we restrict ourselves to real generating matrices for the 3- and 4-leg matchgate tensors in this geometry, our model is not a unitary circuit but a model of Euclidean entanglement renormalization resembling imaginary time evolution, extending ideas from Refs.\ \cite{Evenbly2015TNR1,Evenbly2015TNR2}. This may provide a more realistic representation of the causal structure of an AdS time-slice than the standard MERA. Accordingly, the tensors of our \textsl{matchgate MERA} (mMERA) do not correspond to the usual (norm-preserving) isometries and disentanglers. Remarkably, we can still produce 
almost perfectly translation-invariant boundary states (Fig.~\ref{FIG_COVFEFE_COMPARISON}) while optimizing over only three parameters and recover the expected CFT properties (Table~\ref{TBL_CFT_DATA}). In particular, at bond dimension $\chi \seq 2$ the ground-state energy has a relative error of only $0.02\%$ compared to the exact solution. Note that the optimization process only takes a few minutes on a desktop computer for a network with hundreds of tensors. 
We also find that the $\chi \seq 2$ mMERA possesses a symmetry that allows us to write its 4-leg tensors as contractions of simpler 3-leg tensors (see Fig.\ \ref{FIG_MERA}), yielding a non-regular triangular tiling.  Whether alternating or quasi-periodic
tilings with a larger parameter space than regular tilings can also produce translation-invariant states is an interesting open question.
\begin{figure}
\centering
\includegraphics[height=0.14\textheight]{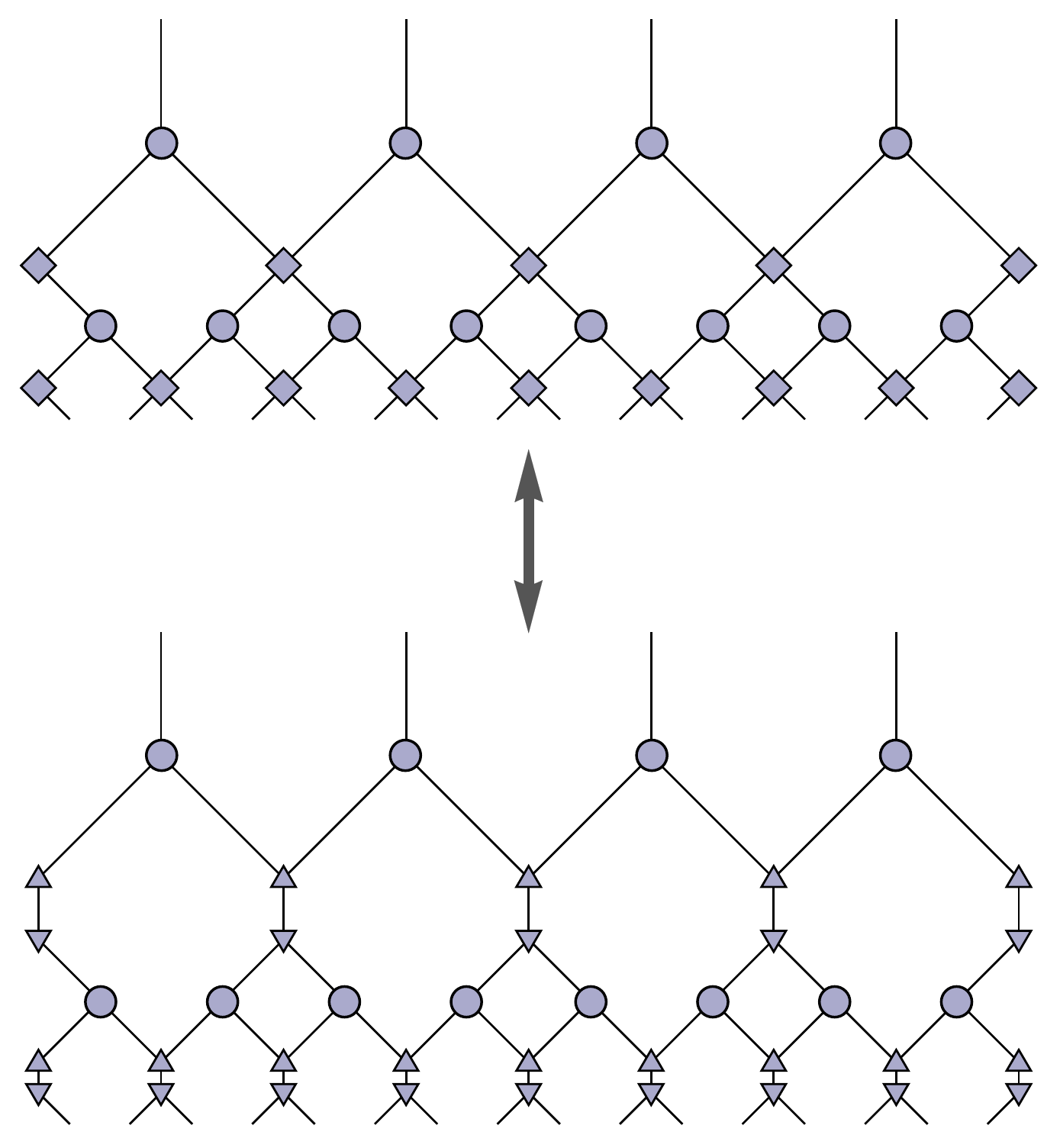}
\hspace{0.05\textwidth}
\includegraphics[height=0.14\textheight]{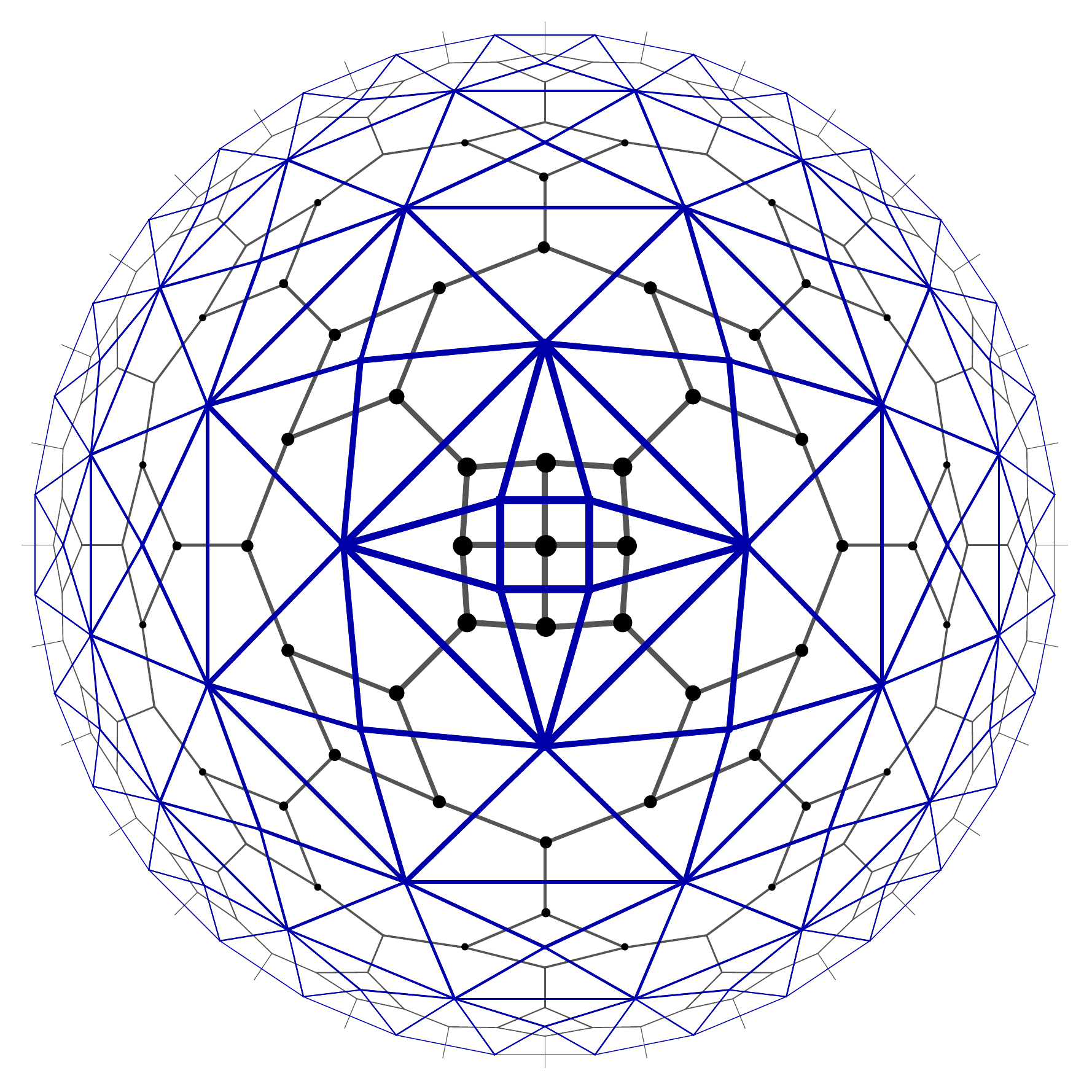}

\caption{\textsc{Left:} Mapping between standard MERA tensor network and a network of 3-leg tensors. \textsc{Right:} Mapping of the corresponding mMERA network onto a triangular tiling in the Poincar\'e disk.}
\label{FIG_MERA}
\end{figure}

\begin{figure}
\centering
\vspace{11pt}

\includegraphics[height=0.118\textheight]{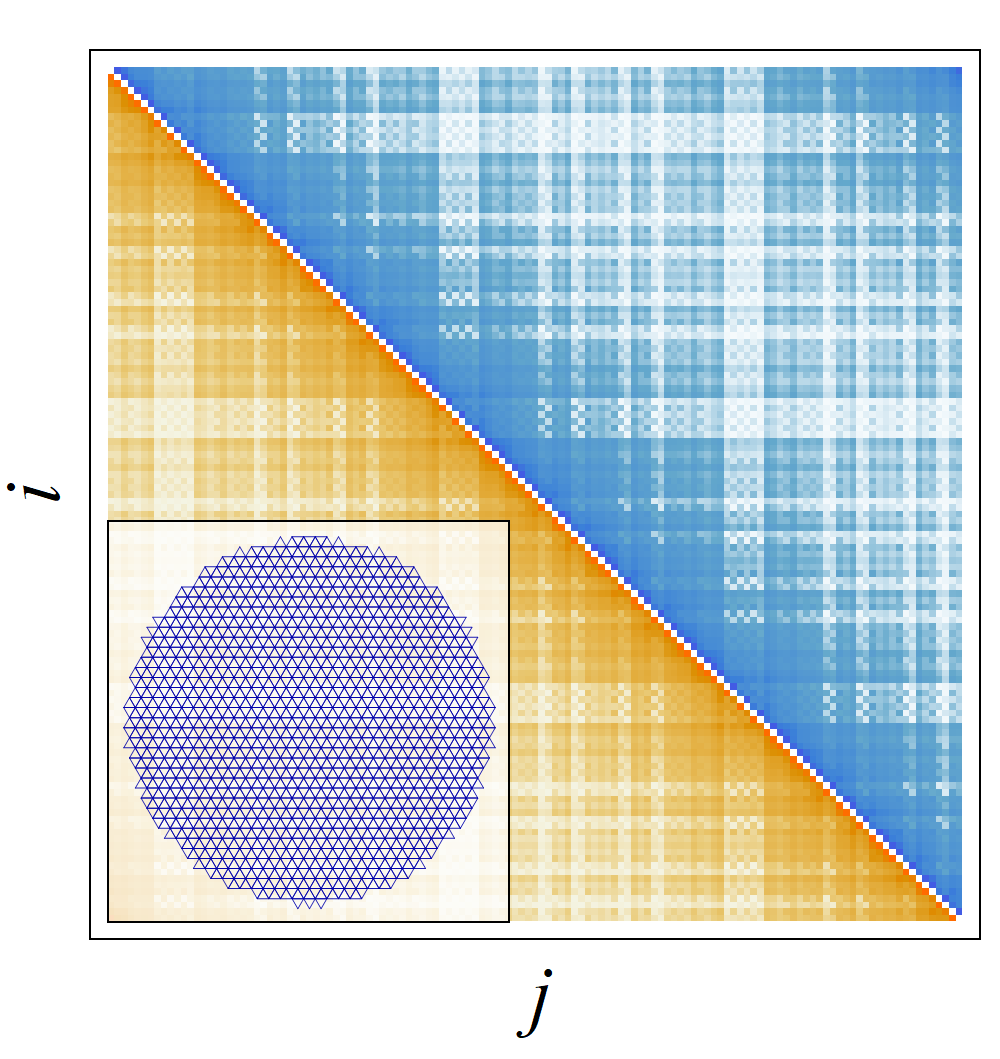}
\includegraphics[height=0.118\textheight]{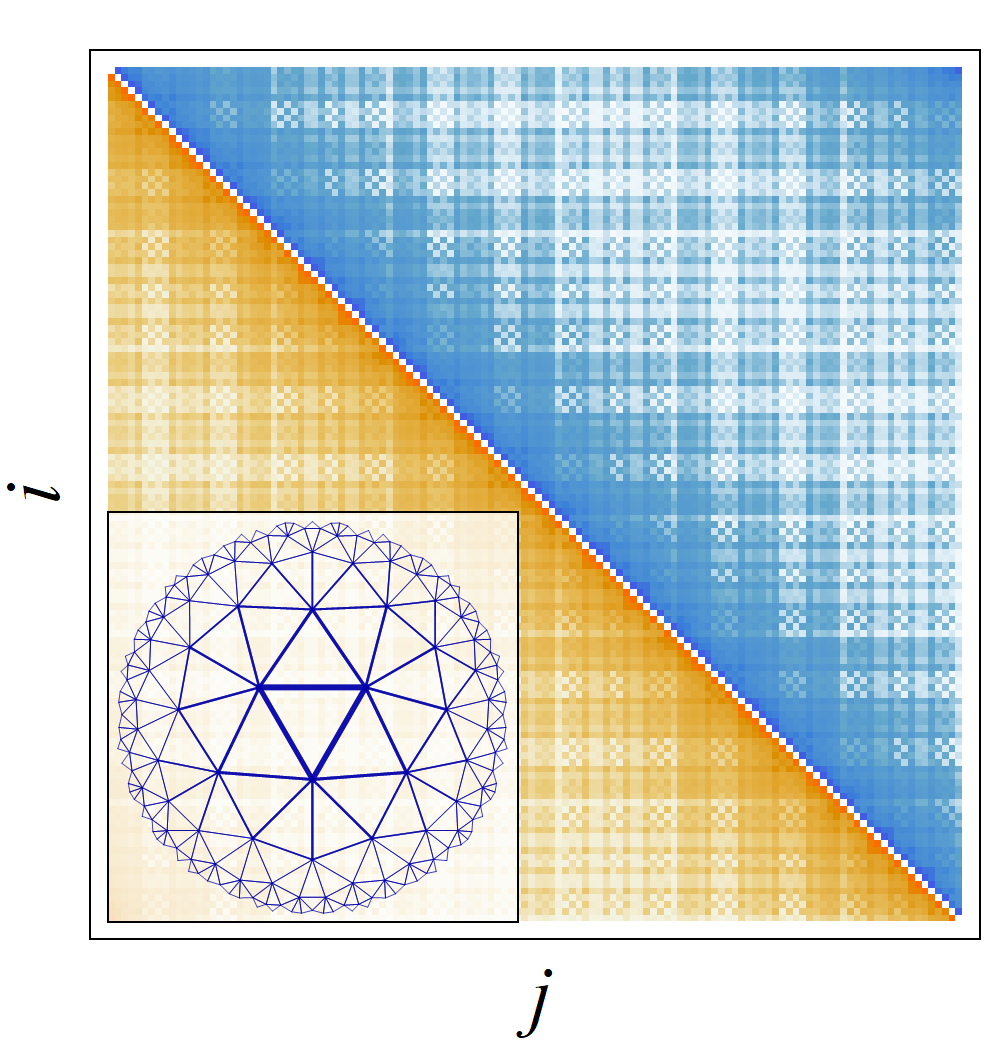}
\includegraphics[height=0.118\textheight]{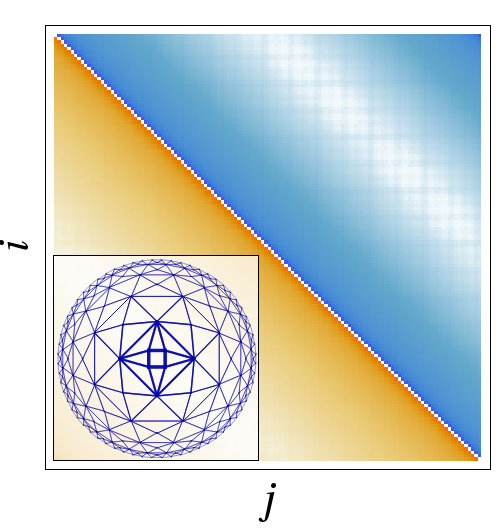}

\caption{Color-coded entries of the field correlation matrix $\langle \psi_i \psi_j \rangle \seq (\Gamma_{2i,2j-1} + \Gamma_{2i-1,2j})/4$ for boundary states of the $\lbrace 3,6 \rbrace$, 
$\lbrace 3,7 \rbrace$ and mMERA tiling at criticality (from left to right) with 129, 126 and 128 boundary sites, respectively. The corresponding tiling is shown in the lower left corner.
}
\label{FIG_COVFEFE_COMPARISON}
\end{figure}

\paragraph*{Discussion}

In this work, we have studied bulk-boundary correspondences in fermionic Gaussian tensor networks, introducing a 
versatile framework and a highly efficient contraction method based on matchgate tensors \cite{Valiant2002, Bravyi2008} for a wide class of flat and hyperbolic bulk tilings.
We showed that our framework includes the holographic pentagon code built from $5$-qubit stabilizer states for fixed bulk inputs. Its boundary states correspond to a non-local bulk pairing of Majorana fermions, opening an avenue to studying the state properties of this holographic model at large sizes. We explicitly computed 2-point correlators and entanglement entropies, which were found to exhibit critical scaling.
Beyond known models, we showed that critical and gapped Gaussian boundary states can be realized by various bulk tilings. In particular, the average scaling properties of the $c \seq 1/2$ Ising CFT (and multiples thereof) can be reproduced using regular one-parameter bulk triangulations with both flat and hyperbolic curvature. This is particularly unexpected for the flat case where boundary theories are typically gapped, and raises the question whether this appearance of criticality is retained in strongly interacting models, as well. Intriguingly, the appearance of equivalent boundary CFT states for flat and hyperbolic bulks resembles the effect of local Weyl transformations in Euclidean path-integrals \cite{Caputa:2017urj}.
Furthermore, we constructed the \textsl{mMERA}, a Euclidean matchgate tensor network based on the MERA geometry. Beyond the results achievable with regular triangulations, this tiling - which can also be expressed as a triangulation - recovers the Ising CFT with translation invariance, while requiring only three free parameters and little computational cost. 
Within the Gaussian setting, further studies could focus on positively curved bulks, higher-dimensional models and random tensors. Beyond Gaussianity, our framework is extendable to interacting bulk-boundary models
which can be explored with interacting fermionic tensor networks \cite{PhysRevA.80.042333,CorbozPEPSFermions, PhysRevA.81.052338,PhysRevB.95.245127,PhysRevB.95.075108}, yet avoiding the prohibitive computational effort of general methods for exact tensor contraction.

\paragraph*{Acknowledgements.} We would like to thank Aleksander Kubica for pointing out that $[[5,1,3]]$ code states 
\cite{Laflamme1996} are in fact ground states of quadratic Majorana Hamiltonians. We also thank Xiao-Liang Qi, Tadashi Takayanagi, and Pawel Caputa for helpful discussions. We thank the ERC (TAQ), the 
DFG (CRC 183, EI 519/7-1, EI 519/14-1), the Templeton Foundation, the EC (PASQuanS), 
the Studienstiftung, and the A.v.-Humboldt Foundation for support.

\newpage
\begin{widetext}
\begin{appendix}
\section*{Appendix}
This appendix contains additional details, technical calculations and proofs of the assertions made in the main text.
We begin by showing how to perform tensor contraction using Grassmann integration in Section\ \ref{sec:app_int}, followed by a minimal example of contracting two matchgate tensors.
In Section\ \ref{sec:app_defs}, we restructure the definitions made in Ref.\ \cite{Bravyi2008} in order to bring the theory of matchgates closer to the free fermionic formalism.
In particular, we prove the correspondence between matchgate tensors and fermionic Gaussian states.
In Section\ \ref{sec:app_conversion}, we show how to convert a generating matrix of a matchgate tensor to the covariance matrix of the corresponding state, yielding the physically relevant correlations.
In Section\ \ref{sec:app_contr}, we provide technical details and calculations for the contraction rules in the Grassmann formalism used in the numerical implementation.
In Section\ \ref{sec:app_explicit}, we give explicit examples of generating matrices relevant to the main text, show the extracted critical data from the respective covariance matrices, and present a construction of states with higher central charges.

\section{Tensor contractions in the Grassmann formalism}
\label{sec:app_int}
In this section, we review the approach to contraction of matchgate tensor networks through Grassmann integration \cite{Bravyi2008}.
In particular we present a simplified version of Lemma 5 from Ref.\ \cite{Bravyi2008} and explain this result through an example.
Grassmann variables will be denoted by $\theta$ and are a set of anti-commuting generators of an algebra
$
(  \theta_i \theta_j = - \theta_j \theta_i )
$
which nevertheless commutes with ordinary scalars $x$
$
(x \theta_i =  \theta_i x).
$
A general element  in this algebra may be written as
\begin{align}
 z=\sum _{k=0}^{n}\sum _{i_{1},\cdots ,i_{k}}c_{i_{1}\cdots i_{k}}\theta _{i_{1}}\cdots \theta _{i_{k}},
\end{align}
where $c_{i_{1}\cdots i_{k}}$ can be arbitrary complex coefficients and $i_k$ form an increasing sequence in $\{1,2,\ldots,n\}$.
In particular, given a tensor $T: \BC r \rightarrow \mathbb C$ we associate to it a polynomial in Grassmann numbers given by 
\begin{align}
\Phi_T(\theta) = \sum_{j\in \BC r} T(j) \theta_1^{j_1}\theta_2^{j_2}\ldots\theta_r^{j_r} \text{ ,}
\end{align}
which we call its characteristic function.

For simplicity, we consider contracting the last index of a rank-$r_1$ tensor $T_1$ with the first index of a rank-$r_2$ tensor $T_2$ where $r_1,r_2 \ge 1$.
Let us denote $r_1'=r_1-1$ and $r_2'=r_2-1$
This operation gives rise to a rank-$(r_1'+r_2')$ tensor $T_{1\star 2}$ with entries
  \begin{align}
  T_{1\star 2}(x,y)=\sum_{z \in \{0,1\}} T_1(x,z)T_2(z,y)
  \end{align}
for $x\in \BC {r_1'}$ and $y\in\BC{r_2'}$ being binary words.
The characteristic function for the contraction of two tensors can be obtained by
  \begin{align}
    \Phi_{T_{1\star 2}}(\tilde\theta,\tilde\eta) =\int d\eta_1\int d\theta_{r_1} \; \Phi_{T_1}(\theta) \Phi_{T_2}(\eta)\exp( \theta_{r_1} \eta_{1} )
  \end{align}
  where $\tilde \theta=(\theta_1,\ldots,\theta_{r_1'})$, $\tilde\eta=(\eta_1,\ldots,\eta_{r_2'})$ correspond to uncontracted indices and $\theta_{r_1}$ and $\eta_{1}$ are the two Grassmann numbers of the two indices that are being contracted.
Let us use $\exp(\theta_{r_1}\eta_1)=1+\theta_{r_1}\eta_1$ on the right hand side
\begin{align}
  RHS=\sum_{\substack{x\in\BC {r_1'}\\y\in\BC{r_2'}}} \sum_{a,b\in\{0,1\}} T_1(x,a)T_2(b,y)\int d\eta_1 \int d \theta_{r_1}\ \theta_1^{x_1}\theta_2^{x_2}\ldots\theta_{r_1'}^{x_{r_1'}} \theta_{r_1}^{a}\eta_1^b\eta_2^{y_1}\ldots\eta_{r_2}^{y_{r_2'}}(1+\theta_{r_1}\eta_1)
\end{align}
and observe that the two integrals commute with the first $r_1-1$ of the $\theta$'s and exponential factor commutes with the $\eta$'s. This gives
\begin{align}
  RHS=\sum_{\substack{x\in\BC {r_1'}\\y\in\BC{r_2'}}} \sum_{a,b\in\{0,1\}} T_1(x,a)T_2(b,y) \theta_1^{x_1}\theta_2^{x_2}\ldots\theta_{r_1'}^{x_{r_1'}} \left[ \int d\eta_1 \int  d \theta_{r_1} \ \theta_{r_1}^{a}\eta_1^b(1+\theta_{r_1}\eta_1)\right]
\; \eta_2^{y_1}\ldots\eta_{r_2}^{y_{r_2'}}.
\end{align}
For the middle bracket, we obtain
\begin{align}
  \int d\eta_1 \int  d \theta_{r_1}\ \theta_{r_1}^{a}\eta_1^b(1+\theta_{r_1}\eta_1)=\delta_{a,b}
\end{align}
and therefore
\begin{align}
  RHS= \sum_{\substack{x\in\BC {r_1'}\\y\in\BC{r_2'}}} \left[ \sum_{z\in\{0,1\}} T_1(x,z)T_2(z,y)\right] \theta_1^{x_1}\theta_2^{x_2}\ldots\theta_{r_1'}^{x_{r_1'}}\eta_2^{y_1}\ldots\eta_{r_2}^{y_{r_2'}}.
  \end{align}
We see that this is exactly the characteristic function for the tensor contraction.
Note that it is important that we contract the last index with the first one.
Lemma 5 of Ref.\ \cite{Bravyi2008}  generalizes this calculation to an arbitrary number of indices that are being contracted in an appropriate order -- this essentially could be derived by iterating the formula that we derived for the case of self-contractions.

\subsection{Minimal example}

As an example, 
we want to show the contraction of two tensors $T_A, T_B$ with Gaussian characteristic functions of the form
\begin{align}
  \Phi_{T_A}(\theta)=\exp\biggl(\frac 12\sum_{j,k=1}^{L} A_{j,k} \theta_j \theta_k\biggr)\;,
  \quad\quad
  \Phi_{T_B}(\theta)=\exp\biggl(\frac 12\sum_{j,k=1}^{L} B_{j,k} \theta_j \theta_k\biggr)\;.
\end{align}
As we summarize in Fig.\ \ref{FIG_TRICONTRACTION}, we associate the Grassmann numbers $\theta_1,\theta_2,\theta_3$ to $A$ and $\theta_4,\theta_5,\theta_6$ to $B$, and fix $\theta_3$ and $\theta_4$ on the edges between the triangles to be integrated out, yielding a state with four edges whose correlation matrix $C$ shall be computed from $A$ and $B$ via Grassmann integration
\begin{equation}\label{eq:ContractToC}
\Phi_C(\Theta)= \int \d\theta_4 \d\theta_3\,  e^{\theta_3 \theta_4 + \frac{1}{2}\sum_{j,k=1}^6 (A\oplus B)_{j,k}\theta_j\theta_k}
\end{equation}
where $\Theta=( \theta_1, \theta_2, \theta_5, \theta_6 )^\top$ contains the four Grassmann numbers that remain after integration.
Taking $A, B$ as input, we find that after the integration $\Phi_{T_C}$ is again Gaussian (the technicalities of integration will be dealt with in Section\ \ref{sec:app_contr}) and the generating matrix is 
\begin{equation}
C = \left(
\begin{matrix}
0                & A_{1 , 2}         & A_{1 , 3} B_{4 , 5} & A_{1 , 3} B_{4 , 6} \\
-A_{1 , 2}         & 0               & A_{2,  3} B_{4 , 5} & A_{2,  3} B_{4 , 6} \\
-A_{1 , 3} B_{4 , 5} &-A_{2,  3} B_{4 , 5} & 0               & B_{5 , 6}         \\
-A_{1 , 3} B_{4 , 6} &-A_{2,  3} B_{4 , 6} &-B_{5 , 6}         & 0               \\
\end{matrix}
\right) \text{ .}
\end{equation}
We observe that the newly created entries in the upper right corner are in fact a dyadic product and this block is a lower-rank matrix.

\begin{figure}[htb]
\centering
\includegraphics[width=0.46\textwidth]{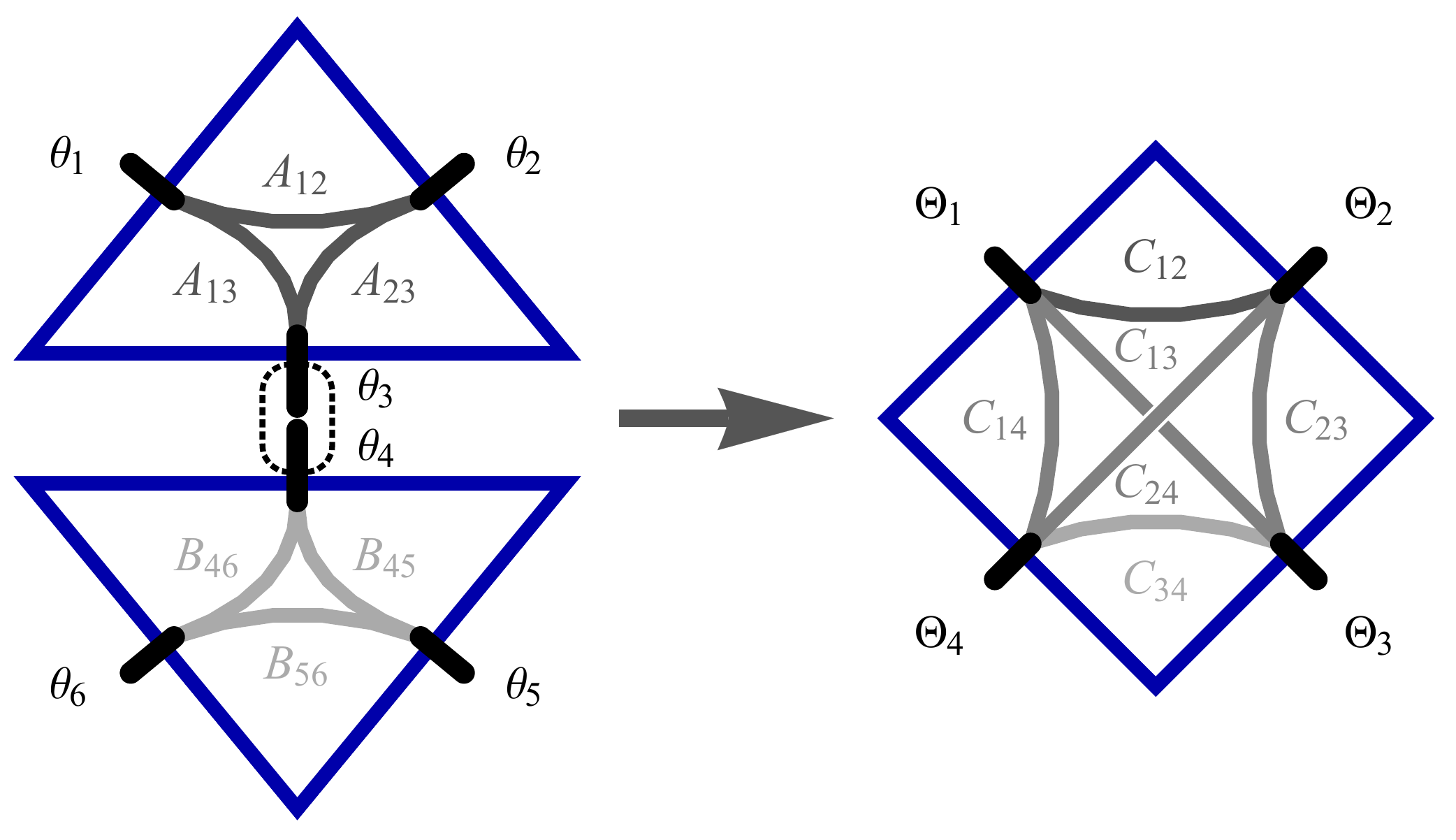}

\caption{Contraction of two triangle states with Grassmann-variable edges $\lbrace \theta_1,\theta_2,\theta_3 \rbrace$ and $\lbrace \theta_4,\theta_5,\theta_6 \rbrace$ into a state with four edges $\lbrace \Theta_1,\Theta_2,\Theta_3,\Theta_4 \rbrace$.}
\label{FIG_TRICONTRACTION}
\end{figure}

\section{Matchgates and fermionic Gaussian states}
\label{sec:app_defs}
We first discuss definitions of matchgate tensors and then explain the connection to fermionic Gaussian states.
\subsection{Definitions of matchgates}
For completeness, we shortly recapitulate the characterization of matchgates by Bravyi in Ref.\
\cite{Bravyi2008}.
Originally matchgates \cite{valiant2002quantum} were characterized as the local tensors of a tensor network that can be contracted efficiently through the \emph{Fisher-Kastelyn-Temperley method} \cite{Temperley1961, Kasteleyn1961}.
Subsequently, the following algebraic characterization has been found \cite{cai2007theory}.
\begin{definition}[Matchgate equations]
A rank-$r$ tensor $T$ is a matchgate if for all $x,y \in \BC r$ it holds that
\begin{align}
  \sum_{k:\ x_k\neq y_k} T(x\xor e^k)T(y\xor e^k)(-1)^{\sum_{j=1}^{k-1}(x_j+y_j)}=0
\end{align}
where $(e^k)_q=\delta_{k,q}$.
\end{definition}
Proposition 2 of Ref.\ \cite{Bravyi2008} shows that one can equivalently define matchgates through Pfaffians, leading to the following equivalent definition.
\begin{definition}[Matchgates as Pfaffians]
A rank-$r$ tensor $T$ is a matchgate if there exists a reference index $z \in \BC r$ and an anti-symmetric matrix $A\in \mathbb C^{r\times r}$, such that for any $x\in \BC r$
\begin{align}
T(x)=\Pf{A_{|x\xor z}}T(z).
\end{align}
In particular, $A$ is explicitly given by $A_{j,k}=T(e^j \xor e^k\xor z)/T(z)$ for $j<k$ if $T(z)\neq0$ or $A\equiv0$ if $T\equiv 0$ and we denote by $A_{|x\xor z}$ the restriction of $A$ to the entries indicated by $x\xor z$. 
\end{definition}
\emph{Proof of the equivalence of these definitions.}
In particular, proposition 2 of Ref.\ \cite{Bravyi2008} shows that whenever $T(z)\neq 0$ then a new matchgate tensor $T'$ fulfills $T'(x)=T(x \xor z)/T(z)=\Pf{A_{|x}}$ which is derived from the matchgate equations.
From this we have that $T'(x\xor z)=T(x\xor z\xor z)/T(z)=T(x)/T(z)=\Pf{A_{|x\xor z}}T(z)$ and therefore $T(x)=\Pf{A_{|x\xor z}}T(z)$.
Finally, for the trivial matchgate tensor $T =  0$ both definitions agree, too.

The converse direction can be shown by a simplification of the argument of Theorem 2 of Ref.\ \cite{Bravyi2008}. 
We start with $T(x)=\Pf{A_{|x\xor z}}T(z)$.
If $T(z)=0$ then $T$ fulfills the matchgate equations trivially.
Otherwise we consider $T'$ with entries $T'(x) = T(x\xor z)/T(z) = \Pf{A_{|x}}$ which has a Gaussian characteristic function 
\begin{align}
  \Phi_{T'}(\theta)=\exp\biggl(\frac 12\sum_{j,k=1}^{L} A_{j,k} \theta_j \theta_k\biggr)\;.
\end{align}
This is argued as follows. 
As this is a Gaussian characteristic function by the theory of Ref.\ \cite{bravyi2004lagrangian} the Lemma 1 in Ref.\ \cite{Bravyi2008} applies which shows that $T'$ fulfills matchgate equations. 
Finally, from Proposition 1 in op.\ cit., or by a shift of variables, we find that $T$ also satisfies these equations.
\qed

The following lemma shows that matchgates have Gaussian characteristic functions.
\begin{lemma}[Grassmann exponentials]
Let $A=-A^\top\in \mathbb C^{r\times r}$ for some positive integer $r$, then we have
\begin{equation}
\exp\biggl(\frac 12\sum_{j,k=1}^{L} A_{j,k} \theta_j \theta_k\biggr)=\sum_{x\in \BC r} \Pf{A_{|x}}\ \theta_1^{x_1}\theta_2^{x_2}\ldots\theta_r^{x_r}\;.
\label{eq:PhiPF}
\end{equation}
\end{lemma}
We omit the proof which proceeds by using the definition of the exponential series which is truncated to first $\lfloor n/2 \rfloor$ powers of the quadratic form and then regrouping terms that have the same normal ordered Grassmann monomial. 
Keeping track of the sign in such reordering, yields the sign of the permutation and subsequently the Pfaffian can be identified.
From this it follows that an even matchgate with $z=\overline 0$ and covariance matrix $A$ has a Gaussian characteristic function 
\begin{equation}
\Phi_T(\theta)=T(\overline 0)\exp\biggl(\frac 12 \sum_{j,k=1}^{L} A_{j,k} \theta_j \theta_k\biggr)\;.
\end{equation}
For matchgates with $z\neq \overline 0$ we refer the reader to Theorem 2 of Ref.\ \cite{Bravyi2008} for a general form of the characteristic function.
Note that the set of creation operators generates a Grassmann algebra too because $\{\fd_j, \fd_k\}=\fd_j \fd_k + \fd_k \fd_j=0$.
This means that if we calculate $\exp((1/2)\sum_{j,k=1}^{L} A_{j,k} \fd_j \fd_k )$ from 
the definition and simplify all terms, we will make the same reordering as in the Grassmann number case, picking up the same sign differences.
Hence, \eqref{eq:PhiPF} is valid if we replace Grassmann numbers by the creation operators, which gives a physical interpretation to the characteristic function as normal-ordered operators.

\subsection{Fermionic Gaussian states}
Our statements will concern even Gaussian state vectors of the form $\ket{\psi_G} = U_G\vacket$.
This means that we fix the reference state vector $\vacket$ to be the vacuum and use a Gaussian unitary that has an even and quadratic generator 
\begin{equation}
H(G)=\frac{\i\,}{4} \sum_{j,k=1}^{2L} G_{j,k} \m_j \m_k
\end{equation}
where $G=-G^\top \in \mathbb R^{2L\times 2K}$ and the $\{\m_j\}$ are the self-adjoint Majorana operators which satisfy the Clifford relations $\{\m_j,\m_k\}=2\delta_{j,k}\id$.
Alternatively, they can be defined via Jordan-Wigner transformation or by relating them to the canonically anticommuting creation-annihilation operators through $\m_{2j-1}=\fe_j+\fd_j$ and $\m_{2j}   =-\i  (\fe_j-\fd_j)$.
As we are using even generators, the resulting state will conform to the fermionic parity superselection rule.
Note that as reference state one could in principle use a Fock-basis state with an odd particle content and again use a Gaussian unitary, however either choice results in the same physics and can be related by an appropriate particle-hole redefinition.

The Bloch-Messiah reduction \cite{bloch1962canonical, kraus2009quantum} employs a unitary transformation $U_1$ of the single particle basis to decouple a given Gaussian state vector as 
\begin{align}
\ket {\psi_G}=U_1\prod_{k=1}^{L/2}(v_{k}+u_{k}\fd_{2k-1}\fd_{2k})\vacket 
\end{align}
where the coefficients $v, u$ depend on $G$.
A particle number preserving (PNP) Gaussian transformation is a unitary of the form $U_1=e^{-\i H_1}$ where $H_1=\sum_{j,k=1}^L h_{j,k} \fd_j \fe_k$ with $h=h^\dagger \in \mathbb C^{L\times L}$.
For the above defined Gaussian states we have the following result \cite{bloch1962canonical, kraus2009quantum}.
\begin{lemma}[Bloch-Messiah reduction]
There exists a PNP transformation $U_1$ decoupling the state vector $\ket {\psi_G}$ in the following sense 
\begin{align}
\ket {\psi_G}=\prod_{k=1}^{\lfloor L/2\rfloor}(v_{k}\id+u_{k}p^\dagger_{2k-1}p^\dagger_{2k})\vacket.
\label{eq:BlochMessiah}
\end{align}
Here, $v,u\in \mathbb C^{\times \lfloor L/2\rfloor}$ and $p_j = U_1 \fe_j U_1^\dagger$.
\end{lemma}
\begin{proof}
We recapitulate the idea of the proof that can be found in the Appendix A of Ref.\ \cite{kraus2009quantum} or in 
Ref.\ \cite{bloch1962canonical}.
The main idea is to show that there exists a unitary transformation that puts into the normal form both the coherent hopping correlations $C_{j,k}=\langle \fd_j \fe_k\rangle$ and pairing terms $P_{j,k}=\langle \fe_j \fe_k\rangle$.
A relation first derived by Bogoliubov between $C$ and $P$ shows that both matrices can be put into a normal form simultaneously by a unitary transformation $U \in\mathbb C^{L\times L}$ \cite{bloch1962canonical, kraus2009quantum}.
For $C$ due to hermiticity we seek a diagonal form, while for $P$ a block diagonal form.
The transformation $U$ that achieves this can be viewed as a representation in mode space of a PNP transformation $U_1$ which defines the special mode operators $p_k=U_1 \fe_k U_1^\dagger= \sum_{k'}U_{k,k'}\fe_k$ where the correlation matrices are particularly simple.
The state vector $\ket {\psi_G}$ expressed in the $\{p_k\}$ takes the particularly simple form
\begin{align}
  \ket \psi = \prod_{k=1}^{L/2}(v_{2k-1}\id+u_{2k}p^\dagger_{2k-1}p^\dagger_{2k})\vacket\;,
\end{align}
where $|v_k|^2+|u_k|^2=1$ and we use the invariance of the vacuum under PNP transformations.
\end{proof}

This means that normal-ordering of a generic $U_G$ with respect to the vacuum is again Gaussian, i.e., it is a quadratic operator exponential which now only contains creation operators
\begin{align}
\ket \psi = \nUG \vacket =\sqrt Z^{-1} \exp\biggl(\frac{1}{2}\sum_{j,k=1}^{L} A_{j,k} \fd_j \fd_k\biggr)\vacket\;.
\end{align}
Indeed, let us for now assume that $v_k\neq 0$.
We now use $f^{\dagger 2}=0$ to write \eqref{eq:BlochMessiah} as
\begin{align}
  \ket \psi& =   \prod_{k=1}^{L/2}v_{2k-1}e^{\frac{u_{2k}}{v_{2k-1}}p^\dagger_{2k-1}p^\dagger_{2k}}\vacket\\
           & =  \sqrt Z^{-1}  e^{\sum_{k=1}^{L/2}\frac{u_{2k}}{2v_{2k-1}}[ p^\dagger_{2k-1}, p^\dagger_{2k}]}\vacket\;,
\end{align}
where we have defined 
\begin{equation}
\sqrt Z^{-1}:=\prod_{k=1}^{L/2}v_{2k-1}.
\end{equation}
In the next step, we exploit the result that any complex anti-symmetric matrix $\tilde A$ can be put into a normal form by a unitary conjugation, i.e., $\tilde A=U N(\lambda) U^\dagger$ where 
\begin{equation}
N(\lambda)=\bigoplus_k \begin{pmatrix} 0& \lambda_k \\-\lambda_k &0 \end{pmatrix} 
\end{equation}
when $L$ is even, with complex $\{\lambda_k\}$ \cite{bloch1962canonical, kraus2009quantum}.
If $L$ is odd, it takes the same form, with an additional $1\times 1$ block of $0$'s.
Thus, if $U_1$ acts as $U$ on modes we will have for $p_k=U_1 \fe_k U_1^\dagger= \sum_{k'}U_{k,k'}\fe_k$
\begin{align}
\ket \psi =\sqrt Z^{-1} \exp\biggl(\frac{1}{2}\sum_{j,k=1}^{L} A_{j,k} \fd_j \fd_k\biggr)\vacket=  \nUG\vacket\;,
\end{align}
where 
\begin{align}	\lambda_k =\frac{u_{2k}}{2v_{2k-1}}.
\end{align}
Clearly, by tuning $u,v$ we can reach any spectrum so any complex anti-symmetric matrix $A$ can be obtained.
For generic states this is enough, while basis states can be expressed as a limit of such expressions.

Finally, we observe that the algebra of creation operators alone is isomorphic to that of Grassmann by replacing $\theta_j$ by $\fd_j$. 
Then we can rewrite
\begin{align}
\label{EQ_PSI_CREATION_OP}
\ket \psi &= \sum_{j\in\BCL} \mathcal T(j) (\fd_1)^{j_1}\ldots (\fd_L)^{j_L}\vacket 
\end{align}
where $\mathcal T(j) =\sqrt Z^{-1} \Pf{A_{|j}}$ which is by definition a matchgate tensor.
Conversely, one can go the other direction defining as state using a matchgate tensor which will be a fermionic Gaussian state.
An alternative derivation of this fact based on a generalized Wick's theorem \cite{bravyi2016complexity} is given below.

\subsection{Alternative proof for matchgates corresponding to Gaussian states}

We express a fermionic Gaussian state vector in the occupation basis
\begin{align}
\ket \psi=\sum_{x\in\{0,1\}^n}\mathcal T(x)\ket x, 
\end{align}
where we have defined the \emph{amplitude tensor} 
$\mathcal T$. Its components can be obtained via $\mathcal T(x)=\braket x {\psi}$ where $\ket x=(\fd_1)^{x_1}\ldots (\fd_L)^{x_L}\vacket$ denotes an occupation basis state.
Note, that identifying $\vacket = {\ket\downarrow}^{\otimes L}$ and using the Jordan-Wigner transformation as in the main-text, this ordering of creation operators will yield $\ket x = (\sigma_1^+)^{x_1}\ldots(\sigma_L^+)^{x_L}\ket\downarrow^{\otimes L}= \otimes_{j=1}^L \ket{x_j}$.
Fixing a basis state $\ket z$, we may define the $z$-offset basis $\ket x_z = (\m_{2L-1})^{x_L} \ldots (\m_1)^{x_1}\ket z$.
We would like to show that $\mathcal T$ is a matchgate tensor by making use of the following result
\cite{bravyi2016complexity}.

\begin{lemma}[Generalized Wick's theorem]
For two Gaussian state vectors $\ket {\phi_1},\ket{\phi_2}$ and corresponding covariance matrices $M_1$, $M_2$, we have the generalized Wick's theorem
\begin{align}
\label{eq:general_wick_thm}
  \bra{\phi_1} (\m_1)^{x_1} \ldots (\m_{2L})^{x_{2L}}\ket{\phi_2}=\langle{\phi_1}\vert{\phi_2}\rangle\Pf{\i\Delta_{|x}}
\end{align}
where 
\begin{align}
\Delta=(-2\id +\i M_1-\i M_2)(M_1+M_2)^{-1}\;.
\end{align}
\end{lemma}
This general result should also be useful in various settings, in particular for studying non-Gaussian states with the methods of fermionic linear optics as it allows to calculate observables in linear combinations of pure Gaussian states.
Here is a first possible application.
\begin{lemma}[Matchgates and Gaussian states]
For a Gaussian state vector $\ket {\psi}= U_G \vacket$ define the $z$-offset such that $\braket z\psi\neq0$.
Then the amplitude tensor $\mathcal T$ in the $z$-offset basis representation $\ket \psi=\sum_{x\in\{0,1\}^L}\mathcal T(x)\ket x_z$ is a matchgate tensor.
\end{lemma}
\begin{proof}
We define $\mathcal T(x) =~_z\braket x \psi$.
We want to show that there exists an anti-symmetric matrix $A$ such that we have $\mathcal T(x)=\mathcal T(\overline 0) \Pf{A_{| x}}$ for all $x$.
Let us observe that $\ket z$ is Gaussian and we can obtain the components via  $\mathcal T(x)=\bra z (\m_1)^{x_1} \ldots (\m_{2L-1})^{x_L} \ket\psi$.
Next, we use the generalized Wick's theorem as stated above for the Gaussian state vectors $\ket z$ and $\ket \psi$ which gives us $\mathcal T(x)= \braket z \psi \Pf{\i\Delta_{|\overline x}}=\mathcal T(\overline 0)  \Pf{\i\Delta_{|\overline x}}$ where $\overline x = x \otimes (1,0)^\top $ and therefore $\mathcal T$ is a matchgate tensor because we have $A=\i\Delta_{|\overline 1\otimes (1,0)^\top}$ such that $\mathcal T(x)=\mathcal T(\overline 0) \Pf{A_{|x}}$.
\end{proof}

\section{Conversion of generating matrices to covariance matrices}
\label{sec:app_conversion}

We now convert generating matrices to their corresponding covariance matrix, the entries of which are
given by
\begin{align}
\label{def:cov_mat}
\Gamma_{j,k}(\psi) =\bra{\psi}\tfrac{\i\,}{2} [\m_j,\m_k]\ket \psi .
\end{align}
This can be seen as the inverse procedure to the Bloch-Messiah reduction that was used above to calculate the normal ordering.
Again, the calculation is based on the normal form of anti-symmetric matrices: We use the fact that any anti-symmetric matrix $A\in \mathbb C^{L\times L}$ can be put into a normal form  $A=W^\top \Sigma W$ where $W\in O(L)$ and $\Sigma$ is block-diagonal consisting of $2\times 2$ blocks of the form 
\begin{align}
\tfrac12\begin{pmatrix} 0& \lambda_k \\-\lambda_k &0 \end{pmatrix}
\end{align} for $L$ even and additionally a $0$ block  if it is odd.
If $A$ is real, it is easy to find $W$ from the eigenvectors of the hermitian matrix $\i A$, while in the general case the appendix of Ref.\ \cite{bloch1962canonical} provides implicitly a possible algorithm.
In the following we will prove a conversion formula for the case of \emph{real} anti-symmetric generating matrices, as the general case is not necessary for the main-text results.
In this case $\lambda_k$'s are real and we will assume that $W$ is such that $\lambda_k>0$ without loss of generality.
Using this convention we define a set of angles $\phi_k$ by identifying
\begin{align}\label{eq:angles}
\cos(\phi_k) := 1/(1+\lambda_k^2 )^{1/2} \quad \text{and}\quad \sin(\phi_k):= \lambda_k/(1+\lambda_k^2 )^{1/2} .
\end{align}
With these definitions, we state the following conversion lemma.
\begin{lemma}[$A \rightarrow \Gamma$ conversion]
Let $A=-A^\top\in\mathbb R^{L\times L}$ with normal form $A=W^\top\Sigma W$ as above.
Then the state vector $\ket \psi=\sqrt Z^{-1} \exp(\frac{1}{2}\sum_{j,k=1}^{L} A_{j,k} \fd_j \fd_k)\vacket$ has the covariance matrix
\begin{align}
  \Gamma(\psi)=\Xi  \tilde WV_\phi \Xi^{-1} \Gamma(\emptyset) (\Xi  \tilde W V_\phi\Xi^{-1})^\top
\end{align}
where $\tilde W = W\otimes \id_2$, $V_\phi=\oplus_{k=1}^{L/2}(\cos(\phi_k)\id+\i\sin(\phi_k) \sigma^y\otimes\sigma^x)$ when $L$ is even, or append $\oplus \id_1$ if it is odd and $\Xi =\oplus_{k=1}^n \begin{pmatrix} 1&1\\-\i&\i \end{pmatrix}$.
\end{lemma}

\begin{proof}
\def\f{\widetilde f^\dagger}

Let $U_W$ be the Gaussian particle number preserving unitary that implements the $W$ action on the modes 
\begin{align}
\fe_j =U_W \fd_j U_W^\dagger=\sum_{k=1}^L W_{j,k} \fd_k.
\end{align}
This choice puts the quadratic form into the normal form because
\begin{align}
\sum_{j,k=1}^LA_{j,k}\fd_j \fd_k= \sum_{j,k=1}^L \sum_{j',k'=1}^L (W^\top)_{j,j'} \Sigma_{j',k'}W_{k',k}\fd_j \fd_k=\sum_{j',k'=1}^L  \Sigma_{j',k'} \ftd_{j'}\ftd_{k'}
\end{align}
which gives
\begin{align}
\ket \psi = \sqrt Z^{-1}e^{ \sum_{j,k=1}^L  A_{j,k}\fd_j \fd_k}\vacket =\sqrt Z^{-1}\prod_{k=1}^{L/2} e^{\lambda_k \ftd_{2k-1} \ftd_{2k}}\vacket =\sqrt Z^{-1}\prod_{k=1}^{L/2} (1+\lambda_k \ftd_{2k-1} \ftd_{2k})\vacket .
\end{align}
 For ease of notation all sums and products going up to $\lfloor L/ 2 \rfloor$ will be denoted with an $L/2$ upper limit.
From this form we can read off the normalization of $\ket \psi$
\begin{align}
\label{EQ_PSI_NORM}
\braket{\psi}\psi= Z^{-1} \prod_{k=1}^{L/2}(1+\lambda_k^2) =1 ,
\end{align}
i.e., $Z=\prod_{k=1}^{L/2}({1+\lambda_k^2})$.
Having found that the state is decoupled, for fixed $k$ each term can be promoted to a unitary with the same action on the vacuum $\vacket$.
Using the angles $\phi_k$ defined above through \eqref{eq:angles} we find that 
\begin{align}
\frac{1+\lambda_k \ftd_{2k-1} \ftd_{2k}}{\sqrt{1+\lambda_k^2}}\vacket=
(\cos \phi_k \id + \sin \phi_k \Omega_k)\vacket .
\end{align}
where we have defined 
\begin{align}
\Omega_k:=\ftd_{2k-1} \ftd_{2k}+ \ft_{2k-1} \ft_{2k} 
\end{align}
which is anti-symmetric and hence can be used as a generator for a unitary operator in Hilbert space.
By observing that $\Omega_k^2\vacket=-\vacket$ we find that $U_{\phi_k}=\cos \phi_k \id + \sin \phi_k \Omega_k=e^{\phi_k \Omega_k}$ 
and using that the $\{\Omega_k\}$ commute we finally arrive at
\begin{align}
\ket \psi= e^{\sum_{k=1}^{L/2} \phi_k \Omega_k } \vacket =: U_\phi \vacket \end{align}
which is an explicit (even) Gaussian rotation of the vacuum.
To summarize this decoupling step, we have decoupled the normal-form of the state with $U_W$ to the Bloch-Messiah form and found the Gaussian unitary $U_\phi$ that rotates the vacuum into the state vector $\ket \psi$.
Note, that this allows to split any Gaussian unitary $U_G$ acting on the vacuum into a particle number preserving part $U_W$ and a squeezing part $U_\phi$.
For Majorana operators $\tilde \m_{2k-1}= \ft_{k} +\ftd_{k}$ and $\tilde \m_{2k}=-\i(\ft_{k}-\ftd_{k})$, we define the matrix $\widetilde \Gamma$ with entries
\begin{equation}
\widetilde \Gamma_{j,k}(\psi) :=\bra{\psi}\tfrac{\i\,}{2} [\tilde \m_j,\tilde \m_k]\ket \psi.
\end{equation}
By noting that $\tilde \m_{2k-1}=\sum_{j=1}^{L}W_{k,j} \m_{2j-1}$ and $\tilde \m_{2k}=\sum_{j=1}^{L}W_{k,j} \m_{2j}$, 
we find  the relation $\widetilde \Gamma(\psi)= \tilde W \Gamma(\psi) \tilde W^t$ where 
\begin{equation}
\tilde W= W\otimes \begin{pmatrix} 1&0\\0&1 \end{pmatrix}.
\end{equation}
Denoting by \begin{equation}\Gamma(\emptyset)=\bigoplus_{k=1}^L\begin{pmatrix} 0&1\\-1&0 \end{pmatrix}\end{equation} the vacuum covariance matrix, it remains to show that $\widetilde \Gamma= W_\phi \Gamma(\emptyset) W_\phi^\top$ where $W_\phi$ is defined by $U_\phi^\dagger \tilde \m_j U_\phi =\sum_{k=1}^{2L}(W_\phi)_{j,k} \tilde \m_k$.
Indeed we find that for $\ft_k(\phi)=U_\phi^\dagger \ft_k U_\phi$ we have the block-decoupled equations of motion
\begin{align}
\partial_{\phi_k} 
\begin{bmatrix} \ft_{2k-1}(\phi)\\ \ftd_{2k-1}(\phi) \\ \ft_{2k}(\phi)\\ \ftd_{2k}(\phi) \end{bmatrix}=
\begin{bmatrix}
0&0&0&1\\
0&0&1&0\\
0&-1&0&0\\
-1&0&0&0
\end{bmatrix}
\begin{bmatrix} \ft_{2k-1}(\phi)\\ \ftd_{2k-1}(\phi) \\ \ft_{2k}(\phi)\\ \ftd_{2k}(\phi) \end{bmatrix}
=\i\sigma^y\otimes \sigma^x\begin{bmatrix} \ft_{2k-1}(\phi)\\ \ftd_{2k-1}(\phi) \\ \ft_{2k}(\phi)\\ \ftd_{2k} (\phi)\end{bmatrix}
\end{align}
and therefore
\begin{align}
\begin{bmatrix} \ft_{2k-1}(\phi)\\ \ftd_{2k-1}(\phi) \\ \ft_{2k}(\phi)\\ \ftd_{2k}(\phi) \end{bmatrix}=
\begin{bmatrix}
\cos(\phi_k)&0&0&\sin(\phi_k)\\
0&\cos(\phi_k)&\sin(\phi_k)&0\\
0&-\sin(\phi_k)&\cos(\phi_k)&0\\
-\sin(\phi_k)&0&0&\cos(\phi_k)
\end{bmatrix}
\begin{bmatrix} \ft_{2k-1}\\ \ftd_{2k-1} \\ \ft_{2k}\\ \ftd_{2k} \end{bmatrix}
=(\cos(\phi_k)\id+\sin(\phi_k) \i\sigma^y\otimes\sigma^x)\begin{bmatrix} \ft_{2k-1}\\ \ftd_{2k-1} \\ \ft_{2k}\\ \ftd_{2k} \end{bmatrix}.
\end{align}
We collect all such rotations to $V_\phi=\oplus_{k=1}^{L/2}(\cos(\phi_k)\id+\i\sin(\phi_k) \sigma^y\otimes\sigma^x)$ when $L$ is even, and append $\oplus\id_1$ if it is odd.
Using the relation  $\tilde m =\Xi \f$ with
\begin{align}
\Xi =\oplus_{k=1}^n \begin{pmatrix} 1&1\\-\i&\i \end{pmatrix}
\end{align}
we can switch between the vector of creation anihiliation operators and Majorana operators which gives $W_\phi=\Xi V_\phi\Xi^{-1}$, so we find
\begin{align}
  \Gamma(\psi)=\Xi  \tilde W V_\phi \Xi^{-1} \Gamma(\emptyset) (\Xi \tilde W V_\phi \Xi^{-1})^\top\ .
\end{align}
\end{proof}

\section{Contraction rules for generating matrices}
\label{sec:app_contr}
\subsection{Contracting two tensors}
Now, we explicitly show how contracting two tensors $U$ and $V$ into a tensor $W$ combines the generating matrices $A$ and $B$ of the two original tensors into a larger generating matrix $C$. We assume that all three tensors are even matchgates, and can thus be written as
\begin{align}
U\left( \Theta_U \right) &= c_U \exp\left( \frac{1}{2} \Theta_U^T  A  \Theta_U \right) \text{ ,} \\
V\left( \Theta_V \right) &= c_V \exp\left( \frac{1}{2} \Theta_V^T  B  \Theta_V \right) \text{ ,}\\
\label{EQ_W_GAUSSIAN_DEF}
W\left( \Theta_W \right) &= c_W \exp\left( \frac{1}{2} \Theta_W^T  C  \Theta_W \right) \text{ ,}
\end{align}
where we have defined the vectors $\Theta_U$, $\Theta_V$, and $\Theta_W$ of Grassmann variables $\theta_i$ as
\begin{align}
\Theta_U &:= (\theta_1,\dots,\theta_{d_U})^T \text{ ,} \\
\Theta_V &:= (\theta_{d_U+1},\dots,\theta_{d_U+d_V})^T \text{ ,} \\
\Theta_W &:= (\theta_1,\dots,\theta_{d_U-1},\theta_{d_U+2},\dots,\theta_{d_U+d_V})^T \text{ .}
\end{align}
Thus, $A$, $B$, and $C$ are $d_U \times d_U$, $d_V \times d_V$ and $(d_U+d_V-2) \times (d_U+d_V-2)$ matrices, respectively. All are anti-symmetric. Note that we want to trace out the degrees of freedom corresponding to the Grassmann variables $\theta_{d_U}$ and $\theta_{d_U+1}$, i.e., the last index of $U$ and the first index of $V$.
As we showed earlier, the contraction is equivalent to the Grassmann integration
\begin{align}
W(\Theta_W) &= \int \text{d}\theta_{d_U+1}\int \text{d}\theta_{d_U} \exp\left( \theta_{d_U} \theta_{d_U+1} \right) U(\Theta_U) V(\Theta_V) \\
&= c_U c_V \int \text{d}\theta_{d_U+1}\int \text{d}\theta_{d_U} \exp\left( \theta_{d_U} \theta_{d_U+1} + \frac{1}{2} \Theta_U^T  A  \Theta_U + \frac{1}{2} \Theta_V^T  B  \Theta_V \right).
\end{align}
Notice that we can easily factorize exponentials because binomial terms in Grassmann variables commute, thus making the Baker-Campbell-Hausdorff formula trivial. This also allows us the remove all terms independent of $\theta_{d_U}$ and $\theta_{d_U+1}$ from the integral
\begin{align}
W(\Theta_W) &= c_U c_V \exp\left( \sum_{i=1}^{d_U-2} \sum_{j=i+1}^{d_U-1} A_{i , j}  \theta_i \theta_j + \sum_{i=2}^{d_V-1} \sum_{j=i+1}^{d_V} B_{i , j}  \theta_{d_U+i} \theta_{d_U+j} \right) \nonumber \\
& \cdot \int \text{d}\theta_{d_U+1} \int \text{d}\theta_{d_U} \exp\left( \theta_{d_U} \theta_{d_U+1} +  \sum_{i=1}^{d_U-1} A_{i  , d_U}  \theta_i \theta_{d_U} + \sum_{j=2}^{d_V} B_{1, j}  \theta_{d_U+1} \theta_{d_U+j} \right).
\end{align}
The expansion of the integrand exponential is fairly simple, as all powers higher than two vanish,
according to
\begin{align}
\exp\left( \theta_{d_U} \theta_{d_U+1} +  \sum_{i=1}^{d_U-1} A_{i  , d_U}  \theta_i \theta_{d_U} + \sum_{j=2}^{d_V} B_{1, j}  \theta_{d_U+1} \theta_{d_U+j} \right) = 1 &+ \theta_{d_U} \theta_{d_U+1} + \sum_{i=1}^{d_U-1} A_{i  , d_U}  \theta_i \theta_{d_U} + \sum_{j=2}^{d_V} B_{1, j}  \theta_{d_U+1} \theta_{d_U+j} \nonumber\\
&+ \sum_{i=1}^{d_U-1} \sum_{j=2}^{d_V} A_{i  , d_U} B_{1, j}  \theta_i \theta_{d_U} \theta_{d_U+1} \theta_{d_U+j}.
\end{align}
Applying the integral leaves us with
\begin{align}
W(\Theta_W) &= c_U c_V \exp\left( \sum_{i=1}^{d_U-2} \sum_{j=i+1}^{d_U-1} A_{i , j}  \theta_i \theta_j + \sum_{i=2}^{d_V-1} \sum_{j=i+1}^{d_V} B_{i, j}  \theta_{d_U+i} \theta_{d_U+j} \right) 
\left( 1 + \sum_{i=1}^{d_U-1} \sum_{j=2}^{d_V} A_{i  , d_U} B_{1, j}  \theta_i \theta_{d_U+j} \right) \nonumber\\
&= c_U c_V \exp\left( \sum_{i=1}^{d_U-2} \sum_{j=i+1}^{d_U-1} A_{i , j}  \theta_i \theta_j + \sum_{i=2}^{d_V-1} \sum_{j=i+1}^{d_V} B_{i , j}  \theta_{d_U+i} \theta_{d_U+j} + \sum_{i=1}^{d_U-1} \sum_{j=2}^{d_V} A_{i  , d_U} B_{1, j}  \theta_i \theta_{d_U+j} \right) \text{ .}
\end{align}
We were able to turn the last factor into an exponential because higher powers of the contraction sum are zero, i.e.,
\begin{equation}
\left(\sum_{i=1}^{d_U-1} \sum_{j=2}^{d_V} A_{i  , d_U} B_{1, j}  \theta_i \theta_{d_U+j} \right)^2 = \left(\sum_{i=1}^{d_U-1} A_{i  , d_U} \theta_i \right)^2 \left(\sum_{j=2}^{d_V} B_{1, j} \theta_{d_U+j} \right)^2 = 0 \text{ ,}
\end{equation}
where we have used the fact that any linear combination of Grassmann numbers $\sum_i a_i \theta_i$ is again a Grassmann number, squaring to zero.
We can now write the explicit structure of $W(\Theta_W)$ in terms of $c_W$ and $C$ in \eqref{EQ_W_GAUSSIAN_DEF}. Obviously, $c_W = c_U c_V$. The matrix $C$ is composed of $A$ and $B$ according the  pattern
\begin{align}
\label{EQ_C_FROM_AB_CONTRACTED}
C &= 
\begin{pmatrix}
A_{1 , 1}                 & \cdots & A_{1  , d_U-1}               & A_{1  , d_U} B_{1, 2}     & \cdots & A_{1  , d_U} B_{1  , d_V}\\
\vdots                  & \ddots & \vdots                      & \vdots                  & \ddots & \vdots \\
-A_{1  , d_U-1}          & \cdots & A_{d_U-1  , d_U-1}           & A_{d_U-1  , d_U} B_{1, 2} & \cdots & A_{d_U-1  , d_U} 
B_{1  , d_V} \\
-A_{1  , d_U} B_{1 , 1}    & \cdots & -A_{d_U-1  , d_U} B_{1 , 1}    & B_{2 , 2}                 & \cdots & B_{2  , d_V} \\
\vdots                  & \ddots & \vdots                      & \vdots                  & \ddots & \vdots \\
-A_{1  , d_U} B_{1  , d_V}& \vdots & -A_{d_U-1  , d_U} B_{1  , d_V}& B_{d_V,  2}             & \cdots & B_{d_V  , d_V} \\
\end{pmatrix} \\
&= \begin{pmatrix}
A_\text{UL} & \begin{pmatrix} A_{1  , d_U} \\ \vdots \\ A_{d_U-1  , d_U} \end{pmatrix} \begin{pmatrix} B_{1, 2} & \cdots & B_{1  , d_V} \end{pmatrix} \\
-\begin{pmatrix} B_{1, 2} \\ \vdots \\ B_{1  , d_V} \end{pmatrix} \begin{pmatrix} A_{1  , d_U} & \cdots & A_{d_U-1  , d_U} \end{pmatrix} & B_\text{BR} 
\end{pmatrix} \text{ .}
\end{align}
The submatrices $A_\text{UL}$ and $B_\text{BR}$ are the upper-left and bottom-right part of the matrices $A$ and $B$, respectively, with one row and column removed. anti-symmetry of $A_\text{UL}$ and $B_\text{BR}$, and by extensions $C$, implies that all diagonal elements are zero.

This result is indeed quite natural, seen from a diagrammatic perspective, where the matrix  $C$ defines the 2-point correlators of the contracted state. 
The correlation between uncontracted Grassmann variables that lie either completely in $\Theta_U$ or $\Theta_V$ remains unaffected by the contraction. 
Correlators $C_{i ,j}$ between a $\theta_i$ in $\Theta_U$ and a $\theta_j$ in $\Theta_V$ are simply given by $A_{i  , d_U} B_{1  j}$, i.e., the product of the correlators over the contracted edge $(\theta_{d_U},\theta_{d_U+1})$. 

\subsection{Self-contractions}
Now, consider the more complicated case of self-contraction. We start with the tensor $T(\Theta)$ given by
\begin{equation}
T\left( \Theta \right) = c \exp\left( \frac{1}{2} \Theta^T  A  \Theta \right) = c \exp\left( \sum_{i=1}^{d-1} \sum_{j=i+1}^d A_{i , j} \theta_i \theta_j \right) \text{ ,}
\end{equation}
using the $d$ Grassmann variables $\Theta = (\theta_1,\dots,\theta_d)$. Without loss of generality (we can always perform index permutation), we want to contract the first two indices of $T$, i.e., contract over $\theta_1$ and $\theta_2$. Again writing the contraction as a Grassmann integration, we find
\begin{align}
T(\Theta)_{1 \star 2} &= c \int \text{d}\theta_2\int \text{d}\theta_1 \exp\left( \theta_1 \theta_2 + \sum_{i=1}^{d-1} \sum_{j=i+1}^d A_{i , j} \theta_i \theta_j \right) \nonumber\\
&= c \exp\left( \sum_{i=3}^{d-1} \sum_{j=i+1}^d A_{i , j} \theta_i \theta_j \right) \int \text{d}\theta_2\int \text{d}\theta_1 
\exp\left( (1+A_{1, 2}) \theta_1 \theta_2 + \sum_{j=3}^d \left(A_{1, j} \theta_1 \theta_j + A_{2, j} \theta_2 \theta_j \right) \right).
\end{align}
Again, we can expand the exponential explicitly, as all terms beyond second order vanish:
\begin{align}
\exp\left( (1+A_{1, 2}) \theta_1 \theta_2 + \sum_{j=3}^d \left(A_{1, j} \theta_1 \theta_j + A_{2, j} \theta_2 \theta_j \right) \right) = & 1 + (1+A_{1, 2}) \theta_1 \theta_2 + \sum_{j=3}^d \left(A_{1, j} \theta_1 \theta_j + A_{2, j} \theta_2 \theta_j \right) \nonumber\\
&+ \sum_{i=3}^{d-1} \sum_{j=i+1}^d (A_{i ,1} A_{2, j} - A_{i, 2} A_{1, j}) \theta_1 \theta_2 \theta_i \theta_j .
\end{align}
Only the second and fourth term survive the integration, giving us
\begin{align}
T(\Theta)_{1 \star 2} &= c \exp\left( \sum_{i=3}^{d-1} \sum_{j=i+1}^d A_{i , j} \theta_i \theta_j \right) \left( 1 + A_{1, 2} + \sum_{i=3}^{d-1} \sum_{j=i+1}^d (A_{i ,1} A_{2, j} - A_{i, 2} A_{1, j}) \theta_i \theta_j \right) \nonumber\\
&= c  (1 + A_{1, 2}) \exp\left( \sum_{i=3}^{d-1} \sum_{j=i+1}^d A_{i , j} \theta_i \theta_j \right) \exp \left( \sum_{i=3}^{d-1} \sum_{j=i+1}^d \frac{A_{i ,1} A_{2, j} - A_{i, 2} A_{1, j}}{1+A_{1, 2}} \theta_i \theta_j \right) \nonumber\\
&= c  (1 + A_{1, 2}) \exp\left( \sum_{i=3}^{d-1} \sum_{j=i+1}^d \left(A_{i , j} + \frac{A_{i ,1} A_{2, j} - A_{i, 2} A_{1, j}}{1+A_{1, 2}} \right) \theta_i \theta_j \right).
\end{align}
To get to the second line, we require that the square (and thus higher powers) of the $O(\theta_i \theta_j)$ term vanish:
\begin{align}
\left(\sum_{i=3}^{d-1} \sum_{j=i+1}^d (A_{i ,1} A_{2, j} - A_{i, 2} A_{1, j}) \theta_i \theta_j \right)^2 
&= \left( \sum_{i=3}^{d} \sum_{j=3}^d A_{i ,1} A_{2, j} \theta_i \theta_j \right)^2 \nonumber\\
&= \left( \sum_{i=3}^{d} A_{i ,1} \theta_i \right)^2 \left( \sum_{j=3}^d A_{2, j} \theta_j \right)^2 = 0.
\end{align}
Thus, we can express the contracted tensor in the form
\begin{equation}
T(\Theta)_{1 \star 2} = c_{1 \star 2} \exp\left( \frac{1}{2} \Theta_{1 \star 2}^T  A_{1 \star 2}   \Theta_{1 \star 2}  \right) \text{ ,}
\end{equation}
where $\Theta_{1 \star 2} = (\theta_3, \theta_4, \dots, \theta_d)$ contains the uncontracted Grassmann variables, and the constant $c_{1 \star 2}$ and $(d-2)\times(d-2)$ matrix $A_{1 \star 2}$ are given by the original constant $c$ and matrix $A$ according to
\begin{align}
c_{1 \star 2} &= (1 + A_{1, 2})  c \text{ ,}\\
(A_{1 \star 2})_{i , j} &= A_{i , j} + \frac{A_{i ,1} A_{2, j} - A_{i, 2} A_{1, j}}{1+A_{1, 2}}  \nonumber \\
&= \frac{A_{i , j} + A_{i , j} A_{1, 2} + A_{i ,1} A_{2, j} - A_{i, 2} A_{1, j}}{1+A_{1, 2}} \text{ .}
\end{align}

The self-contraction integrates out the $A_{1, 2}$ correlator, redefining our ``vacuum term'' $c$. $(A_{1 \star 2})_{i , j}$ now contains all connected and disconnected correlations between site $i$ and $j$, divided by the vacuum contributions.


\subsection{Cyclic permutations}

In order to contract smaller matchgate tensors into larger ones, we need one additional ingredient: rules for cyclic permutation. Our prescription for contracting two tensors $U$ and $V$ works by contracting the last index of $U$ with the first index $V$, while in the self-contraction case we contracted the first two indices of a tensor $T$. Clearly, we can satisfy both conditions by cyclically permuting the indices of the tensors in question, i.e., relabeling the Grassmann variables. We write a cyclic permutation by $n$ bits as $\sigma_n(\Theta)$, for example
\begin{equation}
\sigma_1\left( \theta_1 \theta_2 - \theta_2 \theta_3 \right) = \theta_1 \theta_3 + \theta_2 \theta_3 \text{ .}
\end{equation}
It is easy to see that the cyclic permutation of a Gaussian matchgate tensor $T(\Theta)$ is given by
\begin{align}
\sigma_n\left(T(\Theta)\right) = \sigma_n\left( c \exp\left( \frac{1}{2} \Theta^T  A  \Theta \right)\right) = c \exp\left( \frac{1}{2} \Theta^T  \sigma_n(A)  \Theta \right) \text{ ,}
\end{align}
where the new correlation matrix $\sigma_n(A)$ is simply $A$ where the $i$-th row and $j$-th colummn is replaced by the $(i+n)$th row and the $(j+n)$th column modulo $m$ (where $m$ is the length of the vector of Grassman variables $\Theta$).
With these rules for permutations, contractions and self-contractions, we can contract any planar network of Gaussian matchgate tensors. For odd tensors, where an integral over an additional source term of auxiliary Grassmann variables is required, the rules become significantly more complicated.

\subsection{Graph orientation and boundary conditions}

We now show how a complete network is contracted using the tools developed earlier. As a concrete example, consider the contraction of 11 pentagons (i.e., tensors with five indices) in a $\lbrace 5, 4 \rbrace$ tiling shown in Fig.\  \ref{FIG_CONTR_STEPS}. 
We start with an initial labeling of all pentagon edges in a clockwise orientation, with each index $i$ corresponding to an independent Grassmann variable $\theta_i$. 
Starting from the central tensor, we start contracting adjacent tensors, using cyclic permutations of the indices to ensure that the largest index of the first tensor is adjacent to the smallest index of the second tensor. This process can be easily repeated until a tensor with two adjacent edges is encountered. 
We then contract from the edge with the smaller index (in clockwise orientation), which leaves a protruding double-edge that can be removed through self-contraction.

\begin{figure}[htb]
\centering
\includegraphics[width=0.25\textwidth]{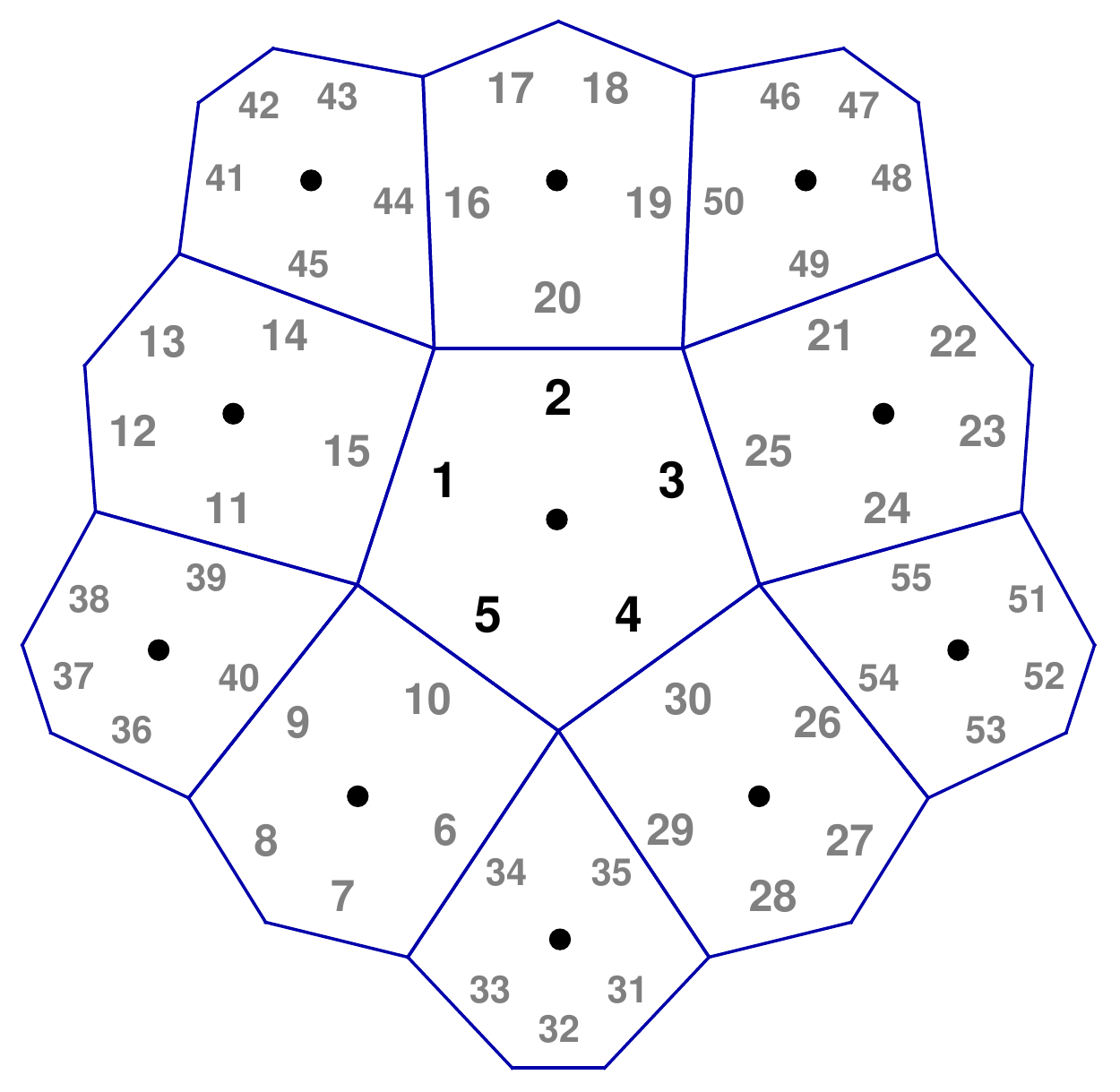}
\includegraphics[width=0.25\textwidth]{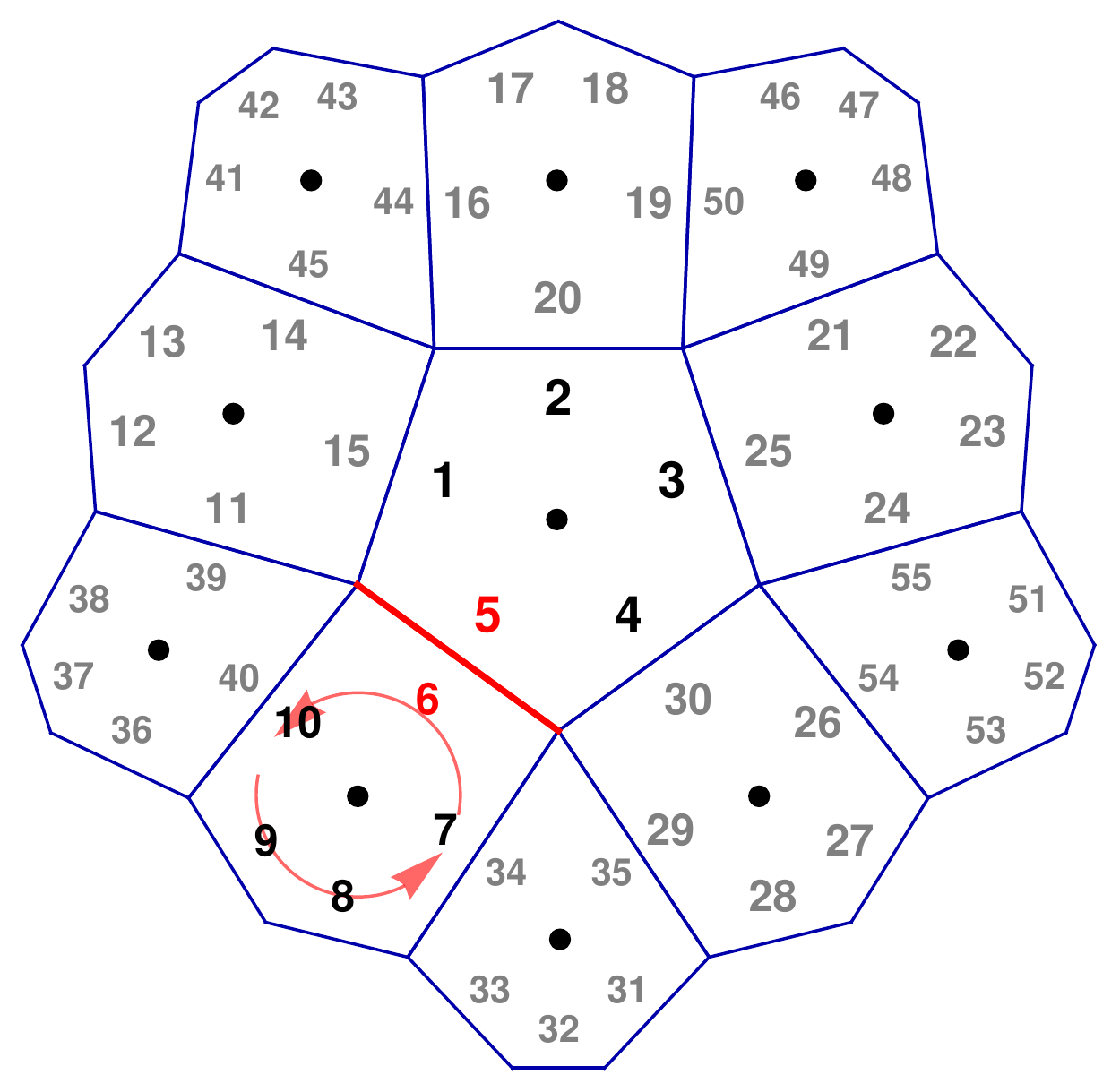}
\includegraphics[width=0.25\textwidth]{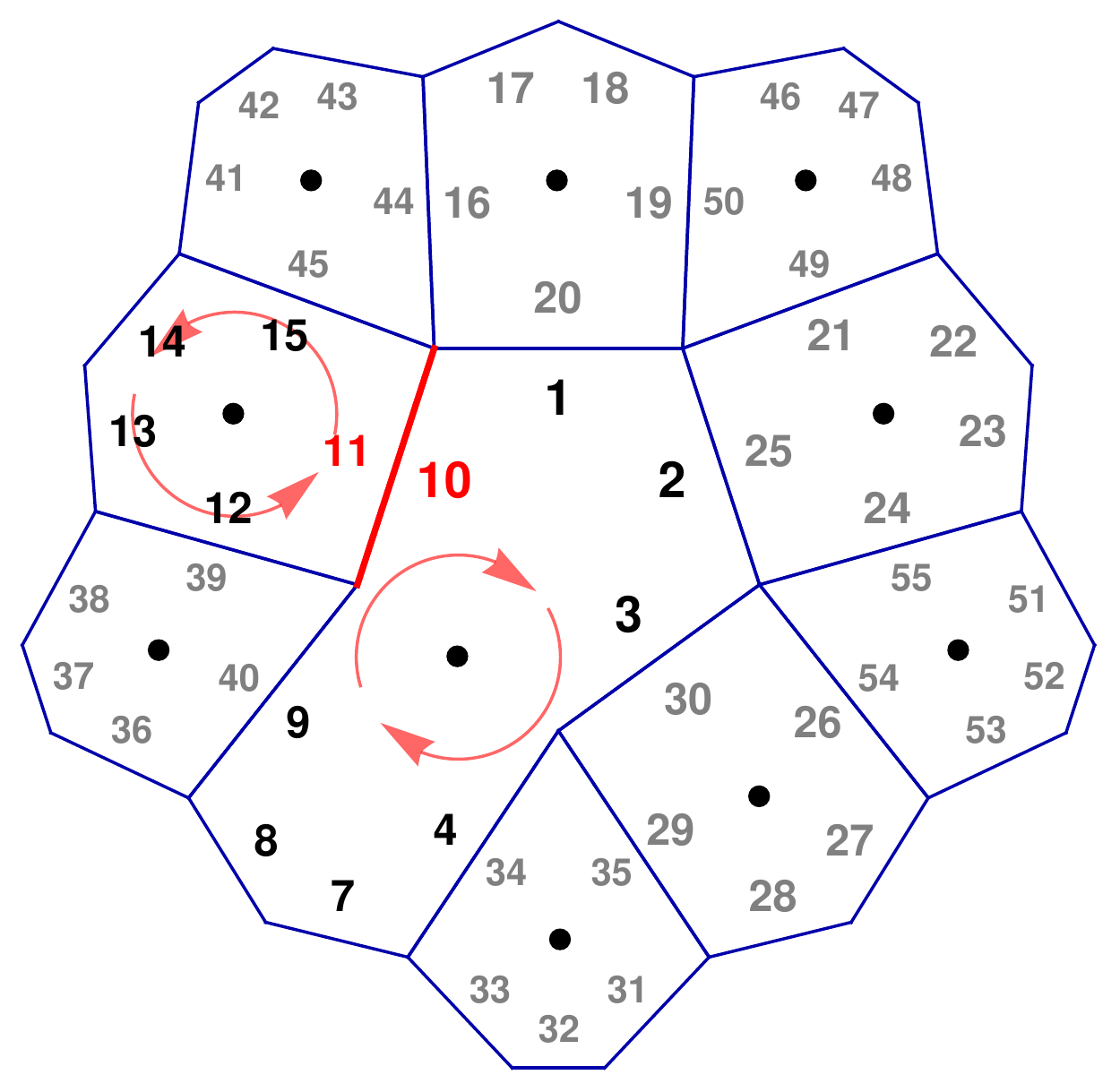}

\includegraphics[width=0.25\textwidth]{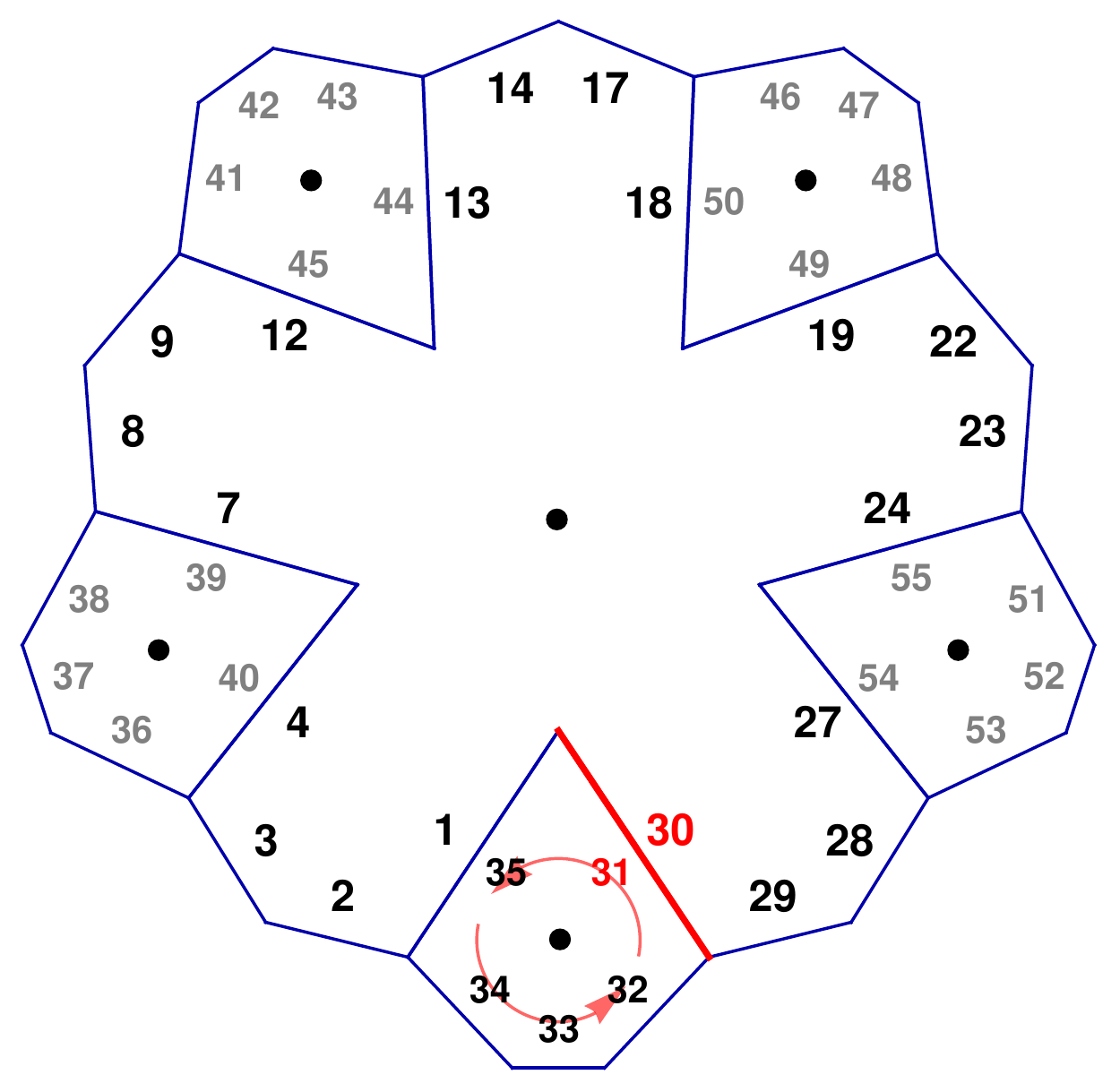}
\includegraphics[width=0.25\textwidth]{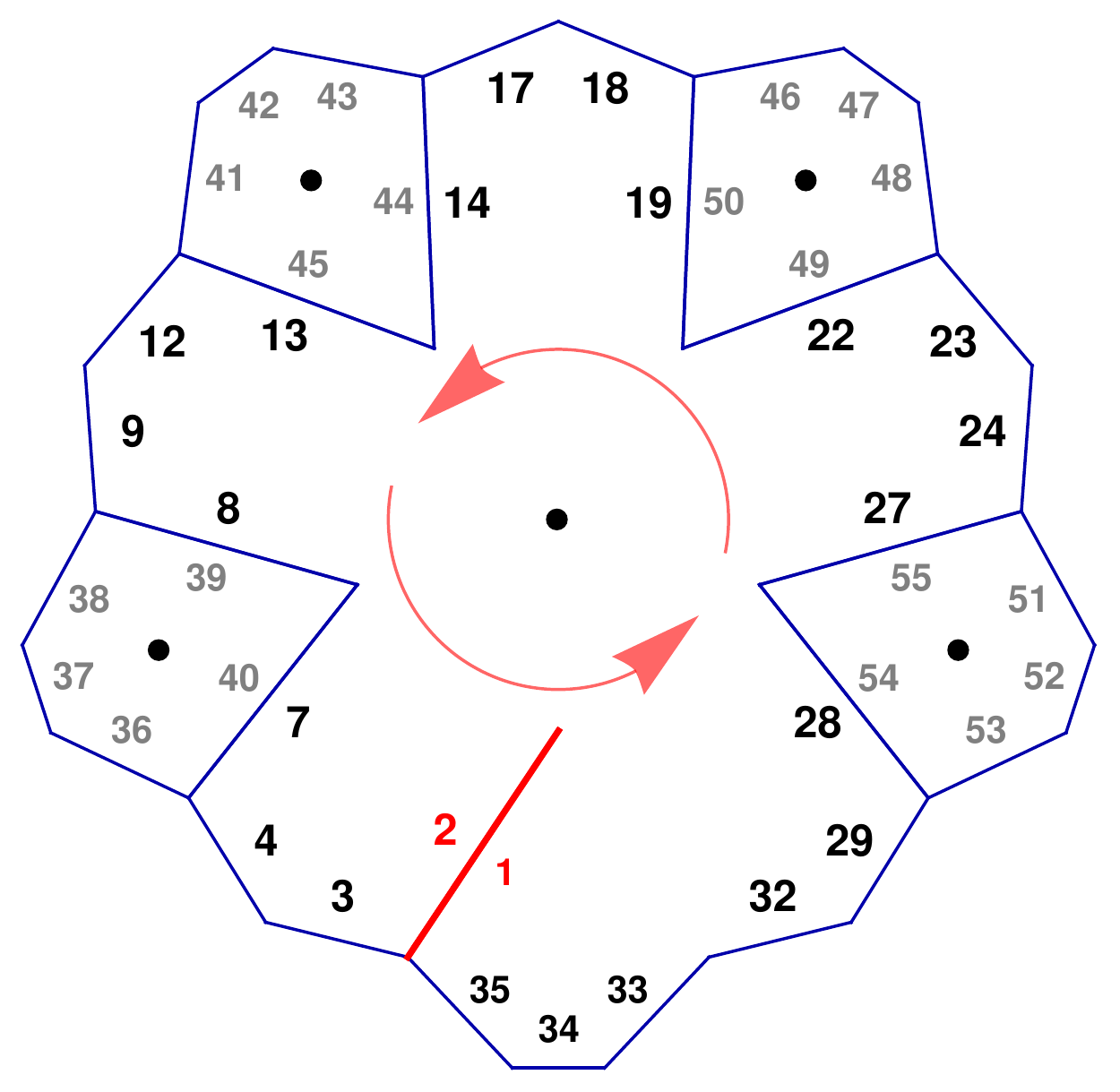}
\includegraphics[width=0.25\textwidth]{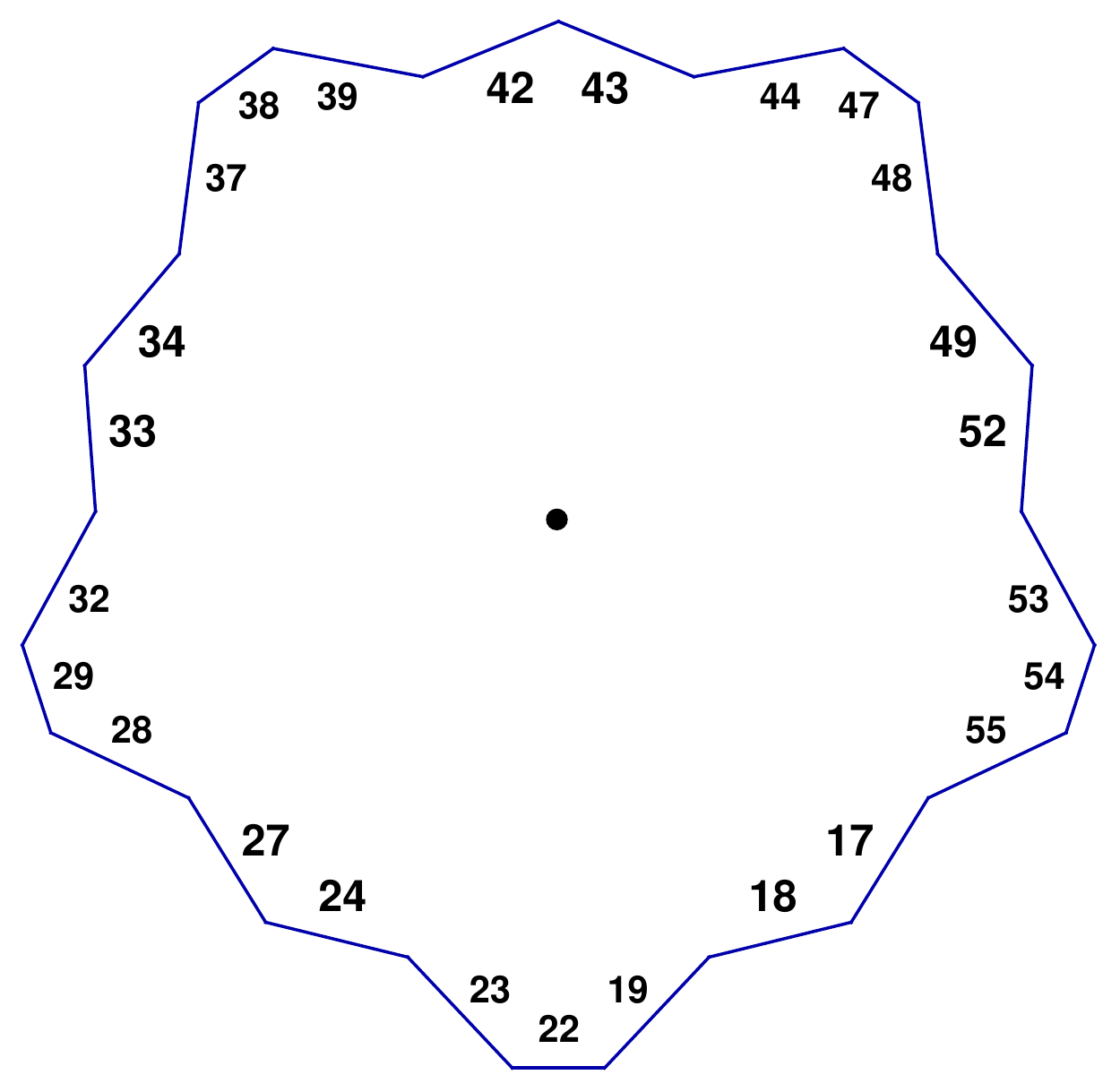}

\caption{Schematic contraction of 11 pentagons, using contractions over joint edges, self-contractions and cyclic permutations of indices. Contracted edges and indices are marked in red, while cyclic permutations are shown by circular arrows.}
\label{FIG_CONTR_STEPS}
\end{figure}

After all tensors are contracted, we are left with a boundary of edges whose indices are still clockwise-oriented, provided that our inital labeling of pentagon edges followed the order in which we contracted the respective tensors. We can think of the remaining indices as specifying our boundary sites, and the contracted indices being bulk sites that were ``integrated out''. 
Note that this contraction process is possible for planar graphs for which a Kasteleyn orientation is guaranteed to exist \cite{Kasteleyn1961}.

In principle, we have the freedom to cyclically permute the indices of each inital tensor. We can fix this freedom by using symmetry constraints on the generating matrix $A_{i,j}$. For anti-periodic boundary conditions, we require $A_{i,j}$ to be positive for $i>j$ and negative for $i<j$. The condition is retained for the full contracted state if we restrict ourselves to applying cyclic permutations only on the indices of the ``inner'' tensor from which we contract outwards, and affix the lowest index of each ``outer tensor'' to the edge over which it is first contracted. These conditions on $A_{i,j}$ allowed us to produce physical covariance matrices with $\Gamma_{i,j}>0$ for $i>j$, as well.

For periodic boundary conditions, $A_{i,j}$ should be positive for $|i-j| < L/2$ (with number of indices $L$) and negative otherwise. This can be achieved using the same index labeling rules, but choosing \textsl{only} the central tensor's generating matrix to produce a locally anti-periodic state, while keeping the states corresponding to all other local tensors as periodic.

Periodic boundary conditions are less convenient for numerical studies, as overlap between positive and negative correlations can occur for networks of finite size. For this reason, we have focused on anti-periodic boundary conditions in our work. In the infinite-size limit, of course, both choices of boundary conditions should lead to the same physical properties of the boundary states.

\section{Explicit generating matrices and numerical results}
\label{sec:app_explicit}
\subsection{Regular tilings}

We start by discussing the correlations achieved for regular tilings.
We can produce the boundary theory with $\sim 1/d$ falloff using a regular $\lbrace 3,k \rbrace$ tiling with $k \ge 6$.
The critical parameter $a=a_\text{crit}$ for each $k$ can be found by maximizing mean long-range correlations 
$({2}/{L} )\sum_{k=1}^{L/2} | \Gamma_{k,k+L/2} |$ in the covariance matrix. 
For $k>6$, the tilings can be embedded into the Poincar\'e disk with metric
\begin{equation}
\text{d}s^2 = 4\, \frac{\text{d}r^2 + r^2 \text{d}\phi^2}{(1-r^2)^2} \text{ ,}
\end{equation}
using polar coordinates $(r,\phi)$ with $0 \leq r < 1$ and $0 \leq \phi < 2 \pi$. 
As the Poincar\'e disk represents an infinite volume, we cut off our tilings at a radius $r=r_c$. For the flat case $k=6$, we simply cut off at Euclidean distance $d_c$ (with all edges set to unit length). We find the following values of $a_\text{crit}$ for a given $d$ ($k=6$) or $r_c$ ($k>6$):
\begin{center}
\begin{tabular}{@{\hspace{0.5cm}} c @{\hspace{0.5cm}} | @{\hspace{0.5cm}} c @{\hspace{0.5cm}} || @{\hspace{0.5cm}} c  @{\hspace{0.5cm}} | @{\hspace{0.5cm}} c @{\hspace{1cm}} c @{\hspace{1cm}} c @{\hspace{1cm}} c @{\hspace{1cm}} c @{\hspace{1cm}} c @{\hspace{0.5cm}}}
$d_c$ & $k=6$: & $r_c$  & $k=7$: & $k=8$: & $k=9$: & $k=10$: & $k=11$: & $k=12$: \\
 \hline\hline
10 & 0.5779 & 0.95 & 0.6063 & 0.6213 & 0.6378 & 0.6482 & 0.6529 & 0.6650 \\
15 & 0.5804 & 0.98 & 0.6051 & 0.6203 & 0.6404 & 0.6447 & 0.6618 & 0.6639 \\
20 & 0.5806 & 0.99 & 0.6082 & 0.6244 & 0.6393 & 0.6502 & 0.6575 & 0.6661 
\end{tabular}
\end{center}

Note that increasing $k$ leads to a larger $a_c$.
We argue that this may compensate for the ``leaking'' of correlations into the higher-curvature bulk. 
While in principle it is possible to extend this reasoning to positive-curvature (spherical) tilings, the largest triangular tiling $\lbrace 3, 5 \rbrace$ corresponds to an icosahedron with only 20 triangles. Thus, no proper choice of an asymptotic boundary can be made.

\subsection{MERA}
We now turn to discussing how the MERA framework can be related to our approach.
The MERA tensor network consists of two types of tensors, isometries and disentanglers with three and four legs, respectively. Thus, the lattice for the equivalent matchgate tensor consists of triangles and quadrilaterals. In the matchgate setting, the MERA tensors are thus fully specified by a $3 \times 3$ generating matrix $S$ and a $4 \times 4$ matrix $B$, corresponding to isometries and disentanglers, respectively. 
For norm-preserving tensors, i.e.\ unitary disentanglers and isometries, real generating matrices are restricted to the
components
\begin{align} 
A &= \left(
\begin{array}{ccc}
\msp 0 &\msp \sqrt{1 + x^2} \cos\theta &\msp \sqrt{1 + x^2} \sin\theta \\ 
    -\sqrt{1 + x^2} \cos\theta &\msp 0 &\msp x \\
    -\sqrt{1 + x^2} \sin\theta &    -x &\msp 0 
\end{array}
\right) \text{ , }\\
B &= \left(
\begin{array}{cccc}
\msp 0 &\msp y &\msp \sqrt{1 + y^2} \cos\phi &\msp \sqrt{1 + y^2} \sin\phi \\
    -y &\msp 0 & - \sqrt{1 + y^2} \sin\phi &\msp \sqrt{1 + y^2} \cos\phi \\
    -\sqrt{1 + y^2} \cos\phi & \msp \sqrt{1 + y^2} \sin\phi &\msp 0 &\msp y \\
    -\sqrt{1 + y^2} \sin\phi & -\sqrt{1 + y^2} \cos\phi  &  -y &\msp 0	
\end{array}
\right) \text{ ,}
\end{align}
with $x,y \in \mathbb{R}$ and $\theta,\phi \in [0,2\pi]$.
These free parameters of the model can be set by numerically minimizing the ground-state energy of the translation-invariant Ising Hamiltonian \eqref{EQ_ISING_H}.
However, with these inputs we were unable to find boundary states that are any more translation-invariant than the regular tilings considered earlier. Instead, we consider a more generic ``matchgate MERA'' (mMERA) with three- and four-leg generating matrices 
\begin{align} 
\label{EQ_MERA_GEN_MAT}
A = \left(
\begin{array}{ccc}
\msp 0 &\msp a &\msp a \\  
    -a &\msp 0 &\msp b \\ 
    -a &    -b &\msp 0  
\end{array}
\right) \text{ , }\;
B = \left(
\begin{array}{cccc}
\msp 0 &\msp c &\msp e &\msp f \\ 
    -c &\msp 0 &\msp d &\msp e \\ 
    -e &    -d &\msp 0 &\msp c \\ 
    -f &    -e &    -c &\msp 0  
\end{array}
\right) \text{ ,}
\end{align}
with the parameters $a,b,\dots,f \in \mathbb{R}$. Now, again minimizing according to \eqref{EQ_ISING_H}, we find that numerical solutions obey the symmetries $c \approx e$ and $a \approx d \approx f$, thus leaving us with three free parameters to optimize. Intriguingly, these symmetries allow us to express the 4-leg ``disentanglers'' as contractions of a 3-leg tensor with its conjugate, visualized in Fig.\ \ref{FIG_MERA_TN_TO_MG}. 
While the individual tensors of our model are no longer norm-preserving, we show in the next section that for large networks, norm preservation can still be achieved. Note that while the usual MERA identities for isometries and disentanglers no longer hold, contractions of the mMERA are still efficient, owing to the matchgate setting.

\begin{figure}[htb]
\centering
\includegraphics[width=0.7\textwidth]{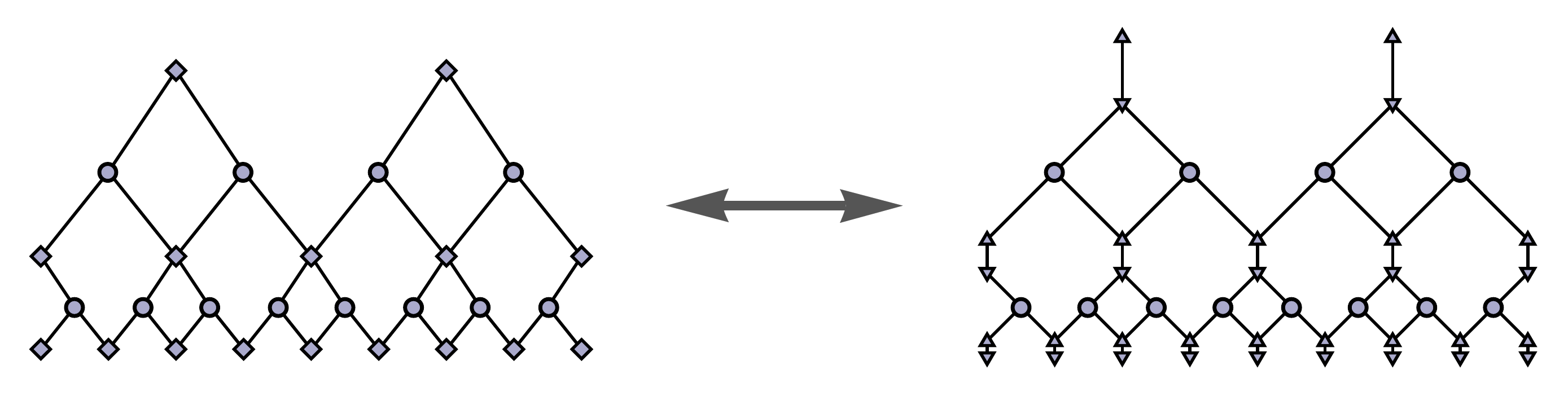}

\includegraphics[width=0.3\textwidth]{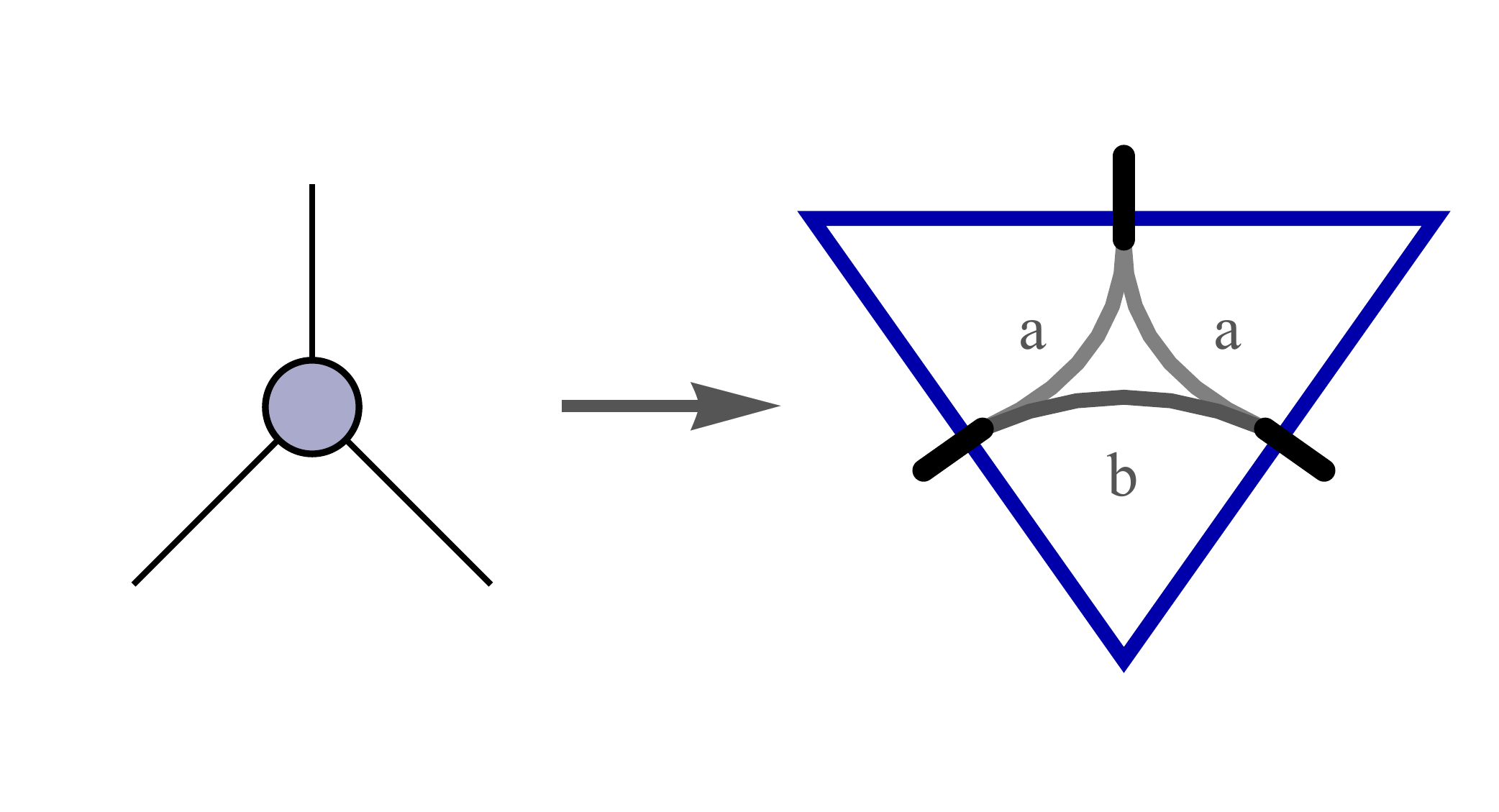}
\hspace{0.04\textwidth}
\includegraphics[width=0.3\textwidth]{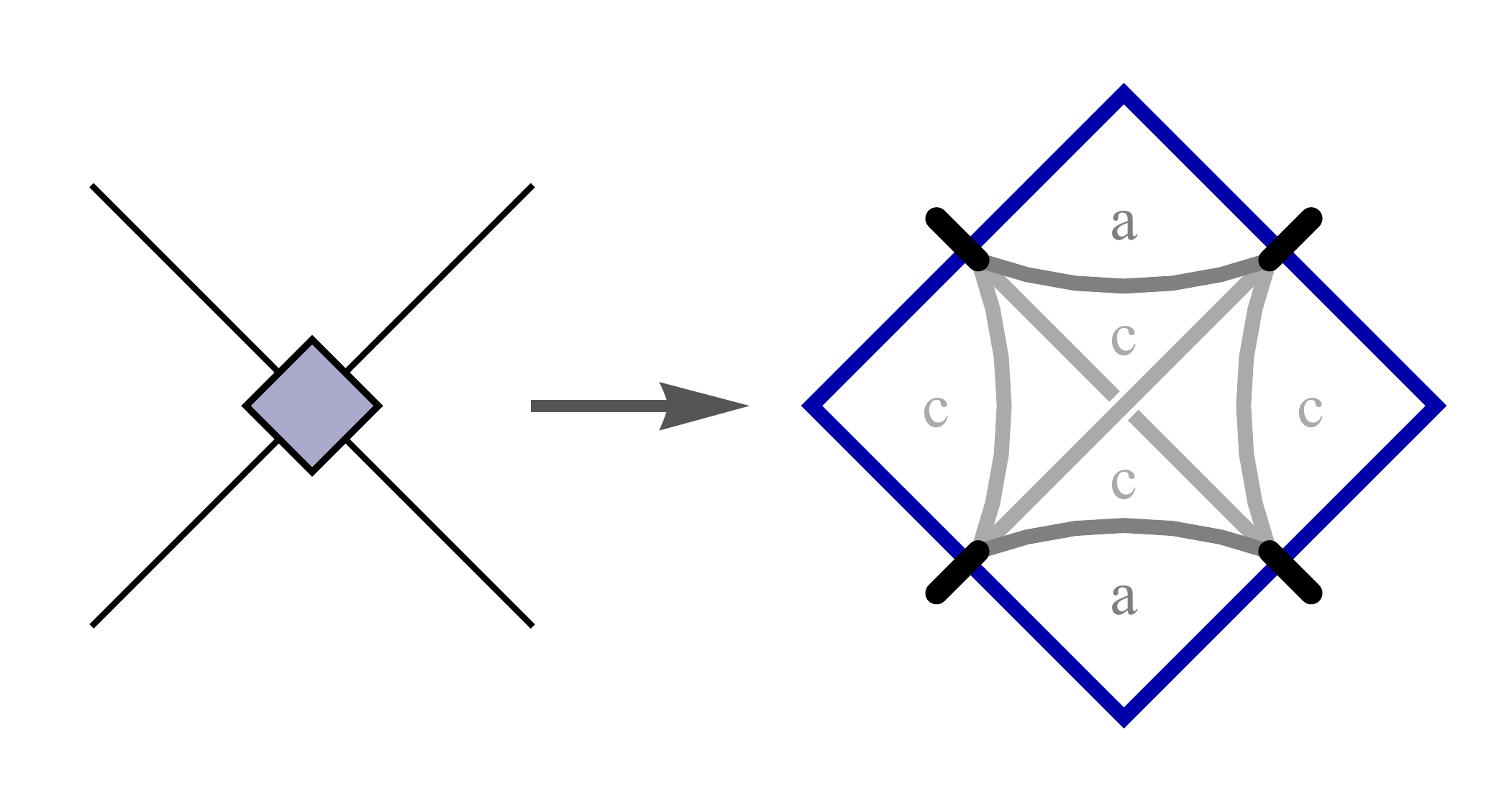}
\hspace{0.04\textwidth}
\includegraphics[width=0.3\textwidth]{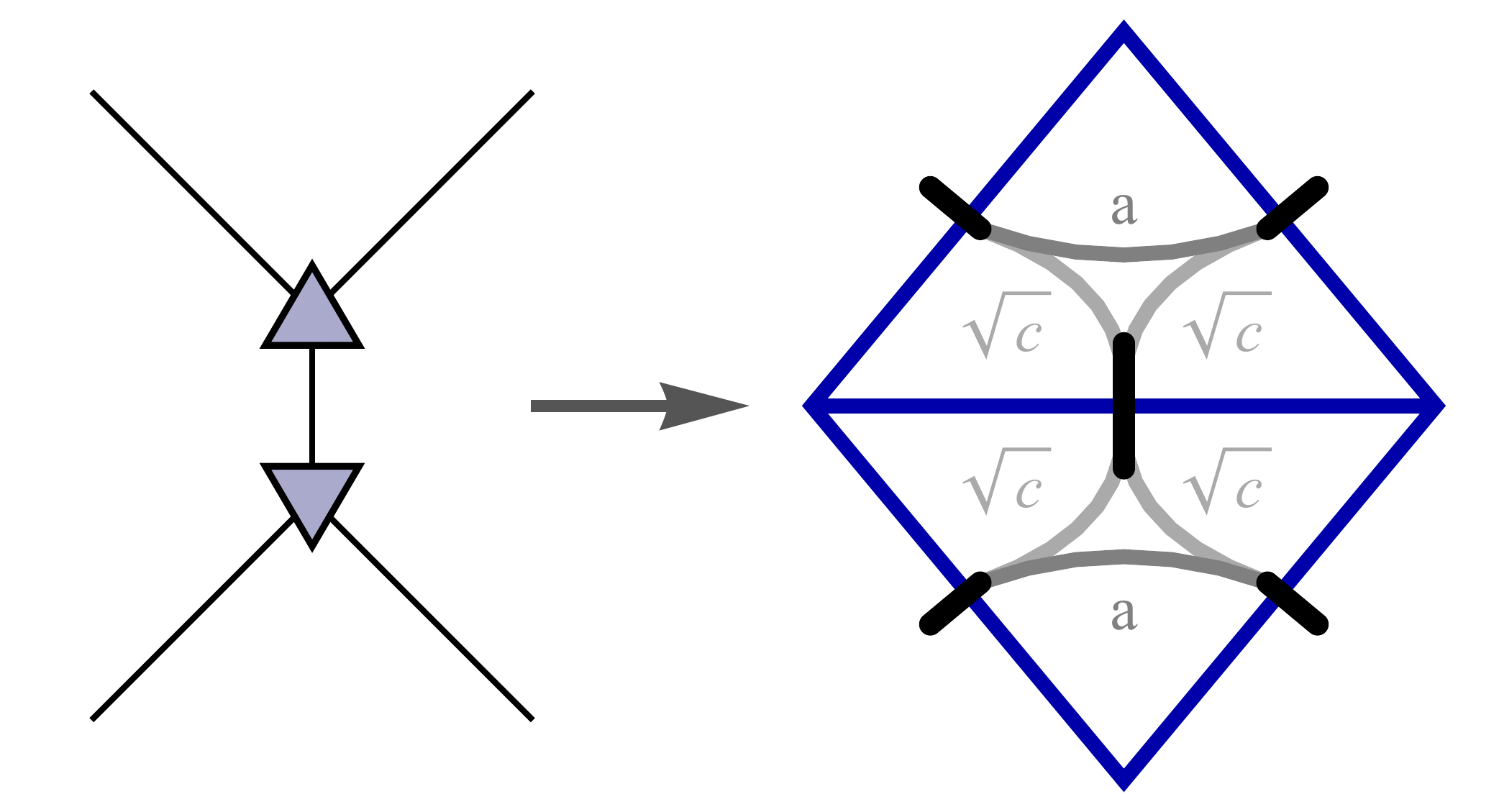}

\caption{\textsc{Top}: The standard MERA tensor network (left) in our numerical matchgate setting is equivalent to a network of purely 3-leg tensors (right). \textsc{Bottom}: Isometries, disentanglers, and triangulated disentanglers (from left to right) expressed as matchgate tensors. The free parameters $a,b,c$ fix the components of the generating matrices \eqref{EQ_MERA_GEN_MAT}.
}
\label{FIG_MERA_TN_TO_MG}
\end{figure}

The central 4-leg tensor of the MERA describes a CFT ground state on four sites, for which the generating matrix $A_0$ and normalization $c_0$ can be given explicitly:
\begin{equation}
A_0 = \left(
\begin{array}{cccc}
\msp 0   &\msp a_0 &\msp b_0 &\msp a_0 \\ 
    -a_0 &\msp 0   &\msp a_0 &\msp b_0 \\ 
    -b_0 &    -a_0 &\msp 0   &\msp a_0 \\ 
    -a_0 &    -b_0 &    -a_0 &\msp 0  
\end{array}
\right) \text{ ,} \quad
c_0=\frac{1}{\sqrt{1+4a_0^2-2b_0^2 + (2a_0^2 - b_0^2)^2}} \text{ ,}
\end{equation}
where the constants $a_0$ and $b_0$ are found by analytically minimizing \eqref{EQ_ISING_H}, yielding
\begin{equation}
a_0 = \sqrt{1 + 
\frac{1}{\sqrt{2}}} - 1 \approx 0.3066 \text{ ,} \quad
b_0 = \frac{68+8a_0+532a_0^2-616a_0^3-290a_0^4-58a_0^5}{43 + 16a_0 - 340 a_0^2 - 474 a_0^3 - 250 a_0^4 - 50 a_0^5} \approx 0.2346 \text{ .}
\end{equation}
All remaining tensors are numerically optimized within our three-parameter model. As shown in Fig.\ \ref{FIG_MERA_E_CONV}, the minimal energy density $\epsilon=\langle H \rangle / L$ converges quadratically with the number of boundary sites $L$. The optimal values for $a,b,c$ converge as well. At $L=1024$, those are given by $a = 0.6854$, $b = 0.5246$, and $c = 0.2172$, yielding a ground-state energy density $\epsilon_0 = -0.636533$ (decimals given up to convergent digits). The relative error with respect to the continuum solution $\epsilon_0 = 2/\pi$ is about $0.014 \%$. 
Note that this MERA model only has bond dimension $\chi = 2$, and that increasing $\chi$ would increase the size of the generating matrices and the number of free parameters, presumably allowing for even higher accuracy.

\begin{figure}[htb]
\centering
\includegraphics[width=0.32\textwidth]{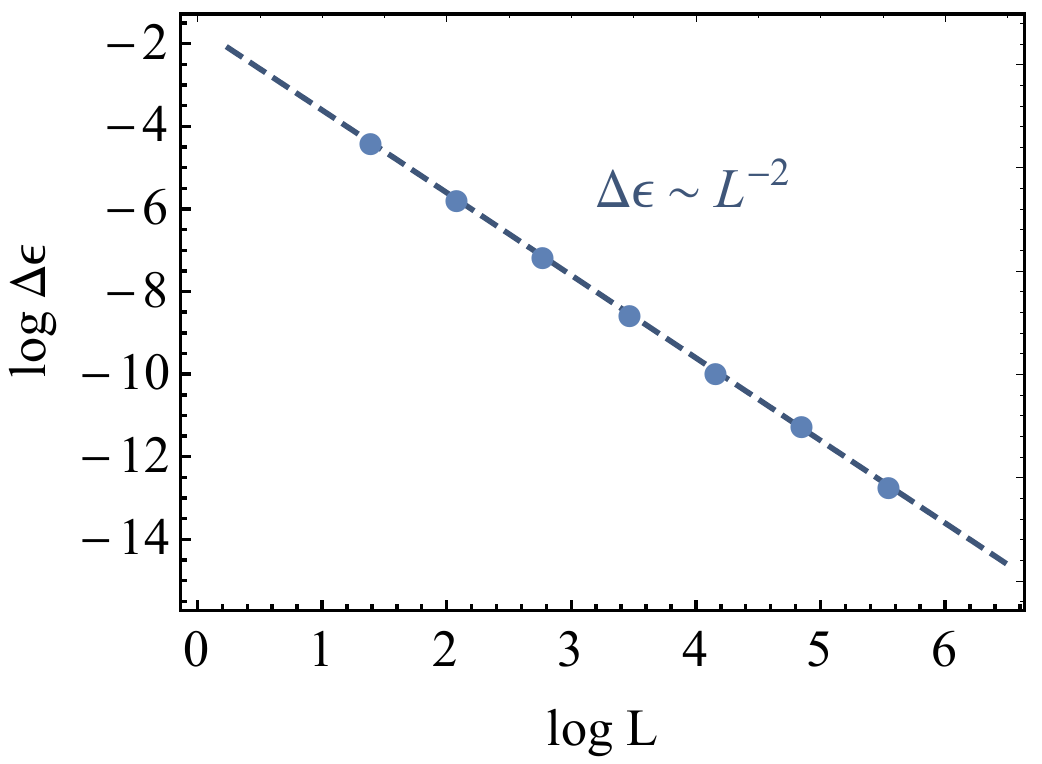}

\caption{Energy density $\epsilon$ of the MERA boundary state of $L$ sites, minimized with respect to the Ising Hamiltonian \eqref{EQ_ISING_H}. Plotted are the differences $\Delta\epsilon = \epsilon(2L)-\epsilon(L)$ between two MERA layers, with a quadratic falloff, fitted from the data, shown as a gray line.
}
\label{FIG_MERA_E_CONV}
\end{figure}

\subsection{Conformal data}
\label{ssec:app_conf_data}

In this subsection, we show how to obtain conformal data from the approach taken here.
The Ising theory at criticality can be described by a 1+1-dimensional conformal field theory (CFT) \cite{BigYellowBook1997}. The operator content of this theory is defined by its primary fields, whose scaling behavior is exactly known. This is because two-dimensional CFTs can be solved exactly, usually by mapping the space and time coordinates $(x,t)$ to a complex number $z=x+\i\,t$ and its complex conjugate $\bar{z}=x-\i\,t$. (Quasi-\/) primary fields $\phi(z)$ have associated \textsl{conformal weights} $h_\phi$ and $\bar{h}_\phi$, with correlations between different space-time points $z$ and $w$ being
given by
\begin{equation}
\langle \phi(z) \phi(w) \rangle = \frac{C_{\phi,\phi}}{(z-w)^{2 h_\phi} (\bar{z}-\bar{w})^{2\bar{h}_\phi}} \text{ .}
\end{equation}
The constant $C_{\phi,\phi}$ is not a fundamental CFT parameter, but determined by the normalization of $\phi$. As we are restricting ourselves to correlations on time-slices, we will find correlators of the form
\begin{equation}
\langle \phi(x) \phi(y) \rangle = \frac{C_{\phi,\phi}}{|x-y|^{2(h_\phi+\bar{h}_\phi)}} = \frac{C_{\phi,\phi}}{|x-y|^{2 \Delta_\phi}} \text{ ,}
\end{equation}
expressed in terms of the \textsl{scaling dimension} $\Delta_\phi = h_\phi + \bar{h}_\phi$. The three-point functions of primary fields $\phi, \chi$ and $\omega$ have the form
\begin{equation}
\label{EQ_CFT_THREE_POINT}
\langle \phi(x) \chi(y) \omega(z) \rangle 
= \frac{\sqrt{C_{\phi,\phi} C_{\chi,\chi} C_{\omega,\omega}}\; C_{\phi,\chi,\omega}}{|x-y|^{\Delta_\phi+\Delta_\chi-\Delta_\omega}\; |y-z|^{\Delta_\chi+\Delta_\omega-\Delta_\phi}\; |z-x|^{\Delta_\omega+\Delta_\phi-\Delta_\chi}} 
 \text{ ,}
\end{equation}
with the \textsl{structure constants} $C_{\phi,\chi,\omega}$ being fundamental CFT quantities. 

For the Ising CFT in two dimensions, there are three primary fields: The identity $\id$, the energy density $\epsilon$ and the spin (or ``order parameter'') $\sigma$. The Jordan-Wigner transformation gives us an alternative description in terms of the fermionic fields $\psi$ and $\bar{\psi}$. The corresponding scaling dimensions are the following.
\begin{center}
\begin{tabular}{@{\hspace{0.5cm}} c @{\hspace{0.5cm}} | @{\hspace{0.5cm}} c  @{\hspace{0.5cm}} @{\hspace{0.5cm}} c @{\hspace{1cm}} c @{\hspace{1cm}} c @{\hspace{1cm}} c @{\hspace{1cm}} c @{\hspace{0.5cm}}}
Field $\phi$ & $\id$ & $\epsilon$ & $\sigma$ & $\psi$ & $\bar{\psi}$ \\
 \hline
$h_\phi$       & 0 & 1/2 & 1/16 & 1/2 & 0 \\
$\bar{h}_\phi$ & 0 & 1/2 & 1/16 & 0   & 1/2 \\
$\Delta_\phi$  & 0 & 1   & 1/8  & 1/2 & 1/2
\end{tabular}
\end{center}
Furthermore, the structure constants in the spin sector are $C_{\sigma,\sigma,\id}=1$ and $C_{\sigma,\sigma,\epsilon}=\frac{1}{2}$.

For the Gaussian states produced by our matchgate tensor networks, all the information on correlators is stored in the Majorana covariance matrix with entries $\Gamma_{j,k}=\langle \tfrac{\i\,}{2} [\m_j,\m_k] \rangle$. Before calculating the scaling dimensions, let us first prove a useful identity regarding the covariance matrix of odd-pairing Hamiltonians of the form 
\begin{equation}
\label{EQ_HOP}
H_\text{OP}=\i\,\sum_{k,d} J_{k,d}\, \m_k \m_{k+2d-1} \text{ ,}
\end{equation}
with the couplings $J_{k,d} \in \mathbb{R}$ between Majorana sites at odd distance. In particular, this includes Hamiltonians with only nearest-neighbor Majorana coupling with $J_{k,d} = \delta_{d,1}\, J_k$, such as the Ising model considered above.
\begin{lemma}[Covariance matrices of odd-pairing Majorana Hamiltonians]\label{lm:COVMATRIX}
Eigenstates of Hamiltonians $H_\text{OP}$ of the form \eqref{EQ_HOP} are described by a covariance matrix $\Gamma$ whose entries $\Gamma_{j,k}$ 
vanish for even $j \spl k$.
\end{lemma}
\begin{proof}
Consider an eigenstate $\ket\psi =\sum_{x\in\{0,1\}^n}\mathcal T(x)\ket x$ of $H_\text{OP}$ with eigenenergy $E$.
We will first prove that $\mathcal T(x)\in \mathbb R$ for all $x\in\{0,1\}^n$.
We first note that we can write 
\begin{align}
H_\text{OP} &= \i\, \sum_{i,d} \left( J_{2i-1,d}\, \m_{2i-1}\m_{2i+2d-2} + J_{2i,d}\, \m_{2i}\m_{2i+2d-1} \right) \nonumber\\
&= \sum_{i,d} \left( J_{2i-1,d} (\fe_i + \fd_i)(\fe_{i+d-1} - \fd_{i+d-1}) + J_{2i,d} (\fe_i-\fd_i)(\fe_{i+d}+\fd_{i+d}) \right) \text{ ,}
\end{align} which implies that the Hamiltonian is not only Hermitian but also invariant under complex conjugation $H_\text{OP}^* = H_\text{OP}$.
We now decompose our eigenstate into its real and imaginary part denoted by $\ket\psi = \Re[\ket\psi]+\i\,\Im[\ket\psi]$.
The eigenequation reads
\begin{align}
  H_\text{OP}\ket \psi = E\ket\psi
\end{align}
and its complex conjugate can be expressed as
\begin{align}
 (H_\text{OP}\ket \psi)^* =  H_\text{OP}\ket \psi^* = H_\text{OP}\, \Re[\ket\psi]-\i\,H_\text{OP}\,\Im[\ket\psi] \overset{!}{=} E\, \Re[\ket\psi]-\i\, E\,\Im[\ket\psi]
\end{align}
where the only difference to the original equation is the minus sign.
Adding and subtracting these two equations yields
\begin{align}
    H_\text{OP}\, \Re[\ket\psi] = E\,\Re[\ket\psi]
\end{align}
and
\begin{align}
    H_\text{OP}\, \Im[\ket\psi] = E\,\Im[\ket\psi]
\end{align}
which means that $\Re[\ket\psi]$ and $\Im[\ket\psi]$ are both eigenvectors.
If the spectrum is non-degenerate then they are collinear and hence we can assume that $\ket\psi ^*=\ket\psi$ up to a phase.
If the spectrum is degenerate then we can also choose real eigenstates because $\text{span}(\ket\psi,\ket\psi^*)=\text{span}(\Re[\ket\psi],\Im[\ket\psi])$.
In the context of matchgates, this means that all eigenstates are expressed by real generating matrices, which are therefore a suitable ansatz for ground states of such Hamiltonians.

Next, we show that the matrix elements of the covariance matrix vanish for even $j \spl k$.
For $j \seq k$ this is true by definition, so we assume $j \sneq k$.
For even $j$ and $k$ we find $\frac{\i\,}{2} [\m_{j},\m_{k}]=\i\,\m_{j}\m_{k}=-\i\,(\fe_{j/2}-\fd_{j/2})(\fe_{k/2}-\fd_{k/2})$
while for odd $j$ and $k$ we have $\frac{\i\,}{2} [\m_{j},\m_{k}]=\i\, \m_{j} \m_{k}=\i\, (\fe_{(j+1)/2}+\fd_{(j+1)/2})(\fe_{(k+1)/2}+\fd_{(k+1)/2})$.
Either way, these Hermitian operators have purely imaginary coefficients in terms of creation and annihilation operators.
Evaluated in a state with real amplitudes, as shown above, the expectation value can only be imaginary. As all observables are real, it must therefore vanish altogether.
\end{proof}

Now let us relate the covariance matrix entries $\Gamma_{j,k}$ to the primary fields. By construction of our covariance matrix in Section \ref{sec:app_conversion}, $\braket{\psi}{\psi}=1$ and thus the identity $\id$ does not scale. However, the normalization factor $Z$ in \eqref{EQ_PSI_NORM} can in principle scale with the size of the contracted network. To ensure normalization, we have to act on each of the $N_T$ contracted tensors with a scaling factor $f=Z^{-1/(2 N_T)}$. We find that this factor $f$ converges for large systems, ensuring $\Delta_\id = 0$. Explicitly, $f_{\lbrace 3,6 \rbrace} \approx 0.972$ and $f_{\lbrace 3,7 \rbrace} \approx 0.941$ for the regular tilings and $f_\text{mMERA} \approx 0.959$ for the mMERA. 

We identify the fermionic fields $\psi$ and $\bar{\psi}$ with physical operators $\psi_k:=  \fe_k=\frac{1}{2} (\m_{2k-1} + \i\, \m_{2k})$ and $\bar{\psi}_k := \i\, \fd_k=\frac{1}{2} (\i\, \m_{2k-1} + \m_{2k})$. We then find that
\begin{equation}
\langle \psi_j \psi_k \rangle =  \langle \bar{\psi}_j \bar{\psi}_k \rangle = \frac{1}{4} \left( \Gamma_{2j,2k-1} + \Gamma_{2j-1,2k} \right) \text{ .}
\end{equation}
Note that we have used Lemma \ref{lm:COVMATRIX} to simplify the result. As we are considering Gaussian states, $\langle \psi_k \rangle=\langle \bar{\psi}_k \rangle=0$ for any $k$, so we do not have to consider expectation values of the individual fields.
Next, we compute the energy density $\epsilon$. On site $k$, we simply consider the local operator
 $\epsilon_k := \i\, \psi_k \bar{\psi}_k = \frac{\i\,}{2} \m_{2k-1}\m_{2k}$. Using Wick's theorem, the two-point functions follow as
\footnote{With our definition, $\langle \epsilon \rangle \neq 0$, so we need to subtract the field's expectation value, equivalent to using a field $\epsilon^\prime = \epsilon - \langle\epsilon\rangle$.}
\begin{equation}
\langle \epsilon_j \epsilon_k \rangle - \langle \epsilon_j \rangle \langle \epsilon_k \rangle 
= \frac{1}{4} \left( -\langle \m_{2j-1}\m_{2j} \m_{2k-1}\m_{2k} \rangle + \langle \m_{2j-1}\m_{2j} \rangle \langle \m_{2k-1}\m_{2k} \rangle \right) = \frac{1}{4} \Gamma_{2j-1,2k} \Gamma_{2j,2k-1} \text{ .}
\end{equation}
The order $\sigma$ is a nonlocal operator in the Majorana picture but corresponds to a $\sigma^x$ operator\footnote{In fact, $\epsilon$ can be related to the $\sigma^z$ operator, which conveniently acts locally in terms of Majorana operators, as well.}
 in the spin picture, obtained through a Jordan-Wigner transformation. A two-point correlator of $\sigma^x_k$ at different sites $k$ corresponds to a chain of Majorana operators,
\begin{align}
\langle \sigma^x_k \sigma^x_{k+1} \rangle &= -\i\,\langle \m_{2k} \m_{2k+1} \rangle 
= - \Pf{
\begin{array}{ll}
\Gamma_{2k,2k}   & \Gamma_{2k,2k+1} \\
\Gamma_{2k+1,2k} & \Gamma_{2k+1,2k+1}
\end{array} 
} \text{ ,} \\
\langle \sigma^x_k \sigma^x_{k+2} \rangle &= - \langle \,\m_{2k} \m_{2k+1} \m_{2k+2} \m_{2k+3} \rangle 
= \Pf{
\begin{array}{llll}
\Gamma_{2k,2k}   & \Gamma_{2k,2k+1}   & \Gamma_{2k,2k+2}  & \Gamma_{2k,2k+3} \\
\Gamma_{2k+1,2k} & \Gamma_{2k+1,2k+1} & \Gamma_{2k+1,2k+2}& \Gamma_{2k+1,2k+3} \\
\Gamma_{2k+2,2k} & \Gamma_{2k+2,2k+1} & \Gamma_{2k+2,2k+2}& \Gamma_{2k+2,2k+3} \\
\Gamma_{2k+3,2k} & \Gamma_{2k+3,2k+1} & \Gamma_{2k+3,2k+2}& \Gamma_{2k+3,2k+3} \\
\end{array} 
}\text{ ,} \\
\langle \sigma^x_j \sigma^x_{k} \rangle &= (-\i)^{k-j}\, \langle \m_{2j} \m_{2j+1} \dots \m_{2k-2} \m_{2k-1} \rangle = (-1)^{k-j}\, \Pf{\Gamma_{|[2j,2k-1]}} \text{ .}
\end{align}
The absolute value of the Pfaffian is given by $|\Pf{M}| = \sqrt{\det{M}}$. Note that because $\sigma^x_k$ is an odd product of Majorana operators, $\langle \sigma^x_k \rangle=0$.

Additionally, we compute the structure constant $C_{\sigma,\sigma,\epsilon}$ from the corresponding three-point correlator:
\begin{align}
\label{EQ_CORR_SSE}
\langle \sigma^x_j \sigma^x_k \epsilon_l \rangle - \langle \sigma^x_j \sigma^x_k \rangle \langle \epsilon_l \rangle &= \frac{-1}{2} \left( (-\i)^{k-j+1} \langle \m_{2j} \m_{2j+1} \dots \m_{2k-2} \m_{2k-1} \m_{2l-1} \m_{2l} \rangle + (-1)^{k-j}\, \Pf{\Gamma_{|[2j,2k-1]}} \Gamma_{2l-1,2l} \right) \\
&= \frac{(-1)^{k-j}}{2} \left( \Pf{\Gamma_{|[2j,2k-1] \cap \{ 2l-1,2l \}}} - \Pf{\Gamma_{|[2j,2k-1]}} \Gamma_{2l-1,2l} \right) \text{ .}
\end{align}
In order to use this result to compute the value of $C_{\sigma,\sigma,\epsilon}$ as in \eqref{EQ_CFT_THREE_POINT}, we consider the special case $k-j=l-k=d$, for some integer distance $d$. We then expect a scaling
\begin{equation}
\label{EQ_CORR_SSE2}
\langle \sigma^x_j \sigma^x_{j+d} \epsilon_{j+2d} \rangle - \langle \sigma^x_j \sigma^x_{j+d} \rangle \langle \epsilon_{j+2d} \rangle = \frac{C_{\sigma,\sigma}\sqrt{C_{\epsilon,\epsilon}}\; C_{\sigma,\sigma,\epsilon}}{2^{\Delta_\epsilon} d^{\,2\Delta_\sigma + \Delta_\epsilon}}
 \text{ .}
\end{equation}

Using these tools for extracting two- and three-point correlators, we compute the scaling powers $p_\phi$ for the various fields $\phi$ by fitting the dependence of $\langle \phi_i \phi_{i+d} \rangle$ on distance $d$. The resulting graphs for $\phi \in \lbrace \psi, \epsilon, \sigma \rbrace$ are presented for the regular $\lbrace 3,6 \rbrace$ and $\lbrace 3,7 \rbrace$ tilings as well as for the mMERA tiling in figures \ref{FIG_CRIT_SCALING_REG_FLAT}, \ref{FIG_CRIT_SCALING_REG} and \ref{FIG_CRIT_SCALING_MERA}, respectively. We also compute $C_{\sigma,\sigma,\epsilon}$ with \eqref{EQ_CORR_SSE} and \eqref{EQ_CORR_SSE2}, using the scaling dimensions $\Delta_\sigma, \Delta_\epsilon$ and normalizations $C_{\sigma,\sigma},C_{\epsilon,\epsilon}$ from the previous fits as inputs. Furthermore, we compute the energy density $\epsilon_0=\langle H \rangle / L$ with respect to the Ising Hamiltonian \eqref{EQ_ISING_H}.
Note that the regular $\lbrace 3,6 \rbrace$ and $\lbrace 3,7 \rbrace$ tilings are not translation invariant, leading to irregularities on small scales and amplified finite-size effects. This also leads to larger deviations from the exact ground state energy density $\epsilon_0=-2/\pi$.

Finally, we can compute the \textsl{central charge} $c$ characterizing the CFT. This is achieved by considering the scaling of the entanglement entropy $S_A$ with the subsystem size $\ell=|A|$. We expect the exact result \eqref{EQ_CALABRESE_CARDY} for a critical theory, with $c=1/2$ for an Ising CFT. For a subsystem $A=[ k, k+\ell ]$, $S_A$ can be computed from the symplectic eigenvalue spectrum of the partial covariance matrix $\Gamma_{|A}$ \cite{Serafini:2003ke}. In detail, one performs an orthogonal transformation $\Gamma_{|A} = Q\, \tilde{\Gamma}_{|A}\, Q^T$ into the form
\begin{equation}
\tilde{\Gamma}_{|A} = \bigoplus_{i=1}^L \left(\begin{matrix}
  0 & \lambda_i \\
  -\lambda_i & 0
\end{matrix}\right) \text{ ,}
\end{equation}
which is most conveniently achieved using numerical Schur decomposition, and then reading off the entanglement entropy as
\begin{equation}
S_A = \sum_{i=1}^L \left( - \frac{1+\lambda_i}{2} \log \frac{1+\lambda_i}{2} - \frac{1-\lambda_i}{2} \log \frac{1-\lambda_i}{2} \right) \text{ .}
\end{equation}
Our combined results for scaling dimensions, the structure constant $C_{\sigma,\sigma,\epsilon}$, the ground state energy $\epsilon_0$ and the central charge $c$ are summarized in Table \ref{TBL_CFT_DATA} of the main text.


\begin{figure}[p]
\centering
\includegraphics[width=0.25\textwidth]{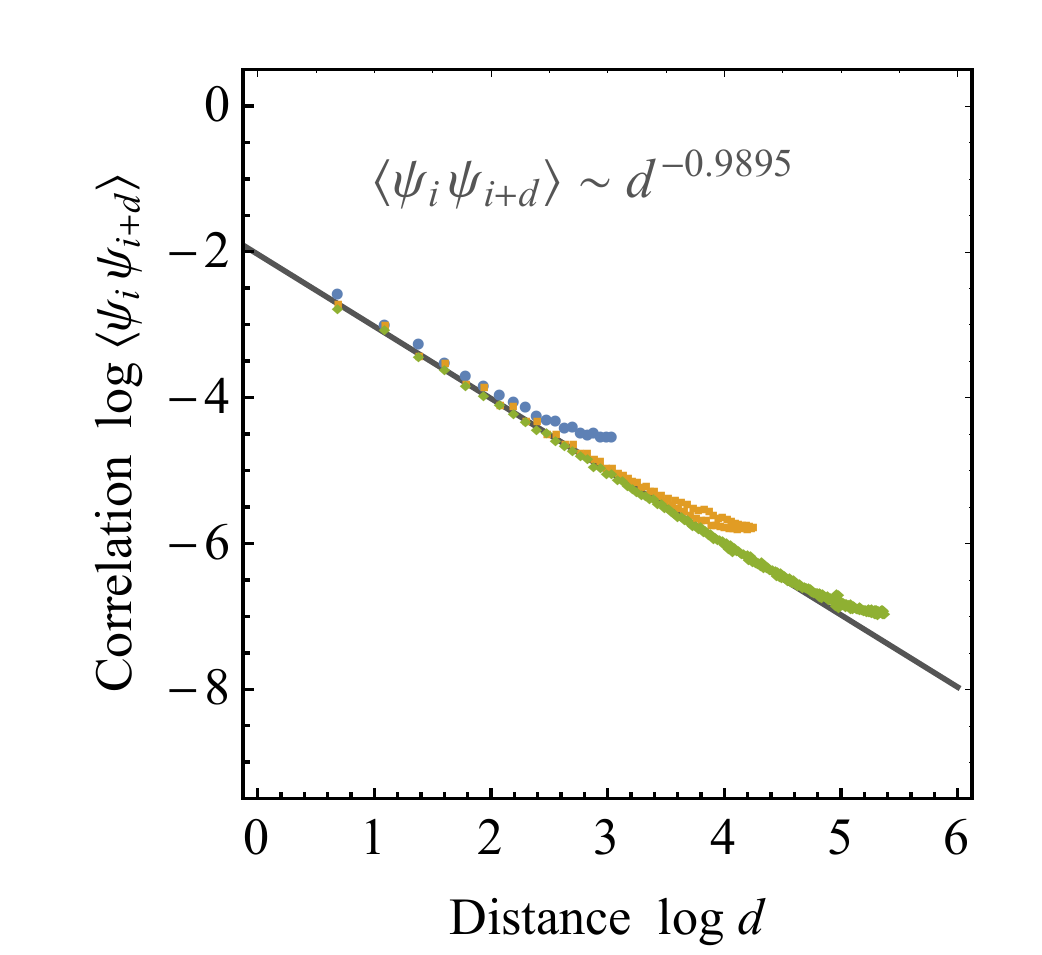}
\includegraphics[width=0.25\textwidth]{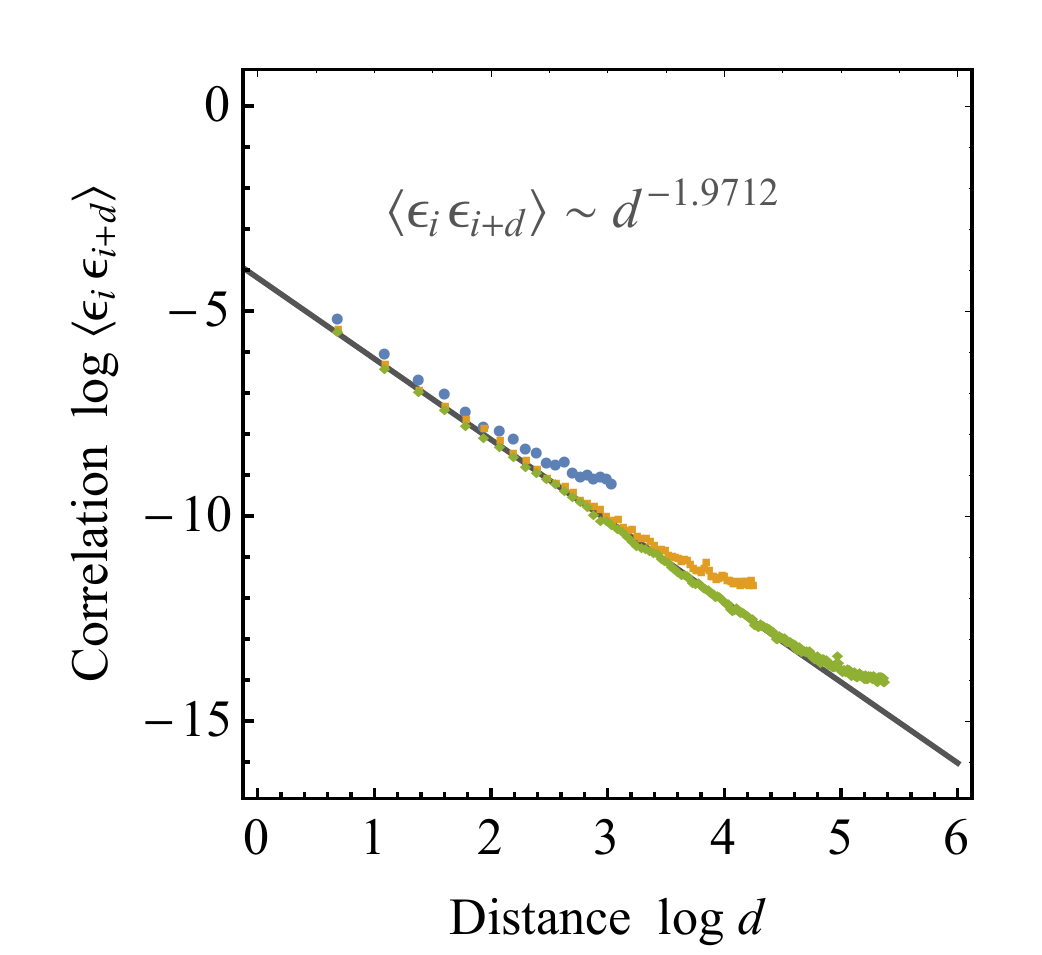}
\includegraphics[width=0.25\textwidth]{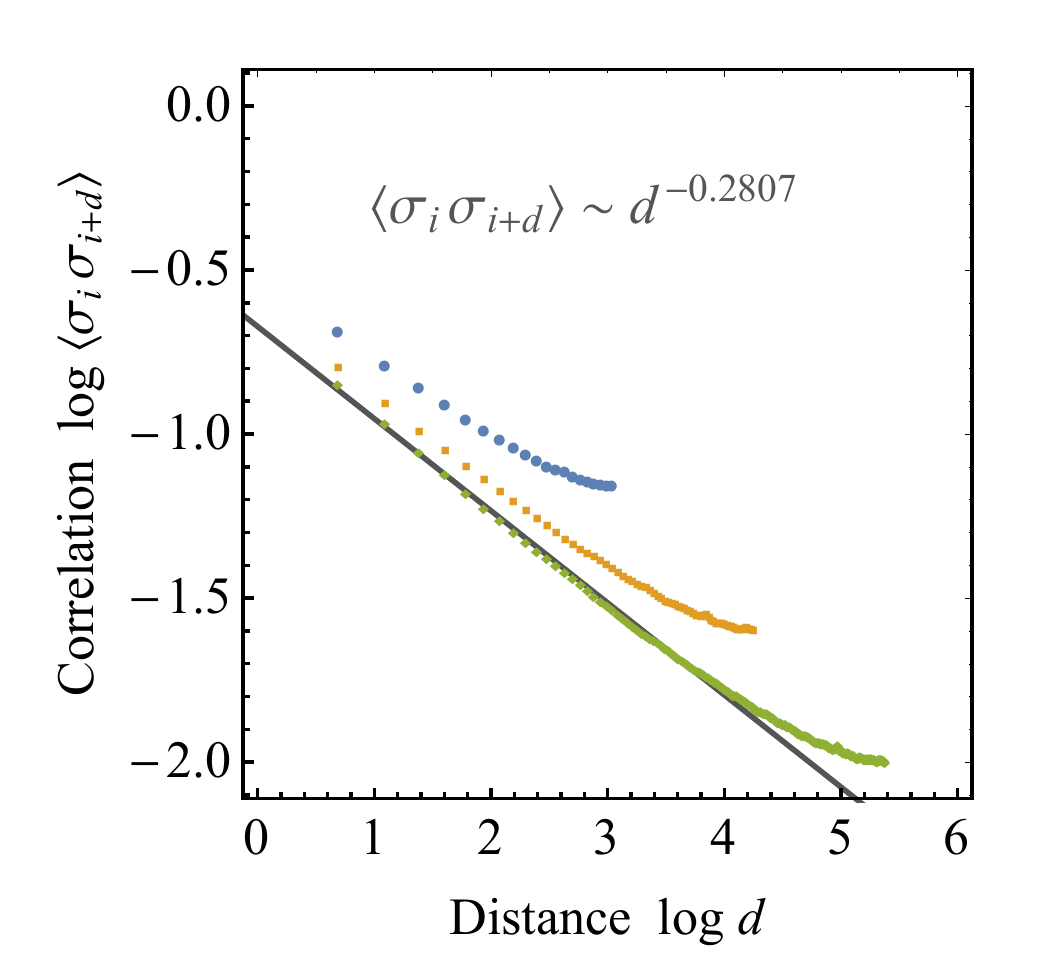}

\caption{Scaling of primary operators $\psi$, $\epsilon$ and $\sigma$ in the regular $\lbrace 3,6 \rbrace$ tiling for boundary states of 84, 282, and 870 Majorana sites (blue, yellow and green points, respectively). Numerical fit of scaling power law shown as grey line. Correlators $\langle \phi_i \phi_{i+d} \rangle$ of fields $\phi$ at distance $d$ are averaged over all sites $i$.}
\label{FIG_CRIT_SCALING_REG_FLAT}
\end{figure}

\begin{figure}[p]
\centering
\includegraphics[width=0.25\textwidth]{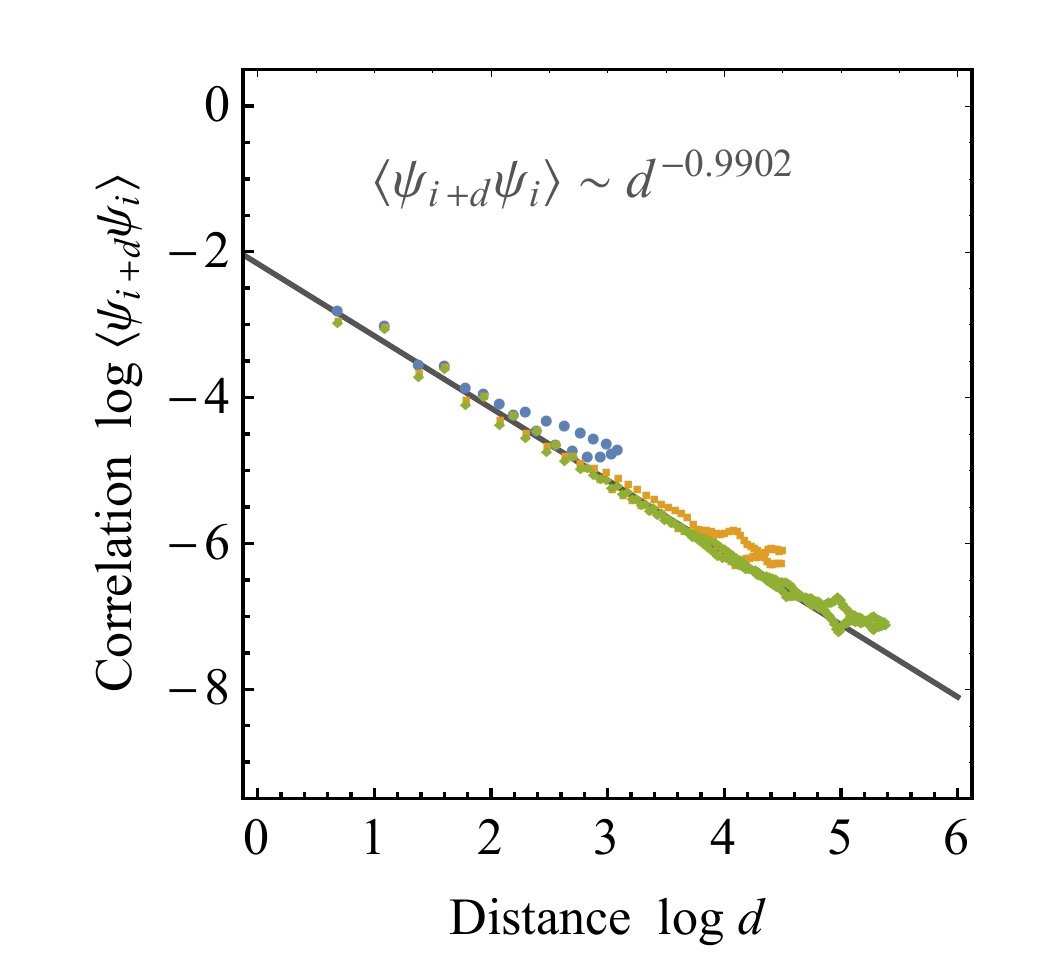}
\includegraphics[width=0.25\textwidth]{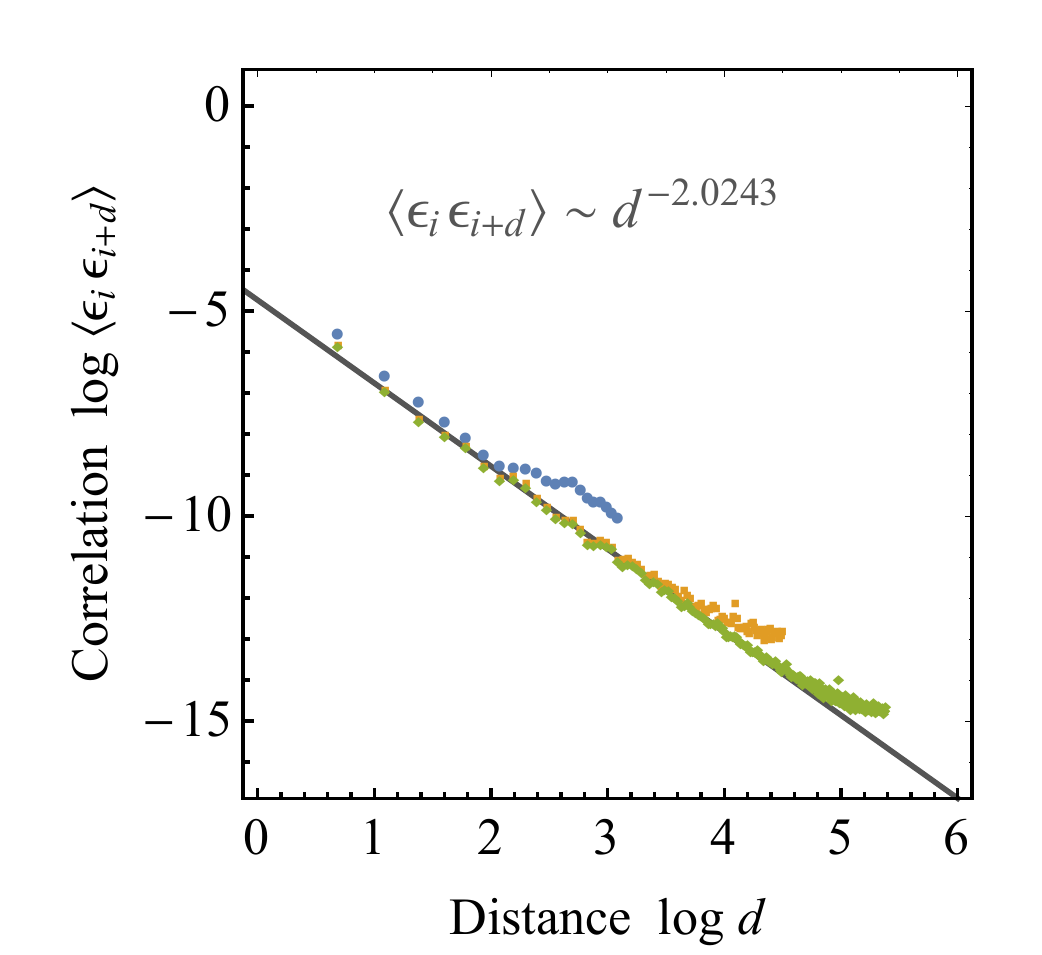}
\includegraphics[width=0.25\textwidth]{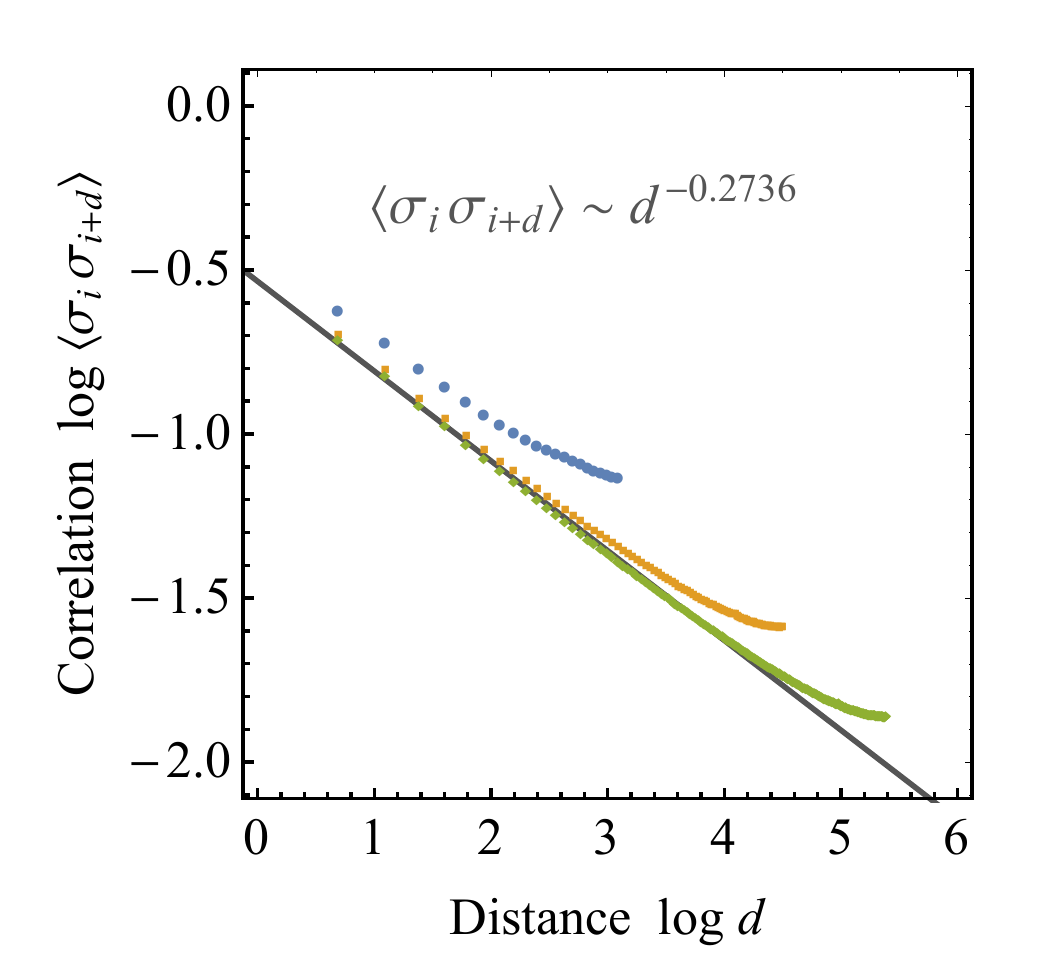}

\caption{Scaling of primary operators $\psi$, $\epsilon$ and $\sigma$ in the regular $\lbrace 3,7 \rbrace$ tiling for boundary states of 90, 360, and 876 Majorana sites (blue, yellow and green points, respectively). Numerical fit of scaling power law shown as grey line. Correlators $\langle \phi_i \phi_{i+d} \rangle$ of fields $\phi$ at distance $d$ are averaged over all sites $i$.}
\label{FIG_CRIT_SCALING_REG}
\end{figure}

\begin{figure}[p]
\centering
\includegraphics[width=0.25\textwidth]{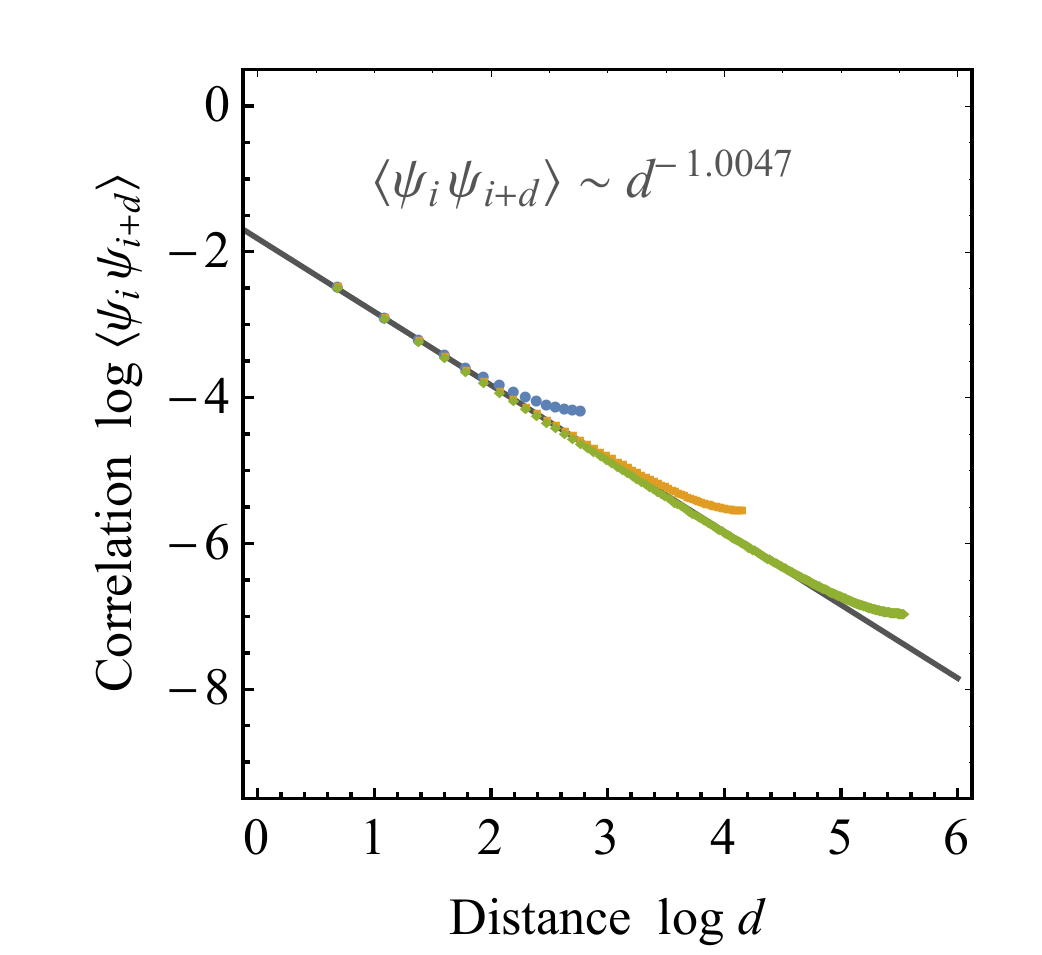}
\includegraphics[width=0.25\textwidth]{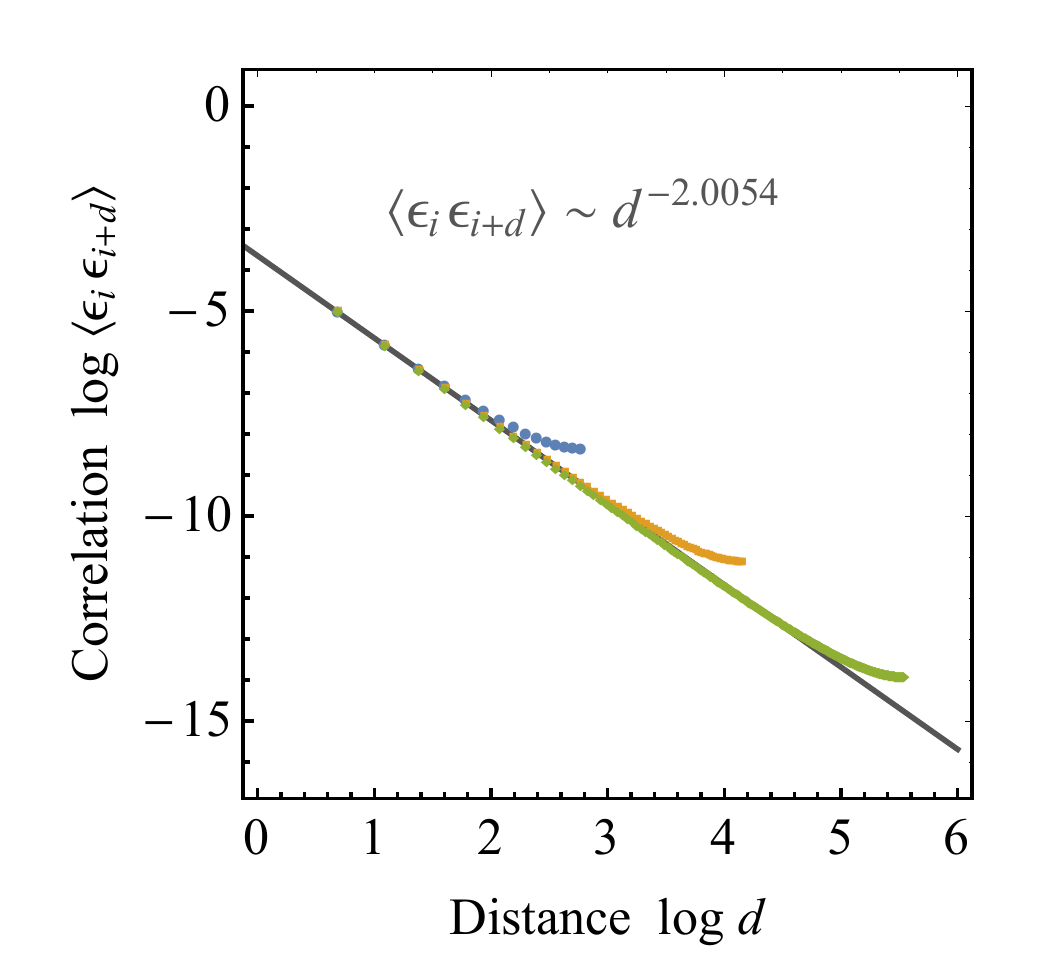}
\includegraphics[width=0.25\textwidth]{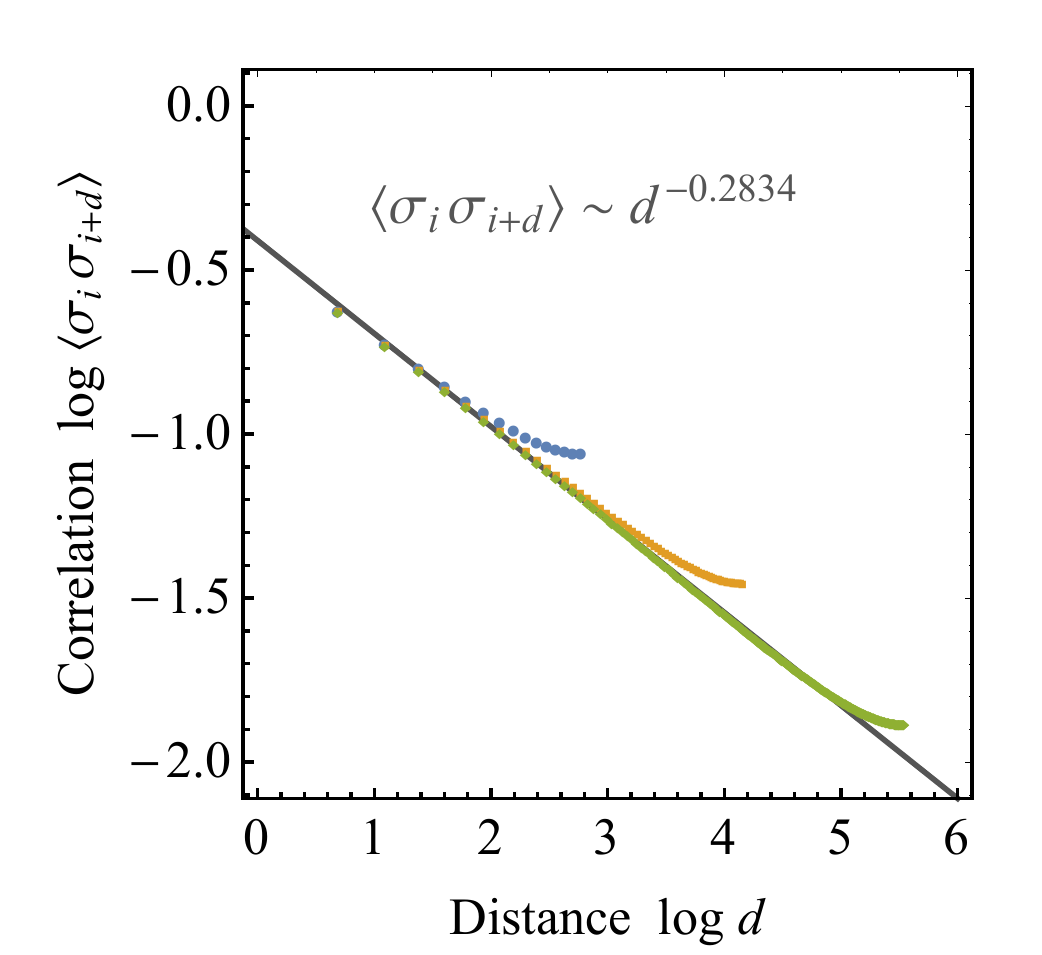}

\caption{Scaling of primary operators $\psi$, $\epsilon$ and $\sigma$ in the mMERA tiling for boundary states of 64, 256, and 1024 Majorana sites (blue, yellow and green points, respectively). Numerical fit of scaling power law shown as grey line. Correlators $\langle \phi_i \phi_{i+d} \rangle$ of fields $\phi$ at distance $d$ are averaged over all sites $i$.}
\label{FIG_CRIT_SCALING_MERA}
\end{figure}

\begin{figure}[p]
\centering
\includegraphics[width=0.25\textwidth]{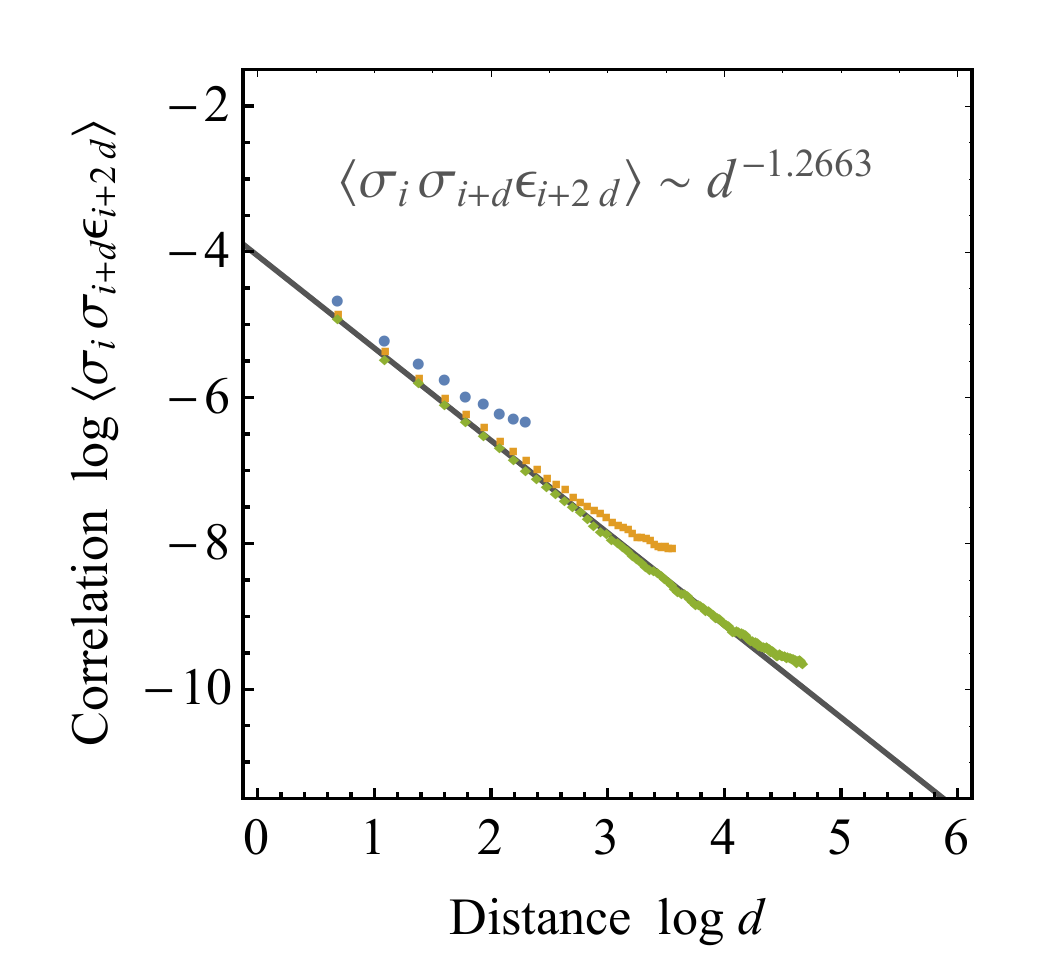}
\includegraphics[width=0.25\textwidth]{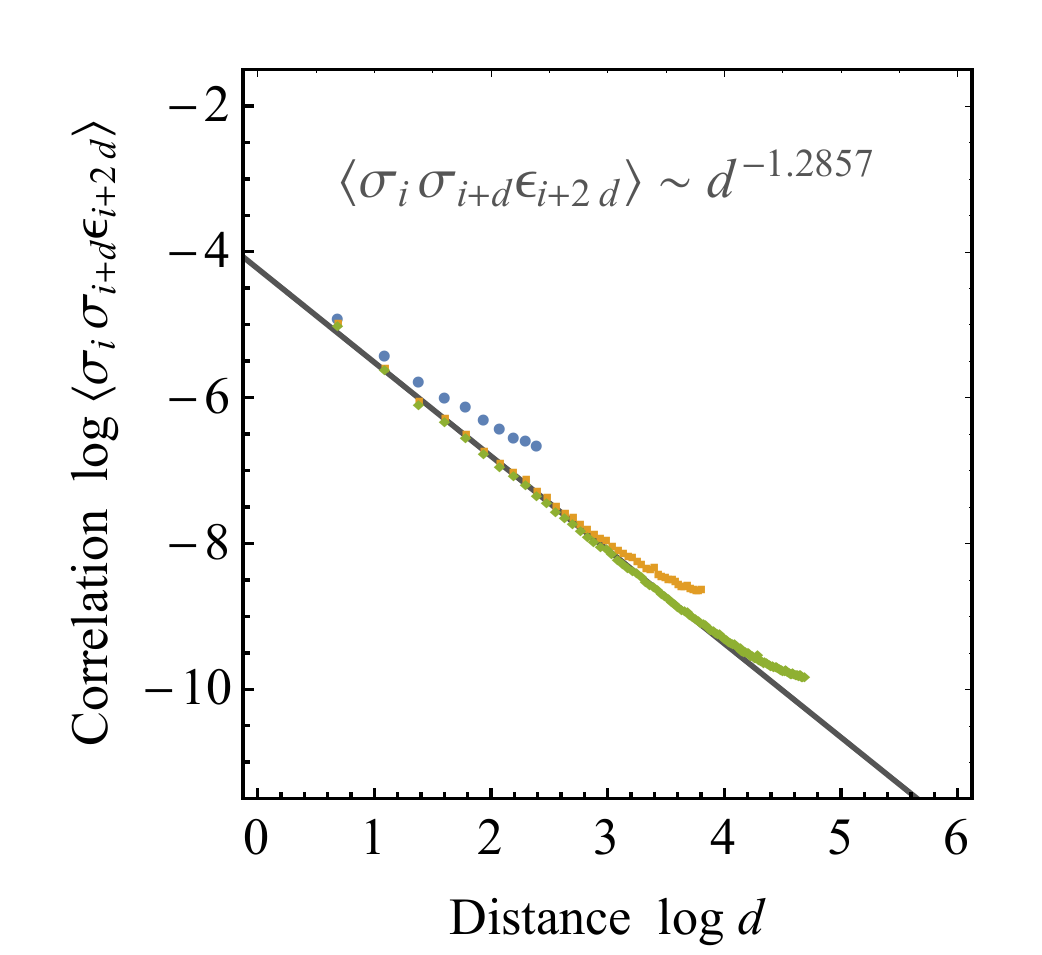}
\includegraphics[width=0.25\textwidth]{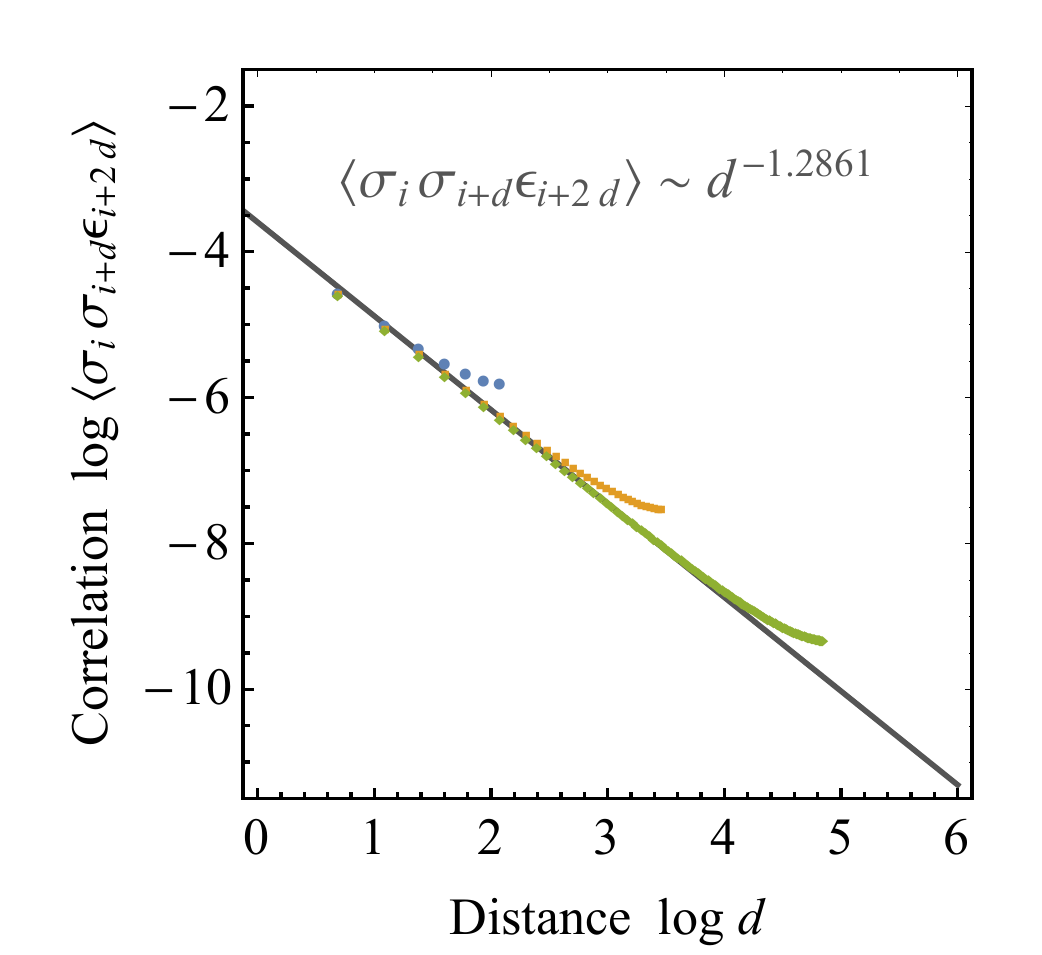}

\caption{Scaling of the three-point function $\langle \sigma_i\sigma_{i+d}\epsilon_{i+2d} \rangle$ at distance $d$, averaged over all sites $i$, for the regular $\lbrace 3,6 \rbrace$ and $\lbrace 3,7 \rbrace$, as well as the mMERA tilings (from left to right). Numerical fit of scaling power law, based on the data from Fig.\ \ref{FIG_CRIT_SCALING_REG_FLAT}-\ref{FIG_CRIT_SCALING_MERA}, shown as grey line.}
\label{FIG_CRIT_SCALING_COEFF}
\end{figure}

\subsection{IR cutoff}
\label{ssec:app_ir_cut}

While the matchgate model is restricted to planar graphs, it is possible to construct an effective IR cutoff, i.e.\ a ``black hole horizon'', by changing the tensor content of tensors in the center of the network. For a regular $\lbrace 3, k \rbrace$ tiling with $k \ge 7$ this cutoff is simply a cutoff radius $r_\text{cut}$ in the Poincar\'e disk  with $0 \leq r_\text{cut} < 1$. For a  flat $\lbrace 3, 6 \rbrace$ tiling $r_\text{cut}$ becomes a radius in the flat Euclidean plane with $0 \leq r_\text{cut} < \infty$. While the MERA can also be embedded in the Poincar\'e disk, it is more convenient to define a cutoff layer $n_\text{cut}$, with the first $n_\text{cut}$ MERA layers (isometries and disentanglers) and the central tensor being affected.

There are two natural choices for the tensor's generating matrices $A$ in the cutoff region: Either setting all components $A_{i,j}$ with $i<j$ to zero or to one, corresponding to a local vacuum or a fully occupied state, respectively. We find that both produce gapped states on the boundary, but that the former choice leads to periodic boundary conditions, while the latter produces anti-periodic ones. As we have been considering the anti-periodic case in the previous examples, we also choose this case here.

The results are shown in Fig.\ \ref{FIG_CRIT_SCALING_CUTOFF} with regard to the scaling of the fermionic field $\psi$ and the dependence of the entanglement entropy $S_A$ on the length $l$ of the subsystem $A$. Outside of the cutoff region, the tensor content is identical to the one used to produce a boundary Ising CFT in the previous section.

After a characteristic length scale $\xi$ depending on the cutoff, we see that the $\psi$ field's power law scaling transitions to an exponential falloff, as would be expected in a gapped (massive) theory. Furthermore, $S_A$ saturates for $l > \xi$, which allows us to directly extract $\xi$ from the entanglement entropy formula for a massive QFT \cite{CalabreseReview},
\begin{equation}
S_A = \frac{c}{3} \log\frac{\xi}{a} \text{ ,}
\end{equation}
which holds in the limit where $\xi$ is much larger than the lattice spacing $a$. The values for $c$ and $a$ are given by the full entanglement entropy scaling \eqref{EQ_CALABRESE_CARDY} at zero cutoff (note that $a$ depends on the tiling). Without a cutoff, $\xi$ can be identified with the length of the system, which is infinite in the CFT limit.

\begin{figure}[htb]
\centering
\includegraphics[width=0.25\textwidth]{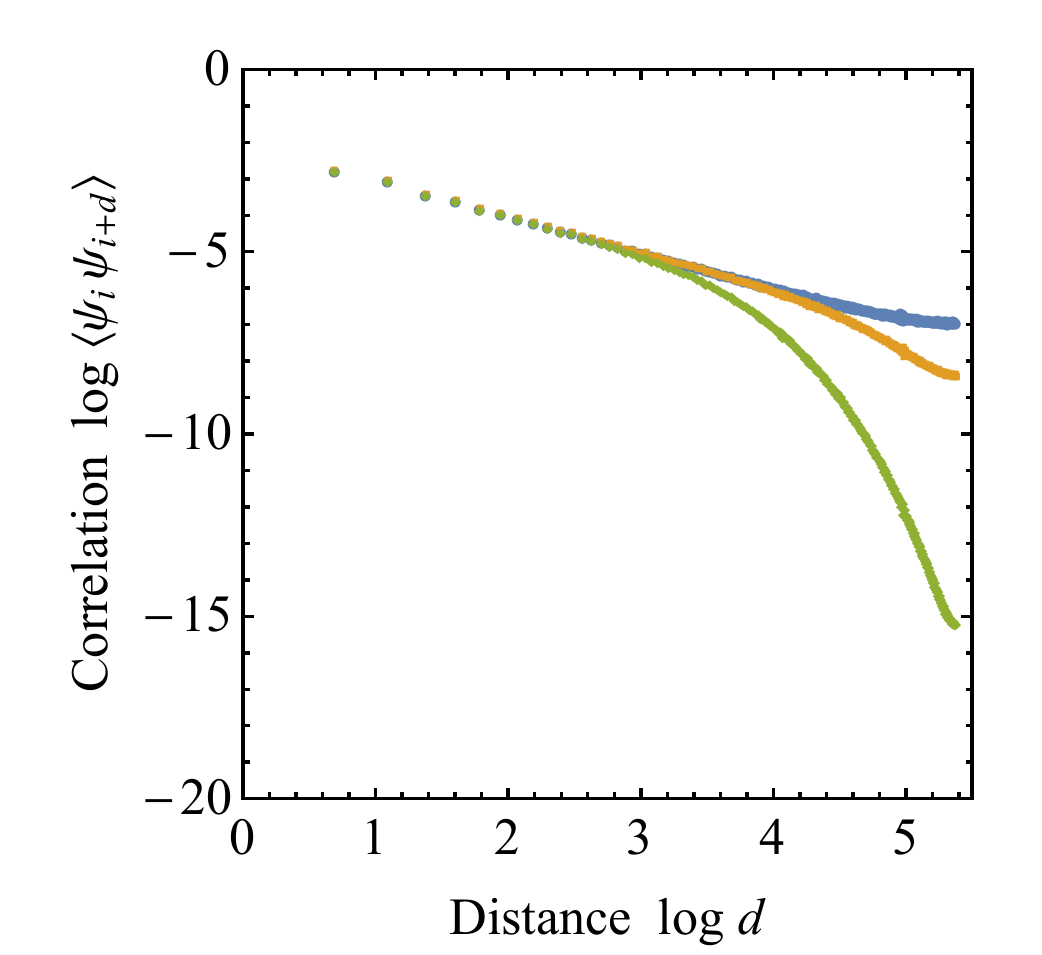}
\includegraphics[width=0.25\textwidth]{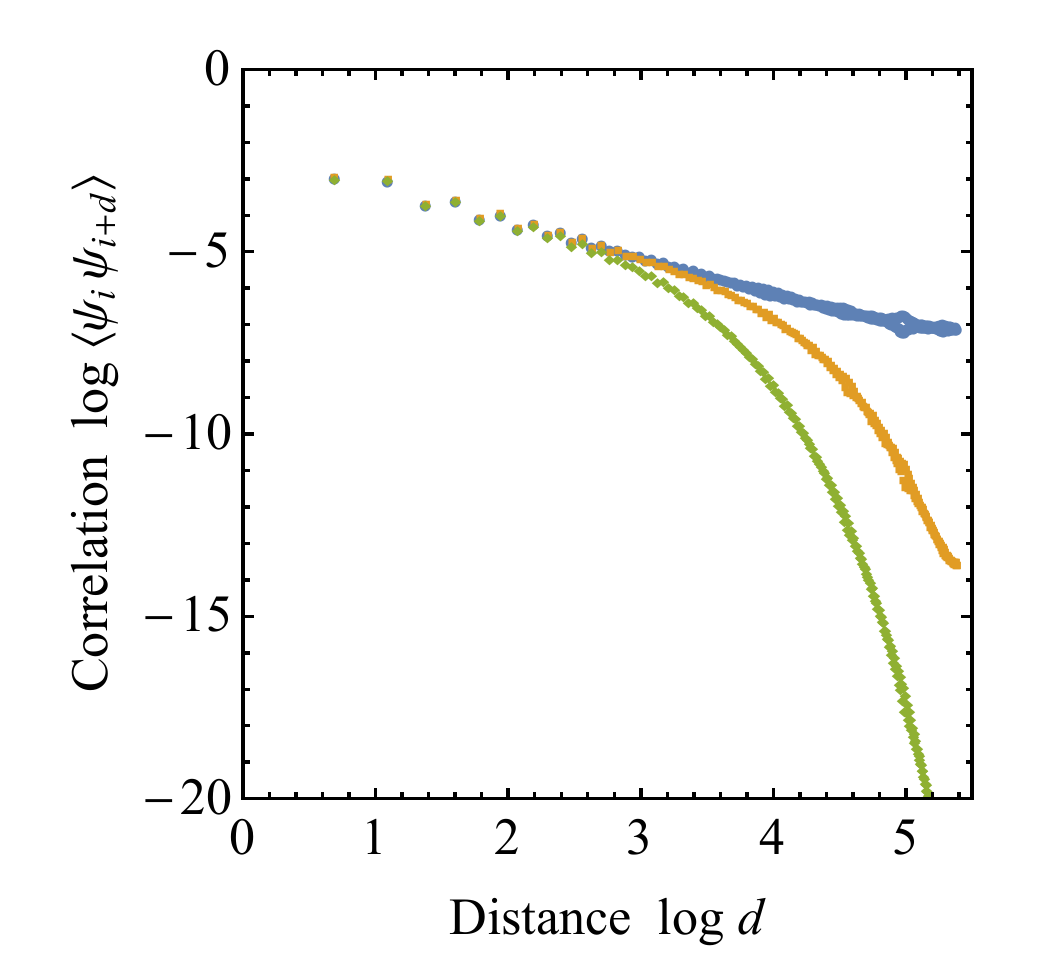} 
\includegraphics[width=0.25\textwidth]{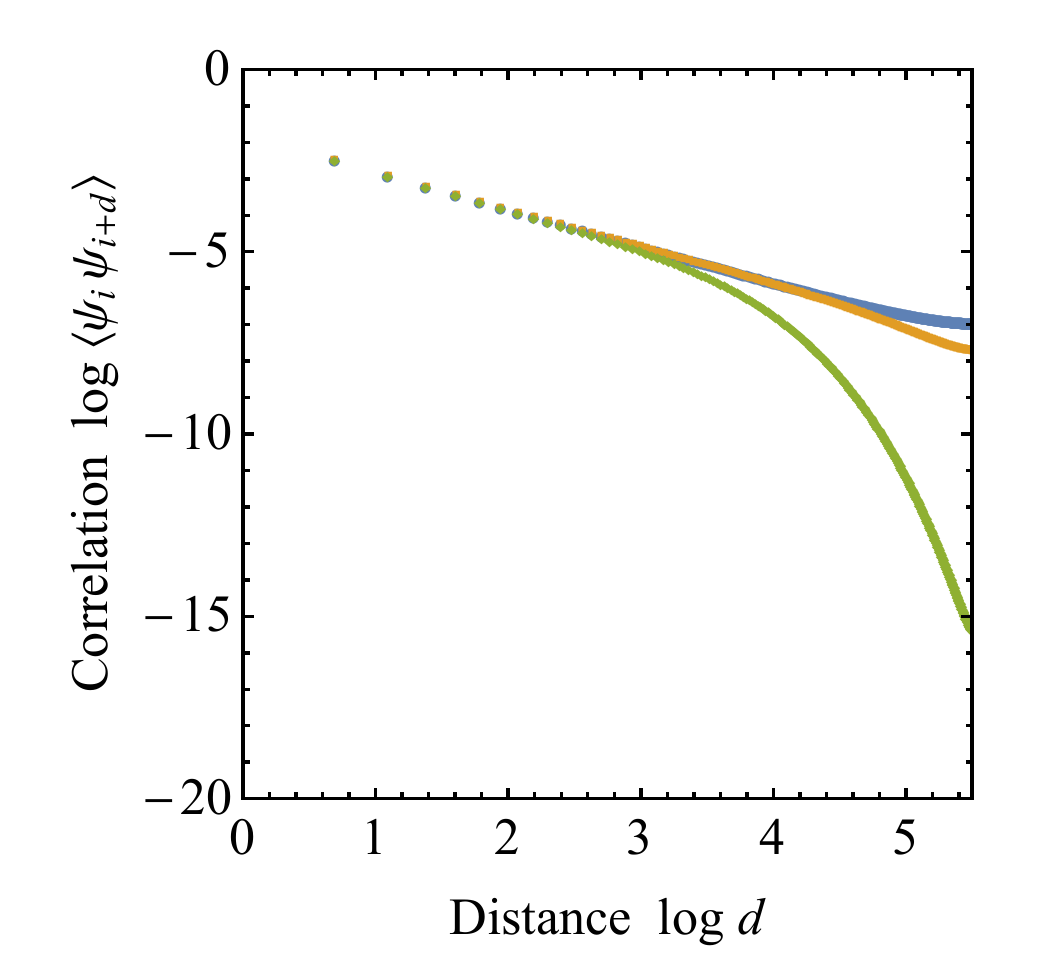}\\
\includegraphics[width=0.25\textwidth]{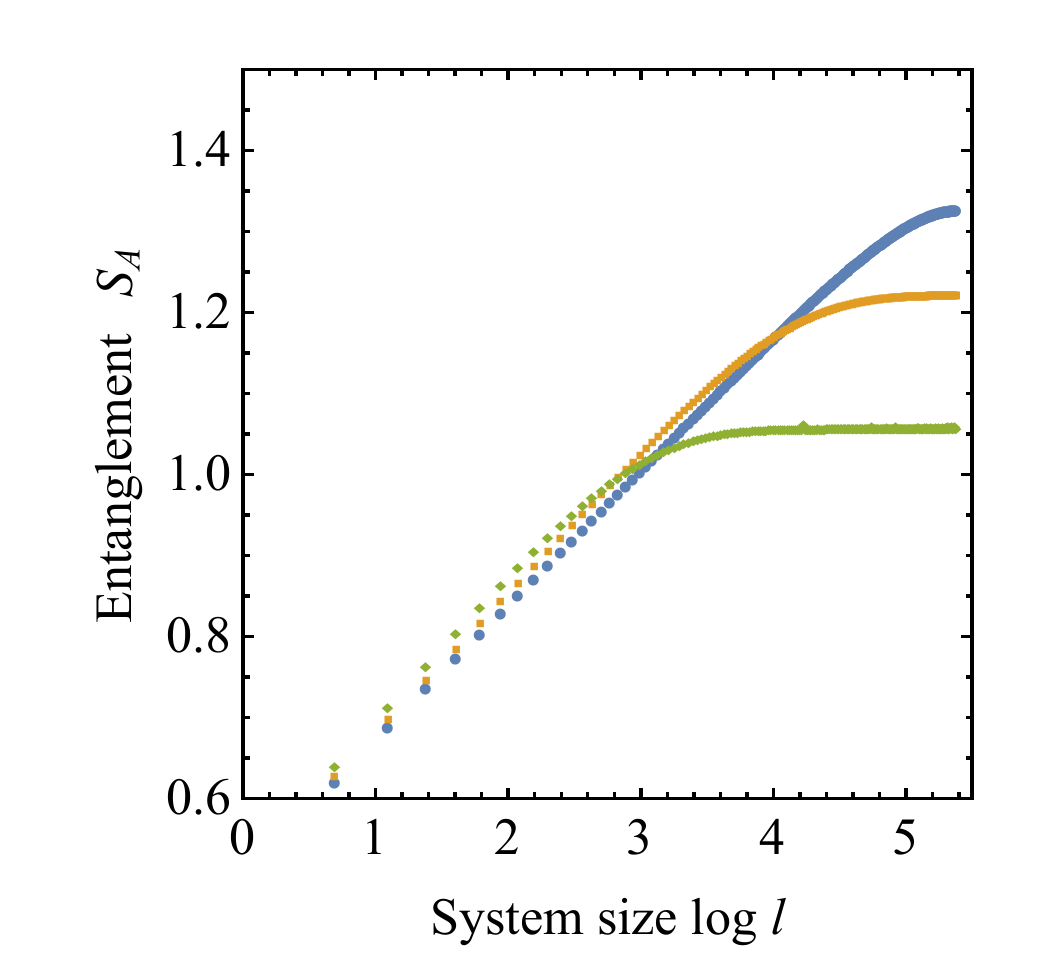}
\includegraphics[width=0.25\textwidth]{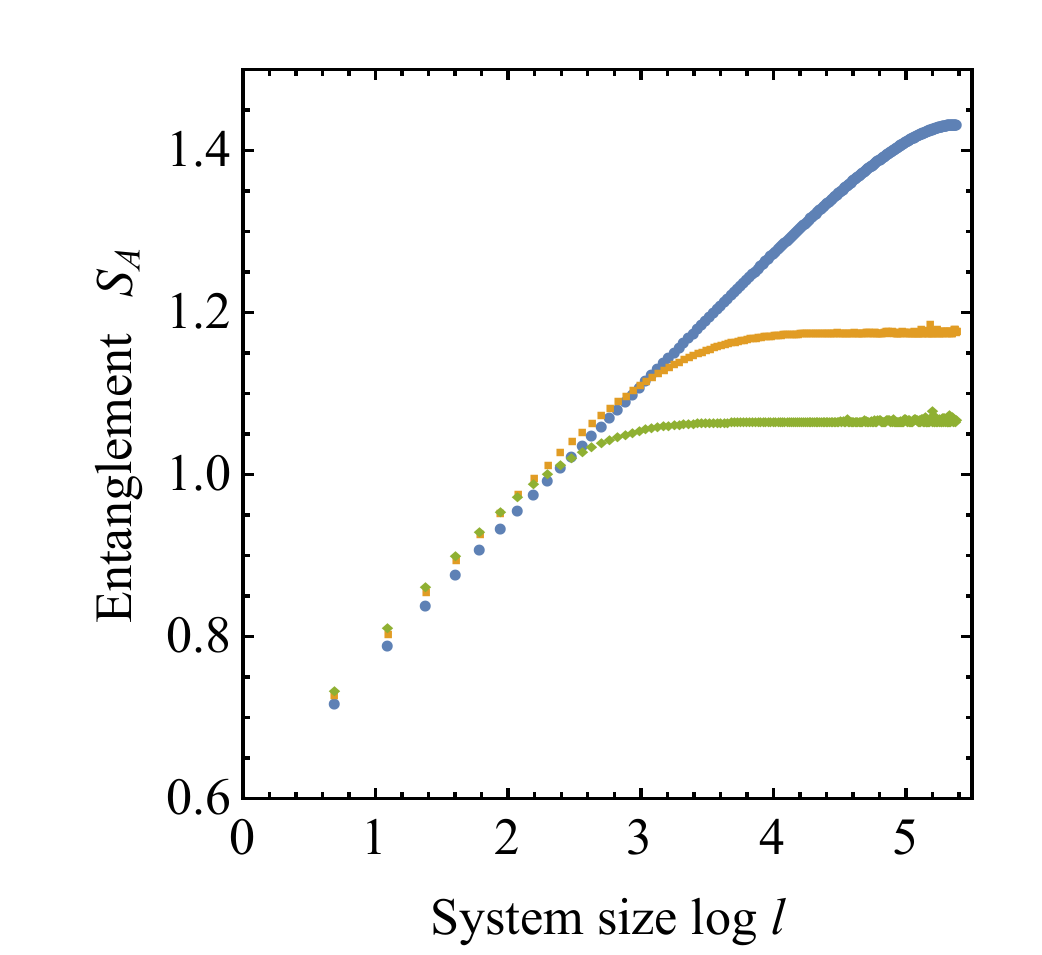}
\includegraphics[width=0.25\textwidth]{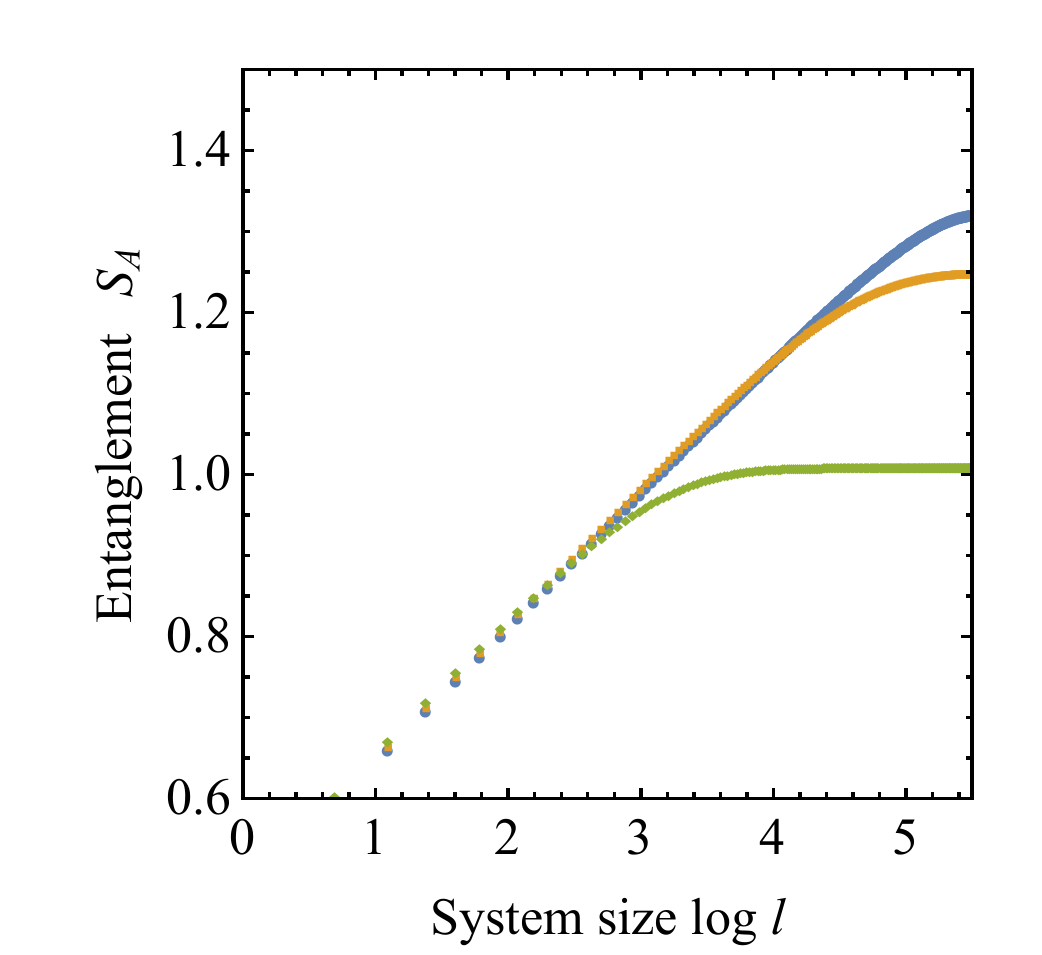}

\caption{Scaling of the $\langle \psi_i \psi_{i+d} \rangle$ correlator with distance $d$ (top) as well as the entanglement entropy $S_A$ with subsystem length $l=|A|$ (bottom) at various cutoffs for the $\lbrace 3,6 \rbrace$, $\lbrace 3,7 \rbrace$ and MERA tiling (left to right). The cutoffs are $r_c=0,50,75$ ($\lbrace 3,6 \rbrace$ case with 870 Majorana boundary sites), $r_c=0.0,0.8,0.9$ ($\lbrace 3,7 \rbrace$ case with 876 sites), and $n_c=0,2,4$ (MERA case with 1024 sites), the data for each cutoff plotted in blue, yellow and green, respectively.
}
\label{FIG_CRIT_SCALING_CUTOFF}
\end{figure}

\subsection{Pentagon code for quantum error correction}
\label{ssec:happy}
First, consider the boundary state of a single pentagon. 
Explicitly, the $+1$ logical state vector of the quantum error correcting code is given by
\begin{equation}
\ket{\overline0} = N \exp\left( \frac{1}{2} \sum_{i,j} A^+_{i,j} \fd_i \fd_j \right) \vacket \text{ ,}
\end{equation}
with a normalization factor $N=\frac{1}{4}$ and the $5 \times 5$ generating matrix
	\begin{equation*}
A^+ = \left(
\begin{array}{ccccc}
\msp 0 &    -1 &\msp 1 &\msp 1 &    -1 \\
\msp 1 &\msp 0 &    -1 &\msp 1 &\msp 1 \\
    -1 &\msp 1 &\msp 0 &    -1 &\msp 1 \\
    -1 &    -1 &\msp 1 &\msp 0 &    -1 \\
\msp 1 &    -1 &    -1 &\msp 1 &\msp 0
\end{array} 
\right) \text{ .}
    \end{equation*}
Correspondingly, the $-1$ logical state vector is given by
\begin{equation}
\ket {\overline 1} = N \exp\left( \frac{1}{2} \sum_{i,j} A^-_{i,j} \fd_i \fd_j \right) \int \text{d}\eta  \exp\left(\eta \sum_i B^-_i \fd_i \right) \widetilde{\vacket} \text{ ,}
\end{equation}
containing an integration over the auxiliary Grassmann variable $\eta$, fulfilling $\eta \fd_i = -\fd_i \eta$. The generating matrix $A^-$ and coupling matrix $B^-$ between $\eta$ and the $\fd_i$ are given by
\begin{equation}
A^- = \left(
\begin{array}{ccccc}
\msp 0   &\msp 0.2 &    -0.6 &\msp 0.6 &    -0.2 \\
    -0.2 &\msp 0   &\msp 0.2 &    -0.6 &\msp 0.6 \\
\msp 0.6 &    -0.2 &\msp 0   &\msp 0.2 &    -0.6 \\
    -0.6 &\msp 0.6 &    -0.2 &\msp 0   &\msp 0.2 \\
\msp 0.2 &    -0.6 &\msp 0.6 &    -0.2 &\msp 0
\end{array} 
\right) \text{ ,} \quad
B^- = \left(
\begin{array}{ccccc}
 1  & 1 &1 & 1 &1
\end{array} 
\right) \text{ .}
\end{equation}
However, we can also write this state in a purely Gaussian form by acting with annihilation operators on the fully occupied state vector $\widetilde{\vacket}= \prod_i \fd_i\vacket $,
\begin{equation}
\ket {\overline1} = -N \exp\left( \frac{1}{2} \sum_{i,j} A^+_{i,j} \fe_i \fe_j \right) \widetilde{\vacket} \text{ .}
\end{equation}
Note that the generating matrix $A^+$ in this form is the same as for the positive-parity state, highlighting the symmetry between the positive- and negative-parity eigenstate. The additional minus sign can be removed by redefining either $N$ or $| \tilde{0} \rangle$.    

\subsection{Higher central charges and critical scaling of entanglement entropies}
\label{ssec:app_higher_c}
By associating a higher bond dimension $\chi=2^n$ with each geometric edge, it is possible to increase the central charge $c$ of the conformal field theory capturing the boundary state. The corresponding $3n \times 3n$ correlation matrix $A$ of each triangle state can be chosen so that correlations separate into $n$ parts. 
An example for $\chi=4$ is given by
\begin{equation}
A = 
\left(
\begin{matrix}
 0 & 0 & a & 0 & b & 0 \\
 0 & 0 & 0 & a & 0 & b \\
-a & 0 & 0 & 0 & c & 0 \\
 0 &-a & 0 & 0 & 0 & c \\
-b & 0 &-c & 0 & 0 & 0 \\
 0 &-b & 0 &-c & \;0\; & \;0\; \\
\end{matrix}
\right)
\text{ ,}
\end{equation}
where $a=b=c$ again corresponds to a rotation-invariant state. The construction of states with higher $\chi$ is visualized in Fig.\ \ref{FIG_HIGH_C}. Note that this separation into $n$ independent ``channels'' can only be sustained in contracted $\lbrace p,q \rbrace$ tilings if $q$ is even; otherwise, self-contractions lead to mixing between different channels. Also shown in Fig.\ \ref{FIG_HIGH_C} is the entanglement entropy scaling of the boundary states of such triangular bulks, yielding a central charge of the equivalent CFT description of $c=n/2 = \log_2 \sqrt{\chi}$. The expected entanglement growth \eqref{EQ_CALABRESE_CARDY} is only reached when the subsystem size $l$ is larger than the size of one geometrical edge, i.e.\ $2n$ Majorana fermions. Indeed, by construction of our Gaussian model, a site of one Majorana fermion always has an entanglement entropy of $1/2 \log 2$, independent of $\chi$. 

\begin{figure}[htb]
\centering
\includegraphics[height=0.2\textheight]{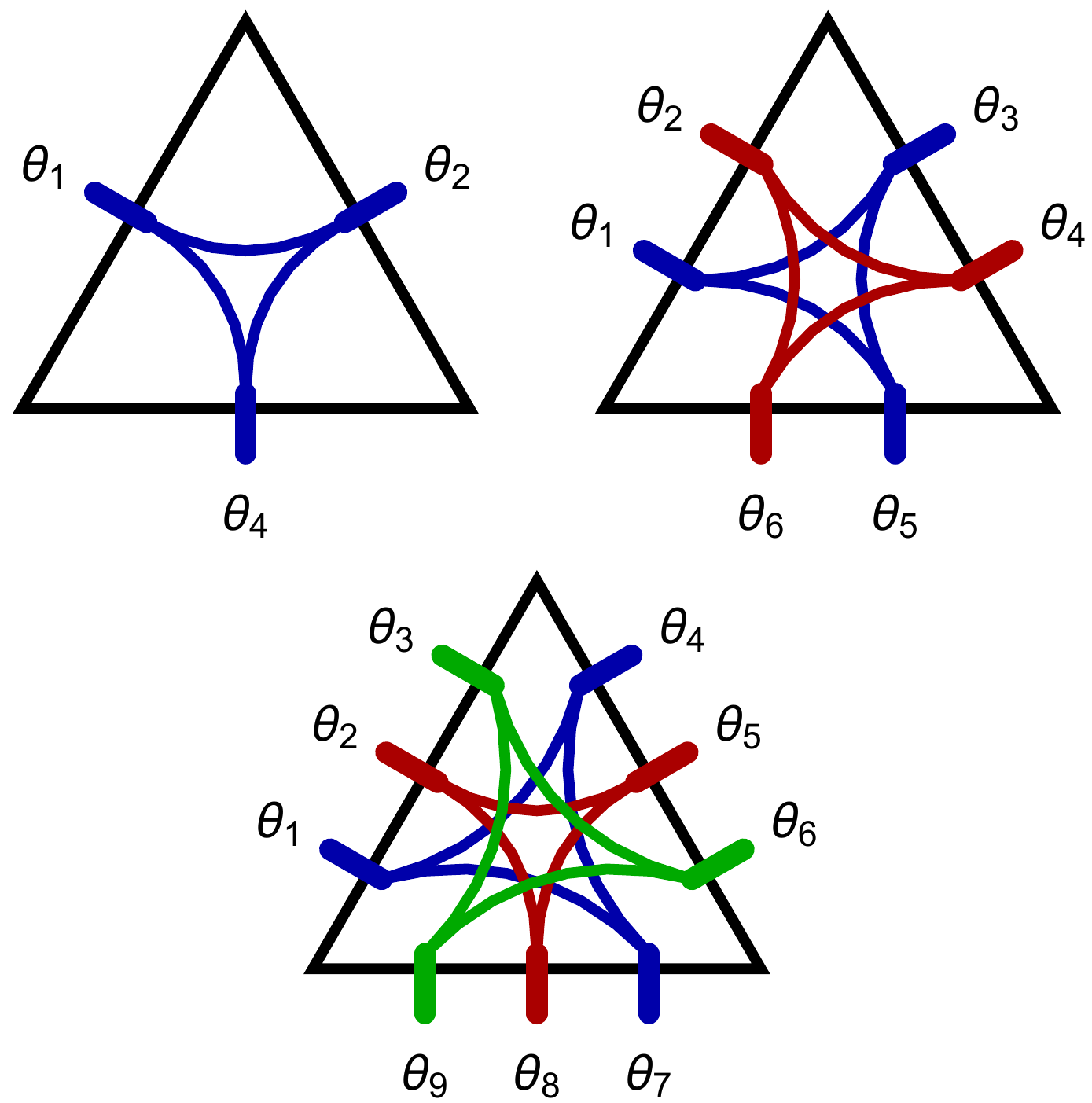}
\hspace{0.08\textheight}
\includegraphics[height=0.2\textheight]{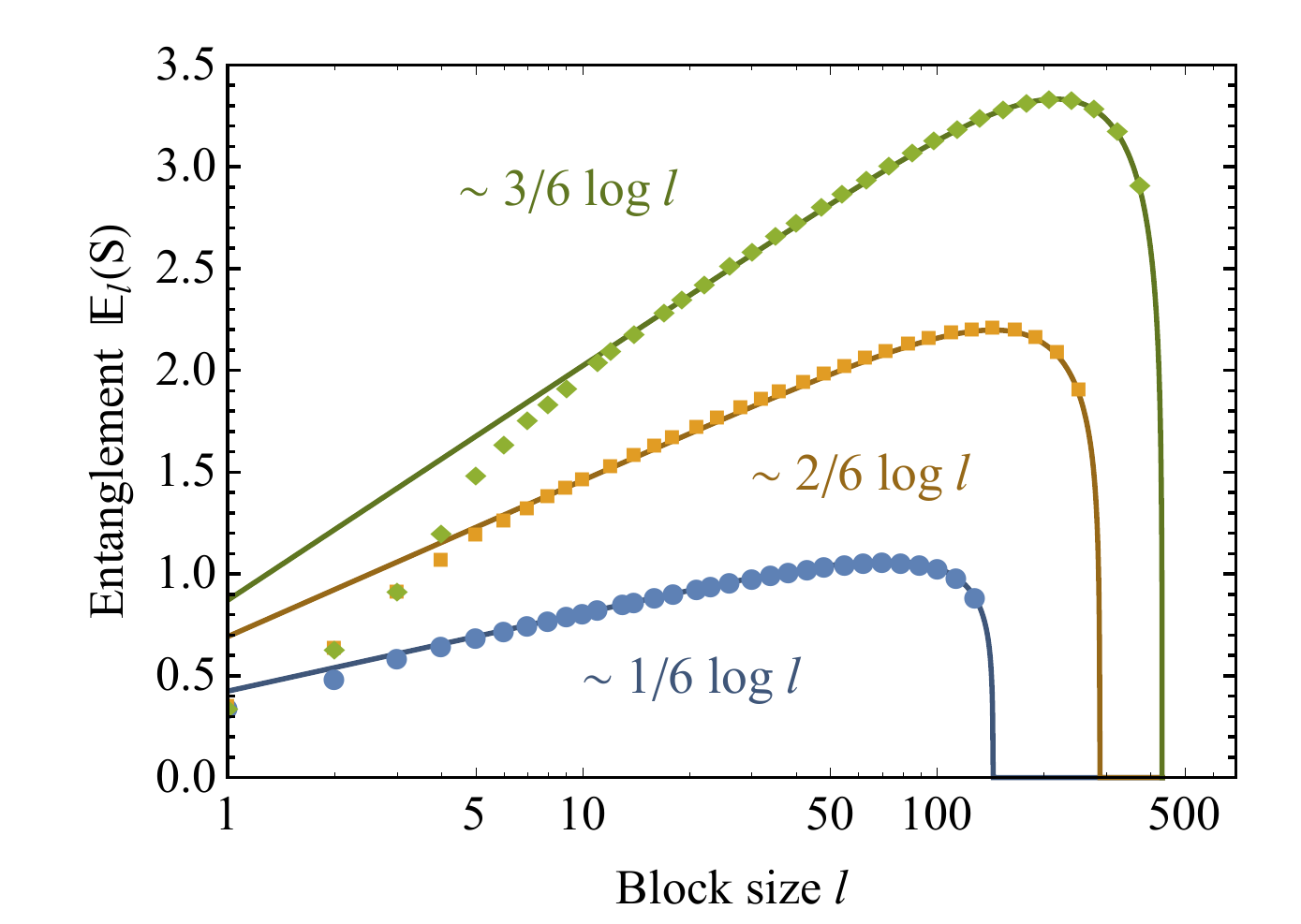}

\caption{\textsc{Left}: Construction of triangle states with bond dimension $\chi=2,4,8$. Colored connections between boundary sites $\{\theta_i\}$ denote non-zero correlations. \textsc{Right}: Mean value of entanglement entropy $\mathbb{E}_\ell(S) =$ $ \sum_{k=1}^L S_{[k,k+\ell]}$ of a boundary subsystem of size $\ell$. Results for $\lbrace 3,8 \rbrace$ tiling with bond dimension $\chi=2,4,8$ (bottom to top; with $144, 288, 432$ Majorana fermions, respectively).
}
\label{FIG_HIGH_C}
\end{figure}

\end{appendix}
\end{widetext}

\bibliography{FermionicPentagonCode} 

\end{document}